\documentclass[letterpaper,11pt]{article}
\usepackage{amsmath, amsthm, amssymb}
\usepackage{graphicx}
\usepackage{caption}
\usepackage{cite}
\usepackage[T1]{fontenc}

\usepackage{todonotes}

\usepackage[noend]{algorithmic}
\usepackage{wrapfig}
\usepackage{xspace}
\usepackage{mathtools}

\usepackage[margin=1in]{geometry} 
\usepackage{hyperref}	% Handles web links

\usepackage{cleveref}
\renewcommand{\cref}{\Cref}

\theoremstyle{definition}
\newtheorem{definition}{Definition}[section]
\theoremstyle{plain}
\newtheorem{theorem}[definition]{Theorem}
\newtheorem{lemma}[definition]{Lemma}

\newtheorem{claim}[definition]{Claim}

\newtheorem{observation}[definition]{Observation}

\newtheorem{conjecture}[definition]{Conjecture}

\newcommand{\Pp}{\mathcal{P}}

\usepackage{tikz}

\newcommand{\w}{\mathfrak{w}}
\newcommand{\elems}
{\ensuremath{\bigcup}}
\DeclareMathOperator{\Oh}{\mathcal{O}}

\newcommand{\Sttt}{\ensuremath{S_{t,t,t}}}
\newcommand{\Pttt}{\ensuremath{\{t,t,t\}}\text{-pyramid}}

\newcommand{\rel}{{\textsf{relevant}}}
\newcommand{\matching}{{\textsf{matching}}}
\newcommand{\esdthm}{{\textsf{esd}}}
\def\epsilon{\varepsilon}

\newcommand{\preimage}[2]{{\overleftarrow{\eta}_{#1}(#2)}}

\renewcommand{\leq}{\leqslant}
\renewcommand{\geq}{\geqslant}

\newcommand{\blind}[1]{#1}
\begin{document}
%\title{Dominated Minimal Separators are Tame While Nearly All Others are Feral}
%title{Dominated Minimal Separators are Tame \\ {(Nearly All Others are Feral)}}
%\title{Families with Dominated Minimal Separators are (Quasi-) Tame \\ {(Nearly All Others are Feral)}}
%\title{Dominating Minimal Separators with Few Vertices }
%\title{Weighted Independent Set on P$_k$-Free Graphs in Quasi-Polynomial Time}

\title{Maximum Weight Independent Set in Graphs with no Long Claws in Quasi-Polynomial Time%
  \blind{\thanks{%
Pe.~G. and Da.~L. were supported by NSF award CCF-2008838.
This work is a part of the project BOBR (Mi.~P. and P.~Rz.) that has received funding from the European Research Council (ERC) 
under the European Union's Horizon 2020 research and innovation programme (grant agreement No.~948057).
T.~M.~was supported by Polish National Science Centre SONATA-17 grant number 2021/43/D/ST6/03312.
Ma.~P.~was supported by Polish National Science Centre SONATA BIS-12 grant number 2022/46/E/ST6/00143.
Ma.~P.~is also part of BARC, supported by the VILLUM Foundation grant 16582.}
}
}

\date{}

\blind{
\author{
Peter Gartland\thanks{University of California, Santa Barbara, USA, \texttt{petergartland@ucsb.edu}.}\and
Daniel Lokshtanov\thanks{University of California, Santa Barbara, USA, \texttt{daniello@ucsb.edu}.} \and
Tom\'{a}\v{s} Masa\v{r}\'{i}k\thanks{Institute of Informatics, University of Warsaw, Poland, \texttt{masarik@mimuw.edu.pl}.} \and
Marcin Pilipczuk\thanks{Institute of Informatics, University of Warsaw, Poland and IT University of Copenhagen, Denmark, \texttt{marcin.pilipczuk@mimuw.edu.pl}.} \and
Micha\l{} Pilipczuk\thanks{Institute of Informatics, University of Warsaw, Poland, \texttt{michal.pilipczuk@mimuw.edu.pl}.} \and
Pawe\l{} Rz\k{a}\.{z}ewski\thanks{Faculty of Mathematics and Information Science, Warsaw University of Technology, Poland and Institute of Informatics, University of Warsaw, Poland, \texttt{pawel.rzazewski@pw.edu.pl}.}
}
}

\begin{titlepage}
\def\thepage{}
\thispagestyle{empty}
\maketitle

\begin{abstract}
We show that the \textsc{Maximum Weight Independent Set} problem (\textsc{MWIS}) can be solved in quasi-polynomial time on $H$-free graphs (graphs excluding a fixed graph $H$ as an induced subgraph) for every $H$ whose every connected component is a path or a subdivided claw (i.e., a tree with at most three leaves). This completes the dichotomy of the complexity of \textsc{MWIS} in $\mathcal{F}$-free graphs for any finite set $\mathcal{F}$ of graphs into NP-hard cases and cases solvable in quasi-polynomial time, and corroborates the conjecture that the cases not known to be NP-hard are actually polynomial-time solvable.

The key graph-theoretic ingredient in our result is as follows. Fix an integer $t \geq 1$. Let $S_{t,t,t}$ be the graph created from three paths on $t$ edges by identifying one endpoint of each path into a single vertex. We show that, given a graph $G$, one can in polynomial time find either an induced $S_{t,t,t}$ in $G$, or a balanced separator consisting of $\Oh(\log |V(G)|)$ vertex neighborhoods in $G$, or an extended strip decomposition of $G$ (a decomposition almost as useful for recursion for \textsc{MWIS} as a partition into connected components) with each particle of weight multiplicatively smaller than the weight of $G$. This is a strengthening of a result of Majewski, Masařík, Novotná, Okrasa, Pilipczuk, Rzążewski, and Sokołowski [ICALP~2022] which provided such an extended strip decomposition only after the deletion of $\Oh(\log |V(G)|)$ vertex neighborhoods. To reach the final result, we employ an involved branching strategy that relies on the structural lemma presented above.
\end{abstract}
\end{titlepage}

    %\mpil[inline]{I am not very happy with the last sentence. 
    %I tried to put here a cool sentence about the algorithmic part, but I failed not to sound incremental
    %over the first Peter-Daniel $P_t$-free algorithm. Any thoughts?}
    %\tm{Alternative:
    %To reach the final result, we employ an involved branching strategy which deals with the second outcome of the structural lemma.
    %}
\thispagestyle{empty}
\tableofcontents \newpage

\pagenumbering{arabic}

\section{Introduction}

The \textsc{Maximum Weight Independent Set} (MWIS) problem takes as input a graph $G$ with vertex weights
$\w : V(G) \to \mathbb{Z}_{\geq 0}$ and asks for a set $X \subseteq V(G)$ of maximum possible weight
that is \emph{independent} (sometimes also called \emph{stable}): no two vertices of $X$ are adjacent.
This classic combinatorial problem plays an important role as a central hard problem 
in several areas of computational complexity: it appears as one of the NP-hard problems
on the celebrated list of Karp~\cite{Karp72},
it is the archetypical W[1]-hard problem in parameterized complexity~\cite{Downey1995},
and is one of the classic problems difficult to approximate~\cite{Hastad99}.

In the light of the hardness of \textsc{MWIS} within multiple paradigms, 
one may ask what assumptions on the input make the problem easier. 
More formally, we can ask for which graph classes $\mathcal{G}$, the assumption that the input graph 
comes from $\mathcal{G}$ allows for faster algorithms for \textsc{MWIS}.
For example, if $\mathcal{G}$ is the class of planar graphs, \textsc{MWIS} remains NP-hard,
but the classic layering approach of Baker~\cite{Baker94} yields a polynomial-time approximation scheme
and simple kernelization arguments give a parameterized algorithm~\cite{Cyganetal}.

This motivates a more methodological study of the complexity of \textsc{MWIS} depending
on the graph class $\mathcal{G}$ the input comes from.
As the space of all graph classes is too wide and admits strange artificial examples, the arguably simplest regularization
assumption is to restrict the attention to hereditary graph classes, i.e., graph classes closed under vertex deletion.
Every hereditary graph class $\mathcal{G}$ can be characterized by \emph{minimal forbidden induced subgraphs}: 
the (possibly infinite) set $\mathcal{F}$ of minimal (under vertex deletion) graphs that are not members of $\mathcal{G}$.
Then, we have $G \in \mathcal{G}$ if and only if no member of $\mathcal{F}$ is an induced subgraph of $G$;
when we want to emphasize the set $\mathcal{F}$, we refer
to the graph class $\mathcal{G}$ as the class of \emph{$\mathcal{F}$-free graphs}
and shorten it to $H$-free graphs if $\mathcal{F} = \{H\}$. 

If a problem turns out to be easier in a class of $\mathcal{F}$-free graphs, in many cases it is a single
forbidden induced subgraph $H \in \mathcal{F}$ that is responsible for tractability, and the problem
at hand is already easier in $H$-free graphs. 
A prime example of this phenomenon are the classes of line graphs and claw-free graphs. 
Recall that a \emph{line graph} of a graph $H$ is a graph $G$ with $V(G) = E(H)$ where two vertices of $G$ are adjacent
if their corresponding edges in $H$ are incident to the same vertex.
Observe that \textsc{MWIS} in a line graph $G$ of a graph $H$ becomes the \textsc{Maximum Weight Matching} problem
in the pre-image graph $H$; a problem solvable in polynomial time by deep combinatorial techniques~\cite{Edmonds}.
It turns out that the tractability of \textsc{MWIS} in line graphs can be explained solely by one of the 
minimal forbidden induced subgraphs for the class of line graphs, namely the \emph{claw} $S_{1,1,1}$.
(For integers $a,b,c \geq 1$, by $S_{a,b,c}$ we denote the tree with exactly three leaves, within distance $a$, $b$, and $c$
from the unique vertex of degree $3$.) 
As proven in 1980, \textsc{MWIS} is polynomial-time solvable already in the class of $S_{1,1,1}$-free graphs~\cite{SBIHI198053,MINTY1980284}, called
also the class of \emph{claw-free graphs} (for recent fast algorithms, see~\cite{clawCubic,clawQuadratic}).

%\pg{why do we focus on the case where $|{\cal F}| = 1$? Alekseev's hardness results works when ${\cal F}$ is finite and does not contain a forest of $\Sttt$'s. So our techniques prove the dichotomy when ${\cal F}$ is finite. (this question can be applied to what is written in the abstract as well)}
%\mpil[inline]{Good comment, I tried to rephrase the paragraph below to capture this.}

Together with the vastness of the space of all hereditary graph classes,
this motivates us to focus on $\mathcal{F}$-free graphs for finite sets $\mathcal{F}$, in particular
on the case $|\mathcal{F}|=1$. 
This turned out to be particularly interesting for \textsc{MWIS}.
As observed by Alekseev~\cite{alekseev1982effect}, for the ``overwhelming majority'' of finite sets ${\cal F}$, \textsc{MWIS} remains NP-hard on $\mathcal{F}$-free graphs.
More precisely Alekseev observed that \textsc{MWIS} remains NP-hard on $\mathcal{F}$-free graphs unless,
for at least one graph in $\mathcal{F}$,
every connected component is a path or an $S_{a,b,c}$ for some integers $a$,$b$,$c$.
Since the original NP-hardness proof of Alekseev~\cite{alekseev1982effect} in 1982, no new finite sets $\mathcal{F}$ have been discovered such that \textsc{MWIS} remains NP-hard on $\mathcal{F}$-free graphs.
We conjecture that this is because all of the remaining cases are actually solvable in polynomial time.
%
%we do not know any NP-hardness proof for \textsc{MWIS} that does not produce long induced paths or induced subgraphs $S_{t,t,t}$ for large $t$. In other words, for every finite $\mathcal{F}$, we know that \textsc{MWIS} is NP-hard in $\mathcal{F}$-free graphs whenever $\mathcal{F}$ does not contain any graph that is a forest whose every component has at most three leaves, but no NP-hardness proof is known in the remaining cases. On the contrary, we expect tractability there.
%
\begin{conjecture}\label{conj:main}
For every $H$ that is a forest whose every component has at most three leaves, 
\textsc{Maximum Weight Independent Set} is polynomial-time solvable when restricted to $H$-free graphs.
\end{conjecture}
To the best of our knowledge, the first place \Cref{conj:main} appeared explicitly is \cite{Lozin17}.
Let us remark that Conjecture~\ref{conj:main}, if true, would yield a dichotomy for the computational complexity of \textsc{MWIS} on $\mathcal{F}$-free graphs for all finite sets $\mathcal{F}$.
Consider any $\mathcal{F}$ such that NP-hardness of \textsc{MWIS} on $\mathcal{F}$-free graphs does not follow from Alekseev's proof.
It follows that the class of $\mathcal{F}$-free graphs is contained in the class of $H$-free graphs for some graph $H$ for which polynomial time solvability of \textsc{MWIS} is conjectured in \cref{conj:main}.

From the positive side, as already mentioned, we know that \textsc{MWIS} is polynomial-time solvable in $S_{1,1,1}$-free graphs
since 1980. Around the same time, it was shown that the class of $P_4$-free graphs (by $P_t$ we denote the path on $t$ vertices)
coincides with the class of \emph{cographs} and has very strong structural properties (in modern terms, has bounded cliquewidth)
thus allowing efficient algorithms for \textsc{MWIS} and many other combinatorial problems. 
Over the years, we have witnessed a few scattered results for some special cases of $H$-free graphs, such as $S_{1,1,2}$-free graphs~\cite{ALEKSEEV20043,DBLP:journals/jda/LozinM08},  $2K_2$-free graphs~\cite{Farber89}, $tK_2$-free graphs~\cite{Farber93}, $\ell{}P_3$-free graphs~\cite{Lozin17}, $\ell{}S_{1,1,1}$-free graphs~\cite{BM18-lclaw}, $tK_2+P5$-free or $tK_2+S_{1,1,2}$-free graphs~\cite{Mosca22}, as well as progress limited to various subclasses (see~\cite{DBLP:journals/endm/LeBS15,BM18,p67,p65,p73,p72,DBLP:journals/gc/LozinMP14,DBLP:journals/jda/LozinMR15,DBLP:journals/tcs/HarutyunyanLLM20,BM18,LozinM05,Mosca99,Mosca09,Mosca13,Mosca21} for older and newer results of this kind).

The research in the area got significant momentum in the last decade.  The progress can be partitioned into two 
main threads. The first one focuses on the framework of \emph{potential maximal cliques}, introduced by Bouchitt\'{e} and Todinca~\cite{BouchitteT01}, and focuses on providing polynomial-time algorithms for $P_t$-free graphs for small values of $t$. 
A landmark result here is due to Lokshtanov, Vatshelle, and Villanger~\cite{LokshtanovVV14} who were the first to show the usability
of the framework in the context of $P_t$-free graphs by providing a polynomial-time algorithm for \textsc{MWIS} 
in $P_5$-free graphs. 
This has been later extended to $P_6$-free graphs~\cite{GrzesikKPP22} and related graph classes~\cite{AbrishamiCPRS21}. 
A notable property of this framework is that in most cases it not only provides algorithms for \textsc{MWIS},
but for a wide range of problems asking for large induced subgraph of small treewidth, for example \textsc{Feedback Vertex~Set}. 

The second thread attempts at treating $P_t$-free or $S_{t,t,t}$-free graphs in full generality, but 
relaxing the requirements on either the running time (by providing subexponential or quasi-polynomial-time algorithms) 
or the accuracy (by providing approximation algorithms, such as approximation schemes). 
Here, the starting point is the theorem of Gy\'{a}rf\'{a}s~\cite{gyarfas1,gyarfas2} (see also~\cite{BacsoLMPTL19}).
\begin{theorem}\label{thm:gyarfas}
Every vertex-weighted graph $G$ contains an induced path $Q$ such that every connected component of $G-N[V(Q)]$
has weight at most half of the weight of $G$.
\end{theorem}
As an induced path in a $P_t$-free graph has less than $t$ vertices, a $P_t$-free graph admits a balanced separator
(in the sense of Theorem~\ref{thm:gyarfas}) consisting of neighborhood of at most $t-1$ vertices.
In other words, $P_t$-free graphs admit a balanced separator dominated by $t-1$ vertices.
Chudnovsky, Pilipczuk, Pilipczuk, and Thomassé \cite{DBLP:conf/soda/ChudnovskyPPT20} observed that this easily gives a quasi-polynomial-time approximation scheme (QPTAS) for
\textsc{MWIS} in $P_t$-free graphs, and they designed an elaborate argument involving the celebrated three-in-a-tree
theorem of Chudnovsky and Seymour~\cite{DBLP:journals/combinatorica/ChudnovskyS10} to extend the result to the $S_{t,t,t}$-free case
and $H$-free case where $H$ is a forest of trees with at most three leaves each.
Abrishami, Chudnovsky, Dibek, and Rzążewski\cite{ACDR21} used also the three-in-a-tree theorem to obtain a polynomial-time algorithm for \textsc{MWIS}
for $S_{t,t,t}$-free graphs of bounded degree.
Gartland and Lokshtanov showed how to use the theorem of Gy\'{a}rf\'{a}s to design exact
quasi-polynomial-time algorithm for \textsc{MWIS} in $P_t$-free graphs~\cite{GartlandL20}, for every fixed $t$.
This algorithm was later simplified by Pilipczuk, Pilipczuk, and Rz\k{a}\.{z}ewski~\cite{PilipczukPR21}
and the union of the authors of these two papers showed that the approach works for a much wider class of problems
and a slightly wider graph class~\cite{GartlandLPPR21}.
Last year, Majewski, Masařík, Novotná, Okrasa, Pilipczuk, Rzążewski, and Sokołowski \cite{ICALP-qptas} gave a cleaner argument for an existence of a QPTAS for \textsc{MWIS}
in $S_{t,t,t}$-free~graphs.

This work provides the pinnacle of the second thread by showing that \textsc{MWIS}
is quasi-polynomial-time solvable in all cases treated by Conjecture~\ref{conj:main}.
\begin{theorem}\label{thm:main}
For every $H$ that is a forest whose every component has at most three leaves, 
there is an algorithm for \textsc{Maximum Weight Independent Set} in $H$-free graphs
running in time $n^{\Oh_H(\log^{19} n)}$.
\end{theorem}
Here $\Oh_H$ denotes constants depending on $|H|$ being repressed.
Theorem~\ref{thm:main} provides strong evidence in favor of Conjecture~\ref{conj:main},
as it refutes the existence of an NP-hardness proof for \textsc{MWIS} for $H$-free graphs as in Conjecture~\ref{conj:main}, unless all problems in NP can be solved in quasi-polynomial time.

\subsection{Our techniques}

As discussed in~\cite{GartlandL20} (in particular Theorem 2), to show Theorem~\ref{thm:main} it suffices to focus on the case $H=S_{t,t,t}$
for a fixed integer $t \geq 1$. Together with a simple self-reducibility argument, it is enough to prove the following.
\begin{theorem}\label{thm:long claw free qp time}
For every integer $t \geq 1$, 
the maximum possible weight of an independent set in a given $n$-vertex $\Sttt$-free graph can be found in $n^{\Oh_t(\log^{16}(n))}$ time.
\end{theorem}
Here $\Oh_t$ denotes constants depending on $t$ being repressed.

\subsubsection{The key structural result}
While Theorem~\ref{thm:gyarfas} provides a balanced separator consisting of a few neighborhoods in
a $P_t$-free graph, it does not seem to be directly usable for $S_{t,t,t}$-free graphs. 
The example of $G$ being a line graph of a clique (which is $S_{1,1,1}$-free) shows
that we cannot hope for merely a balanced separator consisting of a few neighborhoods in $S_{1,1,1}$-free graphs.

However, if $G$ is a line graph, \textsc{MWIS} is solvable in polynomial-time by a very different reason than Theorem~\ref{thm:gyarfas}:
because it corresponds to a matching problem in the preimage graph. 
Luckily, there is a known formalism capturing decompositions of a graph that are ``like a line graph'': extended strip decompositions.

For a graph $G$, a \emph{strip decomposition} consists of a graph $H$ (called the \emph{host}) and a function $\eta$ that assigns to every 
edge $e \in E(H)$ a subset $\eta(e) \subseteq V(G)$ such that $\{\eta(e)~|~e \in E(H)\}$ is a partition of $V(G)$
and a subset $\eta(e,x) \subseteq \eta(e)$ for every endpoint $x \in e$ such that the following holds:
for every $v_1,v_2 \in V(G)$ with $v_1 \in \eta(e_1)$, $v_2 \in \eta(e_2)$ and $e_1 \neq e_2$ we have
$v_1v_2 \in E(G)$ if and only if there is a common endpoint $x \in e_1 \cap e_2$ with $v_1 \in \eta(e_1,x)$ and $v_2 \in \eta(e_2,x)$.
Note that if $G$ is the line graph of $H$, then $G$ has a strip decomposition with host $H$
and $\eta(e,x) = \eta(e,y) = \{e\}$ for every $xy=e \in E(H) = V(G)$.
The crucial observation is that if one provides a strip decomposition $(H,\eta)$ of a graph $G$
together with, for every $xy \in E(H)$, the maximum possible weight of an independent set
in $G[\eta(xy)]$, $G[\eta(xy) \setminus \eta(xy,x)]$, $G[\eta(xy)\setminus \eta(xy,y)]$, and $G[\eta(xy) \setminus (\eta(xy,x) \cup \eta(xy,y)]$
(these graphs are henceforth called  \emph{particles}), then we can reduce computing the maximum weight
of an independent set in $G$ to the maximum weight matching problem in the graph $H$ with some gadgets~attached~\cite{DBLP:conf/soda/ChudnovskyPPT20}.

An \emph{extended strip decomposition} also allows vertex sets $\eta(x)$ for $x \in V(H)$
and triangle sets $\eta(xyz)$ for triangles $xyz$ in $H$; a precise definition can be found in preliminaries, but is
irrelevant for this overview. We refer to Figure~\ref{fig:esd} for an example. Importantly, the notion of a particle generalizes
and the property that one can solve \textsc{MWIS} in $G$ knowing the answers to \textsc{MWIS} in the particles 
is still true. 
Extended strip decompositions come from the celebrated solution to the \textit{three-in-a-tree} problem by Chudnovsky and Seymour.
The task is to determine if a graph contains an induced subgraph which is a tree connecting three given vertices.
The following theorem says that The three-in-a-tree problem can be solved in polynomial time:
\begin{theorem}[\hskip -0.01em{\cite[Section 6]{DBLP:journals/combinatorica/ChudnovskyS10}}, simplified version]\label{thm:three-in-a-tree:intro}
Let $G$ be an $n$-vertex  graph and $Z$ be a subset of vertices with $|Z| \geq 2$.
There is an algorithm that runs in time $\Oh(n^5)$ and returns one of the~following:
\begin{itemize}
\item an induced subtree of $G$ containing at least three elements of $Z$,
\item an extended strip decomposition $(H,\eta)$ of $G$ where for every $z \in Z$ there exists
a distinct degree-1 vertex $x_z \in V(H)$ with the unique incident edge $e_z \in E(H)$ and $\eta(e_z,x_z) = \{z\}$.
\end{itemize}
\end{theorem}
In a sense, an extended strip decomposition as in Theorem~\ref{thm:three-in-a-tree:intro} is a certificate
that no three vertices of $Z$ can be connected by an induced tree in $G$.

\cite{DBLP:conf/soda/ChudnovskyPPT20} combined Theorem~\ref{thm:gyarfas} 
with Theorem~\ref{thm:three-in-a-tree:intro} in a convoluted way to show a QPTAS for \textsc{MWIS} in $S_{t,t,t}$-free graphs;
Thereom~\ref{thm:three-in-a-tree:intro} is used here to construct an induced $S_{t,t,t}$ in the argumentation.
\cite{ICALP-qptas} provided a simpler argument for the existence of a QPTAS:
they derived from Theorem~\ref{thm:three-in-a-tree:intro} the following structural result.

\begin{theorem}[{\cite[Theorem 2]{ICALP-qptas}} in a weighted setting]\label{thm:ICALP:weight}
For every fixed integer $t$, there exists a polynomial-time algorithm that,
given an $n$-vertex graph $G$ with nonnegative vertex weights, 
either:
\begin{itemize}
\item outputs an induced copy of $S_{t,t,t}$ in $G$, or
\item outputs  a set $\mathcal{P}$ consisting of at most $11 \log n + 6$ induced paths in $G$,
  each of length at most $t+1$, and a rigid extended strip decomposition of
  $G - N[\elems{\Pp}]$ with every particle of weight at most half of the total weight of $V(G)$.
\end{itemize}
\end{theorem}
(Here, rigid means that the extended strip decomposition does not have some unnecessary empty sets;
in a rigid decomposition the size of $H$ is bounded linearly in the size of $G$.
 The formal statement of Theorem~\ref{thm:ICALP:weight} in~\cite{ICALP-qptas} is only for uniform weights in $G$, but as observed
 in the conclusions of~\cite{ICALP-qptas}, the proof works for arbitrary vertex weights.)

\cite{ICALP-qptas} showed that Theorem~\ref{thm:ICALP:weight} easily gives a QPTAS for \textsc{MWIS}
in $S_{t,t,t}$-free graphs, along the same lines as how \cite{DBLP:conf/soda/ChudnovskyPPT20} showed that
Theorem~\ref{thm:gyarfas} easily gives a QPTAS for \textsc{MWIS} in $P_t$-free graphs.

However, it seems that the outcome of Theorem~\ref{thm:ICALP:weight} is not very useful if one aims for an exact algorithm
faster than a subexponential one.
Our main graph-theoretic contribution is a strengthening of Theorem~\ref{thm:ICALP:weight} to the following.
\begin{theorem}\label{thm:esd:intro}
For every fixed integer $t$, there exists an integer $c_t$ and a polynomial-time algorithm that, given an $n$-vertex graph $G$
and a weight function $\w : V(G) \to [0,+\infty)$, 
returns one of the following outcomes:
\begin{enumerate}
    \item an induced copy of $S_{t,t,t}$ in $G$;
    \item a subset $X \subseteq V(G)$ of size at most $c_t \cdot \log(n)$ such that
    every component of $G-N[X]$ has weight at most $0.99\w(G)$;
    \item a rigid extended strip decomposition of $G$ where no particle is of weight larger than $0.5\w(G)$.
\end{enumerate}
\end{theorem}
That is, we either provide an extended strip decomposition of the \emph{whole} graph 
(not only after deleting a neighborhood of a small number of vertices as in Theorem~\ref{thm:ICALP:weight})
or a small number of vertices such that deletion of their neighborhood breaks the graph into multiplicatively smaller (in terms of weight) components.

The proof of Theorem~\ref{thm:esd:intro} is provided in Section~\ref{sec:esd}. 
Let us briefly sketch it. 
  We start by applying Theorem~\ref{thm:ICALP:weight} to $G$; we are either already done 
or we have a set $Z \coloneqq  \bigcup_{P \in \mathcal{P}} V(P)$ of size $\Oh(\log n)$ 
and an extended strip decomposition $(H,\eta)$ of $G-N[Z]$ with small particles.
Our goal is now to add the vertices of $N[Z$] one by one back to $(H,\eta)$,
possibly exhibiting one of the other outcomes of Theorem~\ref{thm:esd:intro} along the way.
That is, we want to prove the following lemma:
\begin{lemma}\label{lem:esd:step:intro}
For every fixed integer $t$ there exists an integer $c_t$ and a polynomial-time algorithm that, given an $n$-vertex graph $G$,
a weight function $\w : V(G) \to [0,+\infty)$, a real $\tau \geq \w(G)$, a vertex $v \in V(G)$,
and a rigid extended strip decomposition $(H,\eta)$ of $G-v$ with every particle of weight at most $0.5\tau$, 
returns one of the following:
\begin{enumerate}
    \item an induced copy of $S_{t,t,t}$ in $G$;
    \item a set $Z \subseteq V(G)$ of size at most $c_t$ such that every connected component
    of $G-N[Z]$ has weight at most $0.99 \tau$;
    \item a rigid extended strip decomposition of $G$ where no particle is of weight larger than $0.5\tau$.
\end{enumerate}
\end{lemma}
A simple yet important observation for Lemma~\ref{lem:esd:step:intro} is that for
$x \in V(H)$ of degree at least two, the set $\bigcup_{y \in N_H(x)} \eta(xy,x)$ can be dominated by at most two vertices,
as the sets $\eta(xy,x)$ for $y \in N_H(x)$ are complete to each other. 
Consequently, if $(A,B)$ is a separation in $H$ of small order, 
then the part of $G$ that is placed by $\eta$ in $H[A]$ and the part of $G$ that is placed by $\eta$ in $H[B]$
can be separated by deleting at most $2|A \cap B|$ vertex neighborhoods in $G$.
Hence, if there is a separation $(A,B)$ in $H$ of constant order where both sides of this separation have substantial weight
(at least $0.01\tau$), we can provide the second outcome of Lemma~\ref{lem:esd:step:intro}.

As $N[v]$ is just one neighborhood,
the same observation holds if, instead of looking at $(H,\eta)$, we look at the inherited extended strip decomposition
$(H',\eta')$ of $G-N[v]$. Here, $(H',\eta')$ is obtained from $(H,\eta)$ by first deleting vertices of $N(v)$
from sets $\eta(\cdot)$ and then performing a cleanup operation that trims unnecessary empty sets
and ensures that for every $xy \in E(H')$ there is a path in $G[\eta'(xy)]$ between $\eta'(xy,x)$ and $\eta'(xy,y)$.
Hence, we can take all separations $(A,B)$ in $H'$ of order bounded by a large constant (depending on $t$)
and orient them from the side that contains less than $0.01\tau$ weight to the side containing almost all the weight of $G$.
This orientation defines a tangle in $H'$.
By classic results from the theory of graph minors, this tangle implies the existence of a large wall $W$ in $H'$ which
is always mostly on the ``large weight'' side of any separation $(A,B)$ of constant order.
The cleaning operation ensures that the wall $W$ is also present in $(H,\eta)$. 

An important observation now is that, because $(H',\eta')$ is cleaned as described below, any family of vertex-disjoint paths
in $H'$ projects down to a family of induced, vertex-disjoint, and anti-adjacent paths in $G$ of roughly the same length (or longer):
for a path $P$ in $H$, just follow paths from $\eta(xy,x)$ to $\eta(xy,y)$ in $G[\eta(xy)]$ for consecutive edges $xy$ on $P$.
Furthermore, a wall $W$ is an excellent and robust source of long vertex-disjoint paths.

This allows us to prove that if the neighbors of $v$ are well-connected to the wall $W$ in $(H,\eta)$ --- either they are spread around the wall itself, or one can connect them to $W$ via three vertex-disjoint paths in $H$ --- then $G$ contains an induced $S_{t,t,t}$.
Otherwise, we show that there is a separation $(A,B)$ in $H$ with the neighbors of $v$ essentially all contained
in the sets of $H[A]$, while $W$ lies on the $B$-side of the separation. 
(Here, a large number of technical details are hidden in the phrase ``essentially contained''.)
We construct a graph $G_A$ being the subgraph of $G$ induced by the vertices contained in the $\eta$ sets of $H[A]$,
augmented with a set $Z$ of artificial vertices attached to $\bigcup_{y \in N_H(x) \cap A} \eta(xy,x)$ for $x \in A \cap B$; vertices of $Z$
signify possible ``escape paths'' to the wall $W$. These ``escape paths'' allow us to show that any induced tree in $G_A$
that contains at least three vertices of $Z$ lifts to an induced $S_{t,t,t}$ in $G$, see \cref{fig:sttt}.
Hence, the algorithm of Theorem~\ref{thm:three-in-a-tree:intro} applied to $G_A$ and $Z$
can be used to rebuild $H[A]$ to accommodate $v$ there as well, or to expose an induced $S_{t,t,t}$. This finishes the sketch of the proof of Lemma~\ref{lem:esd:step:intro}
and of Theorem~\ref{thm:esd:intro}.

\begin{figure}
    \centering
    \includegraphics[scale=0.8]{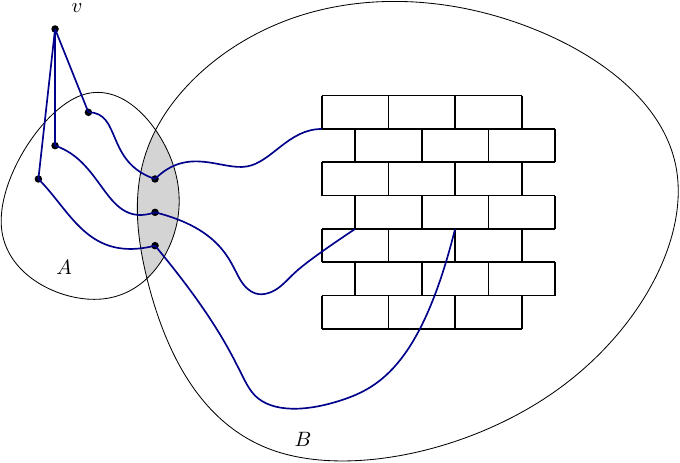}
    \caption{Extending a subdivided claw in $G_A$ to an $S_{t,t,t}$ using the large wall $W$ in $B$.}
    \label{fig:sttt}
\end{figure}

We would like to highlight a significant difference between previous works~\cite{ACDR21,DBLP:conf/soda/ChudnovskyPPT20,ICALP-qptas}
and our use of the three-in-a-tree theorem to exhibit an $S_{t,t,t}$ in a graph or obtain an extended strip decomposition.
All aforementioned previous works essentially picked three anti-adjacent paths $P_1$, $P_2$, $P_3$ of length $t$ each,
with endpoints say $x_i$ and $y_i$ for $i=1,2,3$,
removed their neighborhood except for the neighbors of $y_i$s, and called three-in-a-tree for the set $Z = \{x_1,x_2,x_3\}$;
note that any induced tree in the obtained graph that contains $Z$ contains also an induced $S_{t,t,t}$. 
This method inherently produced extended strip decompositions not for the entire graph, but only for after removal of
a number of neighborhoods. Furthermore, it used the assumption of being $S_{t,t,t}$-free only in a very local sense:
there is no $S_{t,t,t}$ with paths extendable to the given three vertices of $Z$.
In this work, in contrast, we apply the three-in-a-tree theorem to a potentially much bigger set $Z$, and use a subdivided wall
in the host graph of the extended strip decomposition to extend any induced tree found to an induced $S_{t,t,t}$. 
In this way, we used the assumption of being $S_{t,t,t}$-free in a more global way than just merely asking
for three particular leaves. 

\subsubsection{Branching}\label{subsec:branching}

We now proceed with a sketch of our recursive branching algorithm.
On a very high level, it is based on techniques used in the quasi-polynomial time algorithm for independent set on $P_k$-free graphs found in \cite{GartlandL20}, though multiple new ideas are required to make the reasoning work in the setting of $S_{t,t,t}$-free graphs, making both the algorithm and its running time analysis quite a bit more technical. We will soon sketch the algorithm found in \cite{GartlandL20} and describe how to extend it to $\Sttt$-free graphs, but first we must address a major barrier. The fact that $P_k$-free graphs have balanced separators dominated by $k$ vertices, as discussed after Theorem~\ref{thm:gyarfas}, is a crucial fact used in the algorithm of \cite{GartlandL20}. But, as mentioned previously, $\Sttt$-free graphs have no such property (take for instance the line graph of a clique). This is where Theorem~\ref{thm:esd:intro} comes to the rescue.  

When applying Theorem~\ref{thm:esd:intro} to $G$ (the input graph of the current call of the algorithm), since we assume that $G$ is $\Sttt$-free, we are guaranteed that outcome $(1)$ will not occur. If outcome $(3)$ occurs then we get an extended strip decomposition $(H, \eta)$ and, as previously mentioned, we can reduce finding a maximum independent set of $G$ to finding a maximum independent set in each particle of~$(H, \eta)$. That is great news, as each particle has at most half of the weight of $G$, and we can easily employ a divide-and-conquer strategy by recursively calling the algorithm on each particle of~$(H, \eta)$. So, since outcome $(1)$ never happens and outcome $(3)$ gives us an easy algorithm, we can always assume that outcome $(2)$ happens, that is, that Theorem~\ref{thm:esd:intro} gives us a balanced separator of $G$ that is dominated by $\Oh(\log n)$ vertices, and now we can try to extend the techniques found in \cite{GartlandL20} to work for $\Sttt$-free graphs. Therefore, for the rest of this subsection we will focus on sketching an algorithm for independent set on an $\Sttt$-free graph $G$ such that all induced subgraphs of $G$ have a balanced separator dominated by some constant number of vertices (the stronger assumption of a constant number of vertices versus $\log n$ vertices does not change the algorithm very much and simplifies the discussion).

Before sketching the algorithm let us give a few short definitions around balanced separators for an $\Sttt$-free graph $G$ (see Section~\ref{sec:prelims} for formal definitions of balanced separators). For $n'>0$, we say that a set $S \subseteq V(G)$ is a {\em{$n'$-balanced separator}}\ for $G$ if no component of $G-S$ has more than $n'$ vertices. If $A \subseteq V(G)$ and no component of $G-S$ contains over $n'$ vertices of $A$, we say that $S$ is a $n'$-balanced separator for $(G, A)$. The outcome $(2)$ of Theorem~\ref{thm:esd:intro} gives us a $0.99|A|$-balanced separator for $(G, A)$ dominated by $\Oh(\log n)$ vertices (again here for simplicity we will assume that these balanced separators are in fact dominated by a constant number of vertices). 
%
%However, by repeatedly picking a $0.99|A \cap C|$-balanced separator for $G[C]$, $A \cap C$
However, by picking a constant number of balanced separators as provided by Theorem~\ref{thm:esd:intro} and taking their union, we can obtain $c|A|$-balanced separators for $(G, A)$ dominated by a constant number of vertices for any fixed $c \in (0,1)$, so we will assume we have access to such strengthened balanced separators.

%a relatively simple argument allows us to boost the strength of these balanced separators to be a $c|A|$-balanced separator for $(G, A)$ dominated by a constant number of vertices for any fixed $c \in (0,1)$, so we will assume we have access to such strengthened balanced separators.

%However, a relatively simple argument allows us to boost the strength of these balanced separators to be a $c|A|$-balanced separator for $(G, A)$ dominated by a constant number of vertices for any fixed $c \in (0,1)$, so we will assume we have access to such strengthened balanced separators.

\paragraph{Summary of the Quasi-Polynomial Time Algorithm for MWIS on $P_k$-free Graphs.}
The starting point for our algorithm is the algorithm for MWIS on $P_k$-free graphs by Gartland and Lokshtanov~\cite{GartlandL20}, who in turn build on an algorithm of Bacsó, Lokshtanov, Marx, Pilipczuk, Tuza, and van Leeuwen \cite{BacsoLMPTL19}. We therefore give a brief summary of these algorithms.
%We begin by sketching the algorithm for independent set on $P_k$-free graphs found in \cite{GartlandL20}. 
%The algorithm of 

We first consider the simple $n^{\Oh(k\sqrt{n}\log n)}$ time algorithm of \cite{BacsoLMPTL19} for MWIS on $P_k$-free graphs. We begin with an $n$-vertex $P_k$-free graph $G$ and branch on all vertices of degree at least $\sqrt{n}$: we either exclude such a vertex from the solution (and thus remove it from the graph), or we include it (and then remove its whole neighborhood from the graph). After this we may assume that the graph in our current instance (we will still refer to this graph as $G$ although some vertices of the original graph $G$ have been removed) now has maximum degree at most $\sqrt{n}$. We solve this instance by finding an $n/2$-balanced separator, $S$, for $G$ that is dominated by at most $k$ vertices. Since $G$ has maximum degree $\sqrt{n}$ and $S$ is dominated by at most $k$ vertices, $S$ can have size at most $k\sqrt{n}$. We then branch on all $k\sqrt{n}$ vertices of $S$ simultaneously, which then breaks up the graph into small connected components and we recurse on each component. A simple analysis shows that this runs in $n^{\Oh(k\sqrt{n}\log n)}$ time.

Now, let us try to improve it to an algorithm that runs in time $n^{\Oh(kn^{1/3}\log n)}$. We first state a modified form of a lemma that appears in \cite{GartlandL20}.

\begin{lemma}\label{lem:cant pack bs:intro}
Let $G$ be an $n$-vertex $P_k$-free graph and ${\cal F}$ a multi-set of subsets of $V(G)$ such that for every $S \in {\cal F}$ no component of $G-S$ has more than $n/2$ vertices. Assume that no vertex belongs to more than $c$ sets of ${\cal F}$ counting multiplicity. Then provided $|{\cal F}| \geq 3ck$, no component of $G$ contains more than $3n/4$ vertices.
\end{lemma}

\begin{proof}[Sketch of proof.]
Let $S \in {\cal F}$ and assume for a contradiction that the largest component of~$G$, call it~$C$, has more than $3n/4$ vertices. Select vertices $a,b$ uniformly at random from $C$. As $|C| > 3n/4$ the probability that $a$ and $b$ belong to different components of $G-S$ is at least $1/3$. If we let $X_S$ be the random variable that is $1$ if $a$ and $b$ are in different components of $G-S$ and 0 otherwise, then  $\mathbb{E}[X_S] \geq \frac{1}{3}$. By the linearity of expectation, we have $\mathbb{E}[\sum_{S \in {\cal F}} X_S]$ $\geq$ $\frac{1}{3} \cdot 3ck \geq ck$. It follows that there exists vertices $a,b \in S$ such that for at least $ck$ sets, $S'$, in ${\cal F}$ (counting multiplicity) $a$ and $b$ are in different components of $G-S'$. Let ${\cal F}'$ be the subset of ${\cal F}$ that contains these sets $S'$. It follows that for any induced path $P$ with $a$ and $b$ as its endpoints, if $S' \in {\cal F}'$ then $V(P) \cap S' \neq \emptyset$. Since ${\cal F}'$ has at least $ck$ sets and no vertex of $P$ belongs to more than $c$ sets in ${\cal F}'$, $P$ must have at least $k$ vertices, contradicting the assumption that $G$ is $P_k$-free. 
\end{proof}

For the $n^{\Oh(kn^{1/3}\log n)}$ algorithm, we again begin by branching on vertices of high degree, but this time we set the threshold to vertices with degree at least $n^{2/3}$. After this we may assume the graph in our current instance, call it $G^1$, has maximum degree $n^{2/3}$. We then find a balanced separator, $S^1$, for $G^1$ that is dominated by $k$ vertices, hence $S^1$ has at most $kn^{2/3}$ vertices. We then branch on all vertices with at least $n^{1/3}$ neighbors in $S^1$. Now we assume the graph considered in our current instance, call it $G^2$, has maximum degree $n^{2/3}$ and a balanced separator $S^1$ such that no vertex of $G^2$ has more than $n^{1/3}$ neighbors in $S^1$. We then find a balanced separator, $S^2$, for $G^2$ that is dominated by $k$ vertices, hence $S^2$ has at most $kn^{2/3}$ vertices and $S^1 \cap S^2$ has size at most $kn^{1/3}$. We then branch on all vertices with at least $n^{1/3}$ vertices in $S^2$ and we branch on all vertices that belong to $S^1 \cap S^2$, so $S^1$ and $S^2$ ``become disjoint''. We repeat this $3k$ times until we are in an instance where we have a graph $G^{3k}$ and $3k$ pairwise disjoint balanced separators $S^1, \ldots S^{3k}$. By Lemma~\ref{lem:cant pack bs:intro}, $G^{3k}$ has no component with over $3n/4$ vertices and we then recurse on each component. A somewhat more involved, but still fairly simple analysis shows that this runs in $n^{\Oh(kn^{1/3}\log n)}$~time.

In the $n^{\Oh(kn^{1/3}\log n)}$-time algorithm, we branched on vertices that: had over $n^{2/3}$ neighbors, or had $n^{1/3}$ neighbors in any of the balanced separators we picked up, or belonged to two of the balanced separators we picked up. In order to modify this algorithm to run in quasi-polynomial time all that must be done is change the branching threshold. In particular, the algorithm collects balanced separators (each dominated by at most $k$ vertices) and will branch on any vertex that has over $n/2^i$ neighbors that belong to $i$ or more of the collected balanced separators (the algorithm no longer branches on vertices that only have high degree). Any vertex that belongs to $\log n$ of the collected balanced separators will then be branched on, so no vertex will ever belong to more than $\log n$ of the collected balanced separators. So, by Lemma~\ref{lem:cant pack bs:intro}, after collecting $3k\log n$ of these balanced separators, the graph will not have any large component. A runtime analysis of this algorithm shows that it runs in quasi-polynomial time. Note that in all three algorithms discussed here (the $n^{\Oh(kn^{1/2}\log n)}$-time, $n^{\Oh(kn^{1/3}\log n)}$-time, and quasi-polynomial-time algorithm) it is crucial for efficient runtime that the balanced separators we use are dominated by few vertices (they were dominated by $k$ vertices here, but being dominated by $\mathrm{polylog}(n)$ vertices would still be sufficent).

\paragraph{Back to $S_{t,t,t}$-free Graphs.}
Recall that we wish to get a quasi-polynomial time algorithm for MWIS on $S_{t,t,t}$-free graphs for the case where every induced subgraph of the input graph $G$ has a set $S$ of at most $c_t$ vertices such that $N[S]$ is a $n/2$-balanced separator. Up to the bound on the set dominating the separator, this is precisely the case when we keep getting outcome (2) whenever we apply Theorem~\ref{thm:esd:intro}.

We want to mimic the algorithm for $P_k$-free graphs. This algorithm used that the input graph is $P_k$-free in precisely two places. The first is to keep getting constant size sets $S$ such that $N[S]$ is an $n/2$-balanced separator. This is easily adapted to our new setting because we keep getting such sets whenever we apply Theorem~\ref{thm:esd:intro}.

The second place where $P_k$-freeness is used is in Lemma~\ref{lem:cant pack bs:intro}, which  states that a $P_k$-free graph cannot have a set of $3k \log n$ balanced separators such that no vertex of $G$ appears in at most $O(\log n)$ of them.
If we could strengthen the statement of Lemma~\ref{lem:cant pack bs:intro} to $S_{t,t,t}$-free graphs we would be done! Unfortunately such a strengthening is false, indeed a path is a counterexample (each vertex close to the middle of the path is a balanced separator).

Nevertheless, a subtle weakening of Lemma~\ref{lem:cant pack bs:intro} does turn out to be true. In particular, in $S_{t,t,t}$-free graphs it is not possible to pack ``very strong" balanced separators that are dominated by ``very few'' vertices. We will call such balanced separators {\em $c$-boosted balanced separators}. A somewhat simplified definition of a $c$-boosted balanced separator is a set $N[S]$ dominated by a set $S$ of at most $c$ vertices, such that no component of $G-N[S]$ has more than $|V(G)|/16c^2$ vertices (see Definition~\ref{def:bbs}). 
It turns out that on $S_{t,t,t}$-free graphs Lemma~\ref{lem:cant pack bs:intro} is true if ``balanced separators'' are replaced by ``$s$-boosted balanced separators'' for appropriately chosen integer $s$.

%For $\Sttt$-free graphs we follow a more complicated strategy, but at its most basic level is similar to the $P_k$-free independent set algorithm: pick up a balanced separator $N[S]$ (dominated by a small set $S$) and branch, pick up a balanced separator and branch, etc. For $P_k$-free graphs we only needed to concern ourselves with a single type of balanced separator, which would be balanced separators such that $G-S$ has no component with more than $n/2$ vertices. In order to make these ideas work for $\Sttt$-free graphs, we must pick up two distinct types of balanced separators. The branching step remains the same though, we branch on any vertex that has more than $n/2^i$ neighbors that belong to $i$ or more of the collected balanced separators.

%The first type of balanced separators we call {\em $c$-boosted balanced separators}. A somewhat simplified definition of a $c$-boosted balanced separator is a set $N[S]$ dominated by a set $S$ of at most $c$ vertices, such that no component of $G-N[S]$ has more than $|V(G)|/16c^2$ vertices (see Definition~\ref{def:bbs}). We state a lemma that is analogous to Lemma~\ref{lem:cant pack bs:intro} for boosted balanced separators. We skip sketching the proof here (see Section~\ref{sec:cannot pack bbs} for a formal statement and proof of this lemma).

\begin{lemma}\label{lem:cant pack bbs:intro}
Let $G$ be an $n$-vertex $\Sttt$-free graph, $s$ an integer, and ${\cal F}$ a multi-set of subsets of $V(G)$ such that every set in ${\cal F}$ is an $s$-boosted balanced separator. Assume no vertex belongs to more than $c$ sets of ${\cal F}$. Then, provided $|{\cal F}| \geq 80sct$, no component of $G$ contains over $3n/4$ vertices.
\end{lemma}

We skip sketching the proof of Lemma~\ref{lem:cant pack bbs:intro} here (see Section~\ref{sec:cannot pack bbs} for a formal statement and proof of this lemma), but we will remark that one of the key ingredients of the proof is a probabilistic argument akin to the proof of Lemma~\ref{lem:cant pack bs:intro} (the proof of Lemma~\ref{lem:cant pack bbs} is a bit more involved).

At this point we are one ``disconnect'' away from being able to utilize the strategy for $P_k$ free graphs: Theorem~\ref{thm:esd:intro} keeps giving us balanced separators, while Lemma~\ref{lem:cant pack bbs:intro} tells us that we can't pack {\em boosted}\ balanced separators. 
Indeed, if we assumed our $\Sttt$-free graphs always had, say, $c_t$-boosted balanced separators (where $c_t$ is some constant that depends on $t$), then by the exact same reasoning as before, the strategy of iteratively collecting a $c_t$-boosted balanced separator and then branching (on all vertices that have over $n/2^i$ neighbors that belong to $i$ or more of the collected $c_t$-boosted balanced separators) would work. Any vertex that belongs to $\log n$ of the collected $c_t$-boosted balanced separators will then be branched on, so no vertex will ever belong to over $\log n$ of the collected balanced separators. So, by Lemma~\ref{lem:cant pack bbs:intro}, after collecting $80c_tt\log n$ of these $c_t$-boosted balanced separators, the graph will not have any large component. A running time analysis identical to the one for $P_k$-free graphs~\cite{GartlandL20} would then show that this algorithm runs in quasi-polynomial time.

%This lemma suggests that we can copy the branching strategy we saw for $P_k$-free graphs, but use boosted balanced separators instead. Indeed, if we assumed our $\Sttt$-free graphs always had, say, $c_t$-boosted balanced separators (where $c_t$ is some constant that depends on $t$), then by the exact same reasoning as before, the strategy of iteratively collecting a $c_t$-boosted balanced separator and then branching (on all vertices that have over $n/2^i$ neighbors that belong to $i$ or more of the collected $c_t$-boosted balanced separators) would work. Any vertex that belongs to $\log n$ of the collected $c_t$-boosted balanced separators will then be branched on, so no vertex will ever belong to over $\log n$ of the collected balanced separators. So, by Lemma~\ref{lem:cant pack bbs:intro}, after collecting $80c_tt\log n$ of these $c_t$-boosted balanced separators, the graph will not have any large component. A runtime analysis would then show that this algorithm runs in quasi-polynomial time. 

Is it possible to bridge the ``disconnect'' from the other side and keep getting {\em boosted}\ balanced separators? This looks difficult, but we are able to bridge the gap algorithmically, by branching in such a way that a ``normal'' balanced separator becomes boosted. We can then add this boosted balanced separator to our collection of previously created boosted balanced separators, and then apply Lemma~\ref{lem:cant pack bbs:intro} to this collection to conclude that the graph gets sufficiently disconnected before the collection grows too large. We now sketch how to ``boost'' a separator. 

%The issue here is that we can only assume that our $\Sttt$-free graphs have balanced separators dominated by few vertices, while $c_t$-boosted balanced separators are something even stronger which we cannot assume to exist. It turns out though that using balanced separators (dominated by say $c_t$ vertices) and branching in a manner similar to what we have done before, we can enhance an ordinary balanced separator (dominated by few vertices) into a $c_t$-boosted balanced separator, which we then add to our collection of previously created boosted balanced separators. We sketch how this works.

\paragraph{Boosting Separators.}
We begin with a balanced separator $N[S]$, dominated by a set $S$ of at most $c_t$ vertices, such that no component of $G-N[S]$ has more than $n/2$ vertices.
%vertices and $S$ is dominated by, say, $c_t$ vertices, which we can assume exists. 
(For technical reasons in the actual algorithm $N[S]$ is not a balanced separator, but rather a set given by Theorem~\ref{thm:ICALP:weight} so that $G-N[S]$ has an extended strip decomposition with no large particles; from the viewpoint of efficient independent set algorithms this is just as useful.) We wish to turn $N[S]$ into a $c_t$-boosted balanced separator. In order to do this, we consider all vertices of $N[S]$ that have a neighbor in a large component of $G-N[S]$; we call this set $\rel(G,S)$ (see Figure~\ref{fig:relevant}. This is a slight simplification of the actual definition of $\rel(G,S)$ that we use in the algorithm, see Definition~\ref{def:relevant}). 
By ``large component'' we mean any component of $G-N[S]$ that has more than $n/16c_t^2$ vertices (note that if there are no such components, then $N[S]$ is a $c_t$-boosted balanced separator). In order to branch in a way that turns $N[S]$ into a $c_t$-boosted balanced separator, we use the following lemma, similar to Lemmas~\ref{lem:cant pack bs:intro} and \ref{lem:cant pack bbs:intro}. 

\begin{figure}
    \centering
    \includegraphics[width=.6\textwidth]{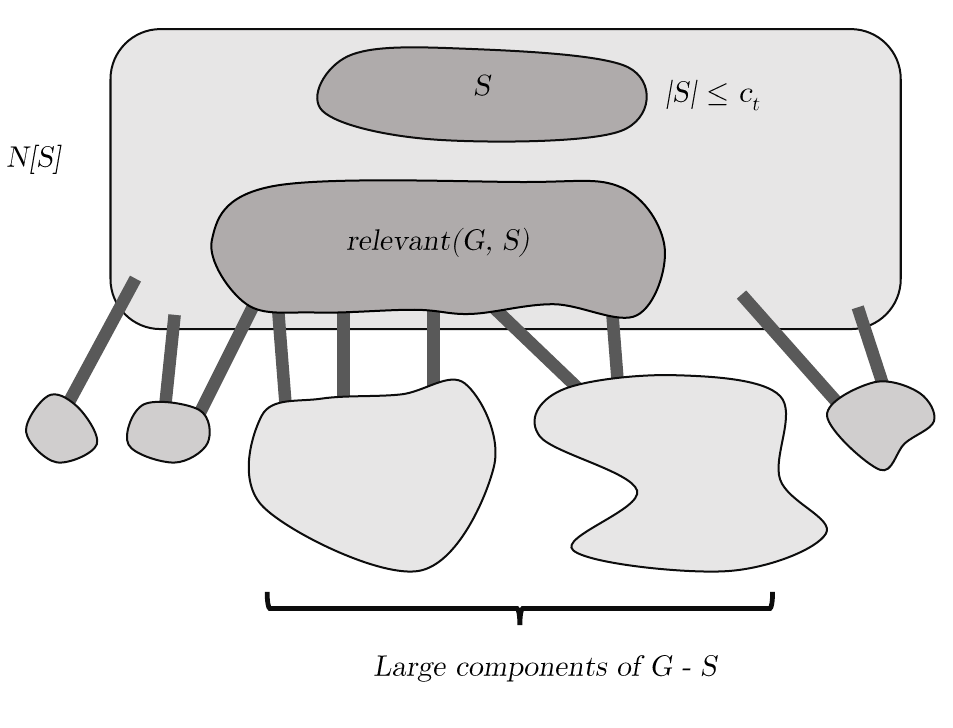}
    \caption{Illustration of how the set $\rel(G,S)$ is obtained from $S$.}
    \label{fig:relevant}
\end{figure}

\begin{lemma}\label{lem:cant pack bs 2:intro}
Let $G$ be an $n$-vertex $\Sttt$-free graph, let $N[S]$ be a balanced separator for $G$ dominated by a set $S$ of at most $c_t$ vertices, and let ${\cal F}$ be a multi-set of $|\rel(G,S)|/100c_t^3$-balanced separators for $(G, \rel(G,S))$. Assume no vertex belongs to over $c$ sets of ${\cal F}$. If $|{\cal F}| \geq 10ct$, either $S$ is a $c_t$-boosted balanced separator or no component of $G$ contains more than $3n/4$ vertices.
\end{lemma}

The proof of Lemma~\ref{lem:cant pack bs 2:intro} follows a similar ``expectation argument'' that Lemma~\ref{lem:cant pack bs:intro} uses, although it is a bit more involved. We do not sketch a proof of Lemma~\ref{lem:cant pack bs 2:intro} here (this lemma statement is more or less a combination of Observation~\ref{obs:bbs or good esd} and Lemma~\ref{lem:cant pack strong bs})

This lemma suggests the following branching strategy. We first pick up an $n/2$-balanced separator $N[S]$ dominated by a set $S$ of $c_t$ vertices, and we will try use Lemma~\ref{lem:cant pack bs 2:intro} to turn $N[S]$ into a $c_t$-boosted balanced separator or break up $G$ into small components. We use the same reasoning as before: iteratively collect $|\rel(G,S)|/100c_t^3$-balanced separators for $(G, \rel(G,S))$ and branch (on all vertices that have over $n/2^i$ neighbors that belong to $i$ or more of the collected balanced separators). Any vertex that belongs to $\log n$ of the collected  balanced separators will then be branched on, so no vertex will ever belong to over $\log(n)$ of the collected balanced separators. So, by Lemma~\ref{lem:cant pack bs 2:intro} after collecting $10t\log n$ of these $|\rel(G,S)|/100c_t^3$-balanced separators for $(G, \rel(G,S))$, either the graph will have no large component (and then we make large progress by calling the algorithm recursively on the components) or $S$ is now a $c_t$-boosted balanced separator, which we then add to our collection of $c_t$-boosted balanced separators. By Lemma~\ref{lem:cant pack bbs:intro} 
%reasoning similar to what we have already seen, 
this collection cannot grow larger than $80c_tt\log n$ before our graph no longer has large connected components. 

The running time analysis of this algorithm essentially looks like this: if we could assume that boosting a single balanced separator to become a boosted balanced separator took constant time, then the analysis would be more or less identical to the analysis of the algorithm for MWIS on $P_k$-free graphs. 
However, now each individual ``boosting'' step is instead a branching algorithm whose analysis again is very similar to the analysis of the algorithm for MWIS on $P_k$-free graphs, so each boosting step corresponds to a recursive algorithm with quasi-polynomially many leaves. 
Since quasi-polynomial functions compose the entire running time is still quasi-polynomial. Finally we need to take into account what would happen if outcome (3) of Theorem~\ref{thm:esd:intro} does occur, but this can fairly easily be shown to only be good for the progress of the algorithm.

%A careful runtime analysis of this algorithm shows that it runs in quasi-polynomial~time.

\section{Preliminaries}\label{sec:prelims}
\paragraph*{Basic notation.}
%\pg{define length of path, $\Sttt$, anti-complete = anti-adjacent, path in a directed graph?}

For a family $\mathcal{Q}$ of sets, by $\bigcup \mathcal{Q}$ we denote $\bigcup_{Q \in \mathcal{Q}} Q$.
Let $G$ be a graph. For $X \subseteq V(G)$, by $G[X]$ we denote the subgraph of $G$ induced by $X$, i.e., $(X, \{uv \in E(G) : u,v \in X\})$.
If the graph $G$ is clear from the context, we will often identify induced subgraphs with their vertex sets.
The sets $X,Y \subseteq V(G)$ are \emph{complete} to each other if for every $x \in X$ and $y \in Y$ the edge $xy$ is present in $G$.
Note that this, in particular, implies that $X$ and $Y$ are disjoint.
We say that two sets $X,Y$ \emph{touch} if $X \cap Y \neq \emptyset$ or there is an edge with one endpoint in $X$ and another in $Y$.
Finally, two disjoint sets are \emph{anti-adjacent} or \emph{anti-complete} if they do not touch.

For a vertex $v$, by $N_G(v)$ we denote the set of neighbors of $v$, and by $N_G[v]$ we denote the set $N_G(v) \cup \{v\}$.
For a set $X \subseteq V(G)$, we also define $N_G(X) \coloneqq  \bigcup_{v \in X} N_G(v) \setminus X$, and $N_G[X]\coloneqq N_G(X) \cup X$.
If it does not lead to confusion, we omit the subscript and write simply $N( \cdot)$ and $N[ \cdot]$. Additionally, if $G'$ is an induced subgraph of $G$, we use $N_G^{G'}(X)$ and  $N_G^{G'}[X]$ to mean $N_G(X) \cap V(G')$ and $N_G[X] \cap V(G')$ respectively.
We often say that a set of vertices $X\subseteq V(G)$ is \emph{dominated} by a set $Y\subseteq V(G)$ if $X\subseteq N_G[Y]$.

The length of a path is the number of edges of the path. $P_t$ denotes an induced path with $t$ vertices (and $t-1$ edges). A claw is a set of three independent vertices, $v_1$, $v_2$, and $v_3$ along with a a vertex $u$ that is neighbors with each $v_i$. An $\Sttt$ is three anti-complete $P_t$'s along with a vertex $u$ that is neighbors with exactly one endpoint from each $P_t$ and no other vertices, so a claw is $S_{1,1,1}$.

Given a graph $G$ and a graph $H$, $G$ is said to be {\em $H$-free} if $G$ does not contain $H$ as an induced subgraph. If ${\cal H}$ is a set of graphs, then $G$ is ${\cal H}$-free if for each $H \in {\cal H}$, $G$ is $H$-free.

\paragraph*{Balanced separators.} 
We define a {\em vertex list}, or more simple a {\em list}, to be an ordered multi-set of subsets $V(G)$. If ${\cal F} = \{F_1, F_2, \ldots, F_k\}$ is a list and $S \subseteq V(G)$ we define ${\cal F} \cup S$ to be the list ${\cal F}$ with $S$ appended at the end, that is ${\cal F} \cup S \coloneqq \{F_1, F_2, \ldots, F_k, S\}$. We define $N_G^{G'}[{\cal F}] \coloneqq \{N_G^{G'}[F_1], N_G^{G'}[F_2], \ldots, N_G^{G'}[F_k]\}$.

Let $G$ be a graph, $G'$ an induced subgraph of $G$, $Y \subseteq V(G')$, $c$ non-negative integer, and $\w$ a weight function for the vertices of $G'$. We say $Y$ is a {\em $c$-balanced separator for $(G', \w)$} if no component, $C$, of $G'-Y$ has $\w(C) > c$.
Now let $Z \subseteq V(G')$. We say that $Y$ is a {\em $c$-balanced separator for $(G',Z)$} when no component of $G'-Y$ contains over $c$ vertices of $Z$. When $Z = V(G')$ then we say that $Y$ is a {\em $c$-balanced separator for $G'$}. Furthermore, if there is a set $X \subseteq V(G)$ such that $Y = N_G^{G'}[X]$ then we say that {\em $Y$ has a core $X$ originating in $G$}. We note that while unintuitive, if $Y$ is a $c$-balanced separator for $(G', \w)$, it may be possible for $G'-Y$ to have {\em fewer} components then $G'$. For instance this is true when $Y = V(G')$.

\paragraph*{Extended strip decompositions.} 
By $T(G)$, we denote the set of all triangles in $G$.
Similarly to writing $xy \in E(G)$, we write $xyz \in T(G)$ to indicate that $G[\{x,y,z\}] \simeq K_3$.
Now let us define a certain graph decomposition which will play an important role in the paper.
An \emph{extended strip decomposition} of a graph $G$ is a pair $(H, \eta)$ that consists of:
\begin{itemize}
\item a simple graph $H$,
\item a \emph{vertex set} $\eta(x) \subseteq V(G)$ for every $x \in V(H)$,
\item an \emph{edge set} $\eta(xy) \subseteq V(G)$ for every $xy \in E(H)$, and its subsets $\eta(xy,x),\eta(xy,y) \subseteq \eta(xy)$,
\item a \emph{triangle set} $\eta(xyz) \subseteq V(G)$  for every $xyz \in T(H)$,
\end{itemize}
which satisfy the following properties (also see \cref{fig:esd}):
\begin{enumerate}
\item The family $\{\eta(o)~|~o \in V(H)\cup E(H) \cup T(H)\}$ is a partition of $V(G)$.
\item For every $x \in V(H)$ and every distinct $y,z \in N_H(x)$, the set $\eta(xy,x)$ is complete to $\eta(xz,x)$.
\item Every $uv \in E(G)$ is contained in one of the sets $\eta(o)$ for $o \in V(H) \cup E(H)\cup T(H)$, or is as follows:
\begin{itemize}
\item $u \in \eta(xy,x), v\in \eta(xz,x)$ for some $x \in V(H)$ and $y,z \in N_H(x)$, or
\item $u \in \eta(xy,x), v\in \eta(x)$ for some $xy \in E(H)$, or
\item $u \in \eta(xyz)$ and $v\in \eta(xy,x) \cap \eta(xy,y)$ for some $xyz \in T(H)$. 
\end{itemize}
\end{enumerate}
%\mpil[inline]{Added requirement that edge bags and edge interfaces are nonempty. (Vertex and triangle sets can be empty.)}
%An object $(H,\eta)$ satisfying all requirements of an extended strip decomposition, except for possibly non-emptiness of 
%sets $\eta(xy)$ and $\eta(xy,x)$, $\eta(xy,y)$ is called a \emph{relaxed extended strip decomposition}.
An extended strip decomposition $(H,\eta)$ is \emph{rigid} if for every $xy \in E(H)$, the sets $\eta(xy)$, $\eta(xy,x)$, and
$\eta(xy,y)$ are nonempty, and for every \emph{isolated} (i.e., with no incident edge) $x \in V(H)$, the set $\eta(x)$ is nonempty.

\begin{figure}
    \centering
    \includegraphics{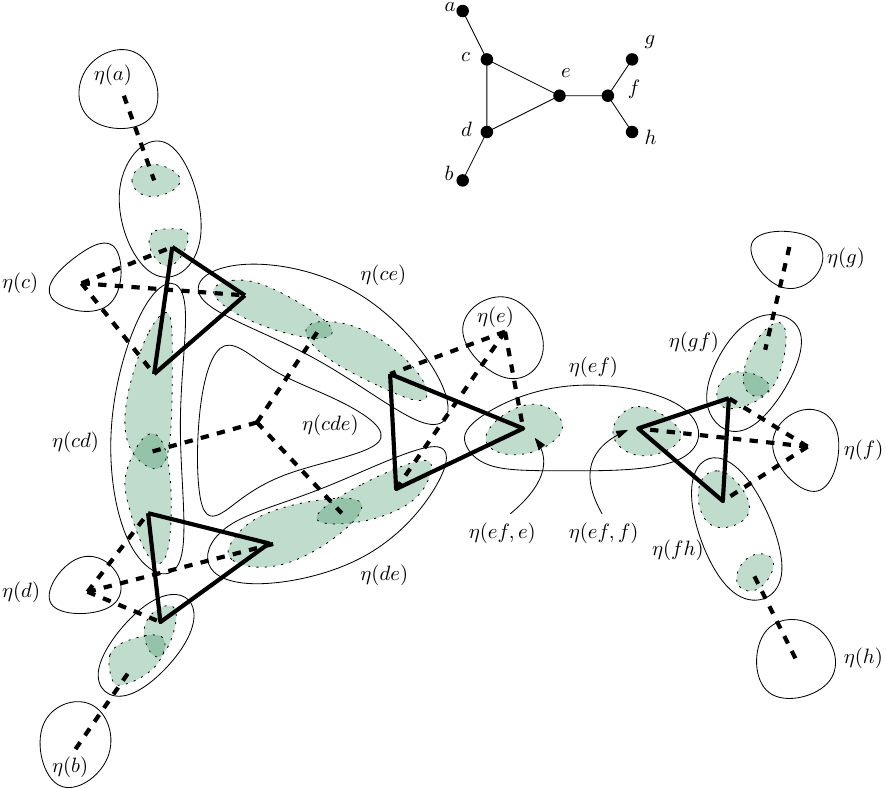}
    \caption{A graph $H$ and an extended strip decomposition $(H,\eta)$ of some graph $G$. Edges within sets $\eta(\cdot)$ are arbitrary. A solid edge across two sets indicates that there are complete to each other. A dashed edge means that edges across these sets are allowed but not mandatory. No edge means that the sets do not touch.}
    \label{fig:esd}
\end{figure}

We say that a vertex $v \in V(G)$ is \emph{peripheral} in $(H,\eta)$ if there is a degree-one vertex $x$ of $H$,
such that $\eta(xy,x)=\{v\}$, where $y$ is the (unique) neighbor of $x$ in $H$.
For a set $Z \subseteq V(G)$, we say that $(H, \eta)$ is an \emph{extended strip decomposition of $(G,Z)$} if $H$ has $|Z|$ degree-one vertices and each vertex of $Z$ is peripheral in $(H,\eta)$.

The following theorem by Chudnovsky and Seymour~\cite{DBLP:journals/combinatorica/ChudnovskyS10}
is a slight strengthening of their celebrated solution of the famous \emph{three-in-a-tree} problem.
We will use it as a black-box to build extended strip decompositions.

\begin{theorem}[{\cite[Section 6]{DBLP:journals/combinatorica/ChudnovskyS10}}]\label{thm:three-in-a-tree}
Let $G$ be an $n$-vertex  graph and consider $Z \subseteq V(G)$ with $|Z| \geq 2$.
There is an algorithm that runs in time $\Oh(n^5)$ and returns one of the following:
\begin{itemize}
\item an induced subtree of $G$ containing at least three elements of $Z$, or
\item a rigid extended strip decomposition $(H,\eta)$ of $(G,Z)$.
\end{itemize}
\end{theorem}

Let us point out that actually, an extended strip decomposition produced by \Cref{thm:three-in-a-tree} satisfies more structural properties,
but for our purpose, we will only use the fact that it is rigid. 

\paragraph*{Particles of extended strip decompositions.}
Let $(H, \eta)$ be an extended strip decomposition of a graph $G$.
We introduce some special subsets of  $V(G)$ called \emph{particles}, divided into five \emph{types}.
\begin{align*}
\textrm{vertex particle:} &\quad A_{x} \coloneqq  \eta(x) \text{ for each } x \in V(H),\\
\textrm{edge interior particle:} &\quad A_{xy}^{\perp} \coloneqq  \eta(xy) \setminus (\eta(xy,x) \cup \eta(xy,y)) \text{ for each } xy \in E(H),\\
\textrm{half-edge particle:} &\quad A_{xy}^{x} \coloneqq   \eta(x) \cup \eta(xy) \setminus \eta(xy,y) \text{ for each } xy \in E(H),\\
\textrm{full edge particle:} &\quad A_{xy}^{xy} \coloneqq  \eta(x) \cup \eta(y) \cup \eta(xy) \cup \bigcup_{z ~:~ xyz \in T(H)} \eta(xyz) \text{ for each } xy \in E(H),\\
\textrm{triangle particle:} &\quad A_{xyz} \coloneqq  \eta(xyz) \text{ for each } xyz \in T(H).
\end{align*}

\paragraph*{Wall notation.}

A \emph{wall} of \emph{sidelength} $\ell$ is depicted in Figure~\ref{fig:wall}; it consists of $\ell$ rows and $\ell$ columns
as in the figure.
A \emph{peg} is a vertex of degree three in a wall.
A path between two pegs that has no other peg as an internal vertex is called a \emph{basic path} in a wall.
We say that wall is $k$-\emph{subdivided} if every basic path has length more than $k$.
A \emph{subwall} of a wall $W$ is a wall whose rows and columns are subpaths of the rows and columns of $W$.

\begin{figure}
    \centering
    \includegraphics{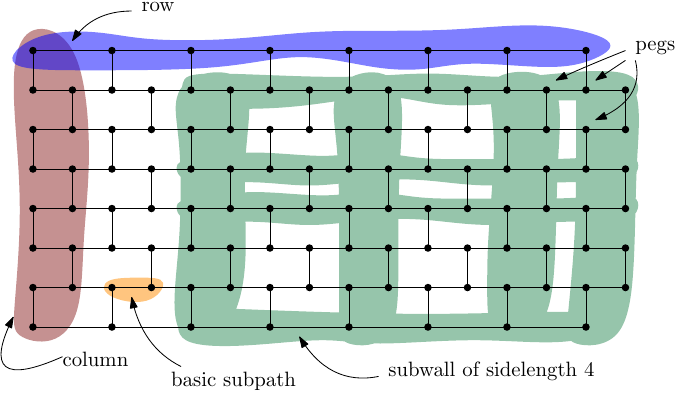}
    \caption{Wall of sidelength 8. The lines between pegs denote paths of arbitrary length.}
    \label{fig:wall}
\end{figure}

\paragraph*{Separations and tangles.}
Let $G$ be a graph. A \emph{separation} in $G$ is an ordered pair $(A,B)$ of vertex sets $A,B \subseteq V(G)$ such that
 $A \cup B = V(G)$ and there is no edge of $G$ with one endpoint in $A \setminus B$ and the second endpoint in $B \setminus A$. 
The \emph{order} of the separation $(A,B)$ is $|A \cap B|$.

A \emph{tangle of order $k$} in a graph $G$ is a family $\mathcal{T}$ of separations of order less than $k$ such that:
\begin{itemize}
    \item For every separation $(A,B)$ of order less than $k$ in $G$, exactly one of the separations $(A,B)$ and $(B,A)$ belongs to $\mathcal{T}$.
    \item For every triple $(A_1,B_1),(A_2,B_2),(A_3,B_3) \in \mathcal{T}$ we have $A_1 \cup A_2 \cup A_3 \neq V(G)$. 
\end{itemize}
Observe that if $\mathcal{T}$ is a tangle of order $k$ and $k' < k$, then the family $\mathcal{T}'$ consisting of all separations of $\mathcal{T}$
of order less than $k'$ is a tangle of order $k'$. We call such $\mathcal{T}'$ the \emph{restriction of $\mathcal{T}$ to order $k'$}.

Let $W$ be a wall in $G$ of sidelength $k$.
Let $(A,B)$ be a separation in $G$ of order $k' < k$. Note that for exactly one $\Gamma \in \{A \setminus B, B \setminus A\}$,
$\Gamma$ contains at least $k-k'$ full rows and at least $k-k'$ full columns of $W$. 
Let $\mathcal{T}_W$ be the family of those separations $(A,B)$ of order less than $\lceil k/3 \rceil$ such that
$B \setminus A$ contains at least $k - \lceil k/3 \rceil + 1$ full rows and at least $k - \lceil k/3 \rceil + 1$ full columns of $W$.
It is straightforward to verify that $\mathcal{T}_W$ is a tangle of order $\lceil k/3 \rceil$; we call it the tangle
\emph{governed by $W$}. 

We make the following simple but important observation.
\begin{lemma}\label{lem:subwall-tangle}
If $W$ is a wall in a graph $G$ and $W'$ is a subwall of $W$, then $\mathcal{T}_{W'} \subseteq \mathcal{T}_W$.
\end{lemma}
\begin{proof}
Let $k$ and $k'$ be the sidelengths of $W$ and $W'$, respectively. 
Let $(A,B)\in\mathcal{T}_{W'}$ be a separation of order less than $\lceil k'/3 \rceil$.
Then, $B \setminus A$ contains at least $k' - \lceil k'/3 \rceil + 1$ full rows of $W$ and at least $k' - \lceil k'/3 \rceil + 1$ full
columns of $W$. Since $W'$ is a subwall of $W$, $B \setminus A$ contains at least $k - \lceil k'/3 \rceil + 1$ full rows of $W$
and at least $k - \lceil k'/3 \rceil + 1$ full columns of $W$. 
Hence, $(A,B) \in \mathcal{T}_{W}$, as desired.
\end{proof}

We will need the following result, which follows from the combination of the polynomial grid minor theorem~\cite{ChekuriC16,ChuzhoyT21},
the duality of tangles and branchwidth~\cite{GM10}, and \cite[Lemma~14.6]{KTW20}.
\begin{theorem}\label{thm:tangle2wall}
There exists a function $f_{\mathrm{KTW20}}(k) = \widetilde{O}(k^{18})$
such that if a graph $G$ admits a tangle $\mathcal{T}$ of order $f_{\mathrm{KTW20}}(k)$ for an integer $k$,
then $G$ contains a wall $W$ of sidelength $3k$ such that $\mathcal{T}_W$ is the restriction of $\mathcal{T}$ to order $k$.
\end{theorem}

\section{Extended Strip Lemma}\label{sec:esd}
%\pg{in the IND set algorithm we use that lemma \ref{lem:esd} runs in qp time. We should probably justify that somewhere here / add it to the statement of the lemma since we use this as a subroutine in the algorithm}
The main result of this section is the following:
\begin{lemma}[Extended strip decomposition or small balanced separator]\label{lem:esd}
For every fixed integer~$t$, there exists an integer $c_t$ and a polynomial-time algorithm that, given an $n$-vertex graph $G$
and a weight function $\w : V(G) \to [0,+\infty)$, 
returns one of the following:
\begin{enumerate}
    \item an induced copy of $S_{t,t,t}$ in $G$;
    \item a $0.99 \w(G)$-balanced separator for $(G,\w)$ dominated by $c_t \cdot \log n$ vertices;
    \item a rigid extended strip decomposition of $G$ where no particle is of weight larger than $0.5\w(G)$.
\end{enumerate}
\end{lemma}

The main difference between Lemma~\ref{lem:esd} and the main result of~\cite{ICALP-qptas}, namely Theorem~\ref{thm:ICALP:weight},
is that Lemma~\ref{lem:esd} promises in the last output an extended strip decomposition of the entire graph, not the graph
with a small number of neighborhoods deleted. 

The algorithm of Lemma~\ref{lem:esd} first applies Theorem~\ref{thm:ICALP:weight} to find either an induced copy of $S_{t,t,t}$
(which can be immediately returned) or a set $Z$ of size $\Oh(\log n)$ together with a rigid extended strip decomposition $(H,\eta)$
of $G-N[Z]$ such that every particle of $(H,\eta)$ has weight at most $0.5 \w(G)$. 
Then, we attempt to put back vertices of $N[Z]$ one-by-one to $(H,\eta)$, maintaining the property that every particle
of $(H,\eta)$ has weight at most $0.5 \w(G)$. The following lemma, whose proof spans the remainder of this section,
shows that in every such attempt, we can either succeed or obtain one of the first two outcomes of Lemma~\ref{lem:esd}.

\begin{lemma}\label{lem:esd:step}
For every fixed integer $t$ there exists an integer $c_t$ and a polynomial-time algorithm that, given an $n$-vertex graph $G$,
a weight function $\w : V(G) \to [0,+\infty)$, a real $\tau \geq \w(G)$, a vertex $v \in V(G)$,
and a rigid extended strip decomposition $(H,\eta)$ of $G-v$ with every particle of weight at most $0.5\tau$, 
returns one of the following:
\begin{enumerate}
    \item an induced copy of $S_{t,t,t}$ in $G$;
    \item a $0.99 \tau$-balanced separator for $(G,\w)$ dominated by at most $c_t$ vertices;
    \item a rigid extended strip decomposition of $G$ where no particle is of weight larger than $0.5\tau$.
\end{enumerate}
\end{lemma}

Let us formally prove Lemma~\ref{lem:esd} using Lemma~\ref{lem:esd:step}.
\begin{proof}[Proof of Lemma~\ref{lem:esd}.]
Let $\tau \coloneqq \w(G)$. 
Run Theorem~\ref{thm:ICALP:weight} on $(G,\w)$. If an $S_{t,t,t}$ is returned, return it as well. 
Otherwise, we have a set $Z$ of size $\Oh(\log n)$ together with a rigid extended strip decomposition $(H,\eta)$
of $G-N[Z]$ such that every particle of $(H,\eta)$ has weight at most $0.5 \tau$. 

Enumerate $N[Z]$ as $\{v_1,v_2,\ldots,v_k\}$. Let $G_i = G-\{v_1,\ldots,v_i\}$ for $0 \leq i \leq k$, so that 
$G_0 = G$ and $G_k = G-N[Z]$. Denote $(H_k,\eta_k) \coloneqq (H,\eta)$.
We
compute a sequence $(H_i,\eta_i)_{i=k}^0$ of rigid extended strip decompositions of graphs $G_i$ whose every particle has weight at most $0.5\tau$
as follows.
For each $i=k,k-1,\ldots,1$ apply Lemma~\ref{lem:esd:step} to $G_{i-1}$, $v_i$ (recall that $G_i = G_{i-1}-v_{i-1}$), $\tau$,
and the rigid extended strip decomposition $(H_i,\eta_i)$. 
If an $S_{t,t,t}$ is returned, terminate the algorithm and return it, too.
If a $0.99\tau$-balanced separator $X$ is returned, return $X \cup N[Z]$ as a $0.99\tau$-balanced separator of $G$
dominated by $\Oh(\log n)$ vertices.
Otherwise, denote the output rigid extended strip decomposition of $G_{i-1}$ by $(H_{i-1},\eta_{i-1})$ and continue with the next step.
If we reach $(H_0,\eta_0)$, we return it as the third output of Lemma~\ref{lem:esd}.
\end{proof}

It will be useful in future work in this direction if we prove a slight strengthening of Lemma~\ref{lem:esd:step}, and will only add a small amount of additional work. To this end, we define a {\em $\Pttt$} to be an $S_{t_1,t_2,t_3}$, where $t_1,t_2,t_3 \geq t$, with the three leaves labeled $x_1,x_2$, and $x_3$, along with three additional vertices $y_1, y_2$, and $y_3$ such that the $y_i$'s are complete with each other (they induced a $K_3$), $y_i$ is a neighbor of $x_j$ if and only if $i = j$, and the $y_i$'s have no other neighbors in the $S_{t_1,t_2,t_3}$. The following lemma is the strengthening of Lemma~\ref{lem:esd:step}. 

\begin{lemma}\label{lem:esd:step-pyramid}
For every fixed integer $t$ there exists an integer $c_t$ and a polynomial-time algorithm that, given an $n$-vertex graph $G$,
a weight function $\w : V(G) \to [0,+\infty)$, a real $\tau \geq \w(G)$, a vertex $v \in V(G)$,
and a rigid extended strip decomposition $(H,\eta)$ of $G-v$ with every particle of weight at most $0.5\tau$, 
returns one of the following:
\begin{enumerate}
    \item an induced copy of $\Pttt$ in $G$;
    \item a $0.99 \tau$-balanced separator for $(G,\w)$ dominated by at most $c_t$ vertices;
    \item a rigid extended strip decomposition of $G$ where no particle is of weight larger than $0.5\tau$.
\end{enumerate}
\end{lemma}

Since a $\Pttt$ contains an $\Sttt$ as an induced subgraph, it is trivial to see that Lemma~\ref{lem:esd:step-pyramid} implies Lemma~\ref{lem:esd:step}. The remainder of this section is devoted to the proof of Lemma~\ref{lem:esd:step-pyramid}.

%\TM[inline]{not needed now
%Paths $P=p_1,p_2,p_3,\ldots$ and $Q=q_1,q_2,q_3$ in $H$ are \emph{almost vertex-disjoint} if there are vertex-disjoint except possibly $p_1=q_1$, if that happens we call it the  \emph{merging} vertex.
%A set of paths is \emph{almost vertex-disjoint} if two of them are almost vertex-disjoint and others are vertex-disjoint.
%}

\subsection{Turning separations in $H$ into separators in $G$}

Let us make the following trivial observation.
\begin{lemma}\label{lem:dom-potato}
    If $(H,\eta)$ is a rigid extended strip decomposition of a graph $G$
    and $x \in V(H)$ is of degree more than one,
    then $\bigcup_{y \in N_{H}(x)} \eta(xy,x)$ is dominated by two vertices.
\end{lemma}
\begin{proof}
    Pick two neighbors $y_1,y_2 \in N_{H}(x)$ and any $v_i \in \eta(xy_i,x)$ for $i=1,2$.
    (Recall that we mandate the interfaces $\eta(xy_i,x)$ to be nonempty in a rigid extended strip decomposition.)
    Then, $v_i$ dominates $\bigcup_{y \in N_{H}(x) \setminus \{y_i\}} \eta(xy,x)$, so $\{v_1,v_2\}$ dominates
    $\bigcup_{y \in N_{H}(x)} \eta(xy,x)$.
\end{proof}

For an extended strip decomposition $(H,\eta)$ of a graph $G$ and a set $A \subseteq V(H)$, the \emph{preimage} of $A$ in $G$
is the set $\preimage{(H,\eta)}{A} \subseteq V(G)$ consisting of:
\begin{itemize}
    \item all vertex sets $\eta(x)$ for $x \in A$;
    \item all edge sets $\eta(xy)$ for $|\{x,y\} \cap A| \geq 1$;
    \item all triangle sets $\eta(xyz)$ for $|\{x,y,z\} \cap A| \geq 2$.
\end{itemize}

We make the following two observations based on Lemma~\ref{lem:dom-potato}. 
\begin{lemma}\label{lem:sepH2G}
Let $(H,\eta)$ be an extended strip decomposition of a graph $G$ and let $(A,B)$ be a separation in $H$.
%such that every vertex of $A \cap B$ is of degree more than one in $H$. 
Let $X = \bigcup_{x \in A \cap B} \bigcup_{y \in N_H(x)} \eta(xy,x)$. 
Then, every connected component of $G-X$
is contained in one of the following sets: $\preimage{(H,\eta)}{A \setminus B}$, $\preimage{(H,\eta)}{B \setminus A}$, 
$\eta(x)$ for some $x \in A \cap B$, $\eta(xy)$ for some $xy \in E(H[A])$, or $\eta(xyz)$ for some triangle $xyz \in T(H)$
with $|\{x,y,z\} \cap A \cap B| \geq 2$. 
Furthermore, if $(H,\eta)$ is rigid and every vertex of $A \cap B$ has degree at least $2$,
then $X$ is dominated by at most $2|A \cap B|$ vertices.
\end{lemma}
\begin{proof}
Observe that every set $\Gamma$ being either $\eta(x)$ for $x \in A \cap B$, $\eta(xy)$ for some $xy \in E(H[A])$, or
$\eta(xyz)$ for a triangle $xyz \in T(H)$ with $|\{x,y,z\} \cap A \cap B| \geq 2$ satisfies $N_G(\Gamma) \subseteq X$. 
Similarly, every edge that has exactly one endpoint on $\preimage{(H,\eta)}{A \setminus B} \setminus X$ has its second endpoint in $X$
and every edge that has exactly one endpoint on $\preimage{(H,\eta)}{B \setminus A} \setminus X$ has its second endpoint in $X$.
This proves the desired separation properties of $X$. The second part of the lemma follows directly from Lemma~\ref{lem:dom-potato}.
\end{proof}

\begin{lemma}\label{lem:substantial-particle}
Let $0 < \delta < 0.5$ be a constant.
Let $(H,\eta)$ be an extended strip decomposition of a graph $G$ with weight function $\w$.
Assume that no particle of $(H,\eta)$ has weight more than $(1-\delta)\w(G)$,
but there is a particle of $(H,\eta)$ that has weight at least $\delta\w(G)$. 
Then there exists a set $F \subseteq V(H)$ of size at most $2$
such that $X \coloneqq \bigcup_{x \in F} \bigcup_{y \in N_H(x)} \eta(xy,x)$ is an
$(1-\delta)\w(G)$-balanced separator in $G$.
Furthermore, if $(H,\eta)$ is rigid, then $X$ is dominated by at most four vertices. 
\end{lemma}
\begin{proof}
    Observe that inclusion-wise maximal particles are vertex particles for isolated vertices of $H$
    and full edge particles. 
    Without loss of generality, we can assume that there is a particle of $(H,\eta)$ of one of those two types that has weight at least 
    $\delta\w(G)$.

    Assume first that $\w(\eta(x)) \geq \delta\w(G)$ for some isolated $x \in V(H)$.
    We also have $\w(\eta(x)) \leq (1-\delta)\w(G)$ by the assumptions of the lemma.
    Since $\eta(x)$ is the union of some connected components of $G$ by the properties of an extended strip decomposition, 
    $\emptyset$ is a $(1-\delta)\w(G)$-balanced separator in $G$ and the we are done.
    % \mipil{Did we define balanced sepsarators?}
    % \mpil{There is a todonote in Section~3 to move it to prelims.}

    Assume now that $\w(A_{xy}^{xy}) \geq \delta\w(G)$ for an edge $xy \in E(H)$. 
    Again, by the assumptions of the lemma we have $\w(A_{xy}^{xy}) \leq (1-\delta)\w(G)$.
    Let $F$ be the set of those vertices of $\{x,y\}$ that are of degree more than one in $H$
    and let $X \coloneqq \bigcup_{x' \in F} \bigcup_{y' \in N_H(x')} \eta(x'y',x')$.
    It follows from the properties of an extended strip decomposition
    that $X$ separates $A_{xy}^{xy}$ from $V(G) \setminus A_{xy}^{xy}$, i.e., every path from $A_{xy}^{xy}$ to $V(G) \setminus A_{xy}^{xy}$ contains a vertex from $X$.
    If $(H,\eta)$ is rigid, then Lemma~\ref{lem:dom-potato} implies that $X$ is dominated by at most $2|F| \leq 4$ vertices.
    Since $\delta\w(G) \leq \w(A_{xy}^{xy}) \leq (1-\delta)\w(G)$, $X$ is the desired $(1-\delta)\w(G)$-balanced separator.
\end{proof}

\subsection{Locally cleaning an extended strip decomposition}\label{sub:localCleaning}

We will need a few connectivity properties of an extended strip decomposition, a bit stronger than just being rigid.
Luckily, they are easy to obtain via local modifications.

Let $(H,\eta)$ be an extended strip decomposition of a graph $G$. 
A \emph{local cleaning step} for $(H,\eta)$ is one of the following modifications.
\begin{description}
    \item[removing an isolated vertex with an empty set]\phantomsection \label{mov:EmptyIV}
      If $x \in V(H)$ is an isolated vertex satisfying $\eta(x) = \emptyset$,
    delete $x$ from $V(H)$.
  \item[moving an isolated vertex set]\phantomsection \label{mov:IsolatedVS}
    If $x \in V(H)$ is an isolated vertex with $\eta(x)$ nonempty
    and ${y \in V(H)}$ is any other vertex that is not an isolated
    vertex with $\eta(y)$ empty, we set $\eta(y) \coloneqq \eta(y) \cup \eta(x)$ and $\eta(x) \coloneqq \emptyset$.
    \item[moving a disconnected component of an edge set]\phantomsection \label{mov:DisconncetedES}
      If for an edge $xy \in E(H)$ there exists a connected component $C$
    of $\eta(xy) \setminus (\eta(xy,x) \cup \eta(xy,y))$ with no neighbors in $\eta(xy,y) \setminus \eta(xy,x)$, we set $\eta(x) \coloneqq \eta(x) \cup C$
    and $\eta(xy) \coloneqq \eta(xy) \setminus C$. 
    \item[moving a disconnected component of a triangle set]\phantomsection \label{mov:DisconnectedTS}
      If $C$ is a connected component of a triangle $xyz \in T(H)$ with
    no neighbors in $\eta(yz,y) \cap \eta(yz,z)$, then we set $\eta(x) \coloneqq \eta(x) \cup C$ and $\eta(xyz) \coloneqq \eta(xyz) \setminus C$.
    \item[moving a disconnected vertex of an interface]\phantomsection \label{mov:DisconnectedVI}
      If for an edge $xy \in E(H)$ there is a vertex $v \in \eta(xy,x) \setminus \eta(xy,y)$
    such that $N_G[v] \cap \eta(xy) \subseteq \eta(xy,x)$, set $\eta(x) \coloneqq \eta(x) \cup \{v\}$ and $\eta(xy) \coloneqq \eta(xy) \setminus \{v\}$. 
    \item[removing an edge with an empty interface]\phantomsection \label{mov:EmptyEI}
      If $xy \in E(H)$ is an edge with $\eta(xy,x) = \emptyset$,
    we set $\eta(y) \coloneqq \eta(y) \cup \eta(xy)$ and delete the edge $xy$ from $H$.
    \item[suppressing a degree-1 vertex]\phantomsection \label{mov:Sup} If $x \in V(H)$ is of degree $1$ in $H$, with its unique neighbor $y$, set $\eta(y) \coloneqq \eta(y) \cup \eta(xy) \cup \eta(x)$, $\eta(x) \coloneqq \emptyset$ and delete the edge $xy$. 
\end{description}
An extended strip decomposition is \emph{locally cleaned} if no local cleaning step is applicable.

The following observations are immediate.
\begin{lemma}\label{lem:clean-step-ok}
If $(H,\eta)$ is an extended strip decomposition of $G$ and $(H',\eta')$ is a result of applying the first applicable local cleaning step
to $(H,\eta)$, then $(H',\eta')$ is also an extended strip decomposition of $G$ with $V(H') \subseteq V(H)$ and $E(H') \subseteq E(H)$.
Furthermore, we have the following:
\begin{itemize}
    \item For every $xy \in E(H')$ we have $\eta'(xy) \subseteq \eta(xy)$, $\eta'(xy,x) \subseteq \eta(xy,x)$, and $\eta'(xy,y) \subseteq \eta(xy,y)$.
    \item For every $xyz \in T(H')$, we have $\eta'(xyz) \subseteq \eta(xyz)$. 
    \item The following potential strictly increases from $(H,\eta)$ to $(H',\eta')$: 
the number of vertices of $G$ in  vertex sets of  $H/H'$, minus the number of vertices and edges of $H/H'$, and additionally 
minus the number of vertices of $H/H'$ that are not isolated vertices with empty vertex sets.
\end{itemize}
\end{lemma}
\begin{proof}
    The only nontrivial check is that whenever we delete an edge $e$ of $H$, all triangles involving $e$ already have empty sets.
    This follows for the ``\hyperref[mov:EmptyEI]{removing an edge with an empty interface}'' step due to inapplicability of the
    ``\hyperref[mov:DisconnectedTS]{moving a disconnected component of a triangle}'' step. 
\end{proof}

Note that the last property ensures that the local cleaning operation terminates and indeed produces a locally cleaned extended strip decomposition.
\begin{lemma}\label{lem:locally-cleaned-ok}
Let $(H,\eta)$ be an extended strip decomposition of $G$ that is locally cleaned.
Then $(H,\eta)$ is a rigid extended strip decomposition such that 
either $H$ consists of a single vertex with the whole $V(G)$ in its vertex set,
or every vertex of $H$ has degree at least $2$.
\end{lemma}
\begin{proof}
  If there was an isolated $x \in V(H)$ with nonempty $\eta(x)$, the ``\hyperref[mov:IsolatedVS]{moving isolated vertex set}'' step would apply,
    unless already $\eta(x) = V(G)$.
    If there were an isolated vertex $x \in V(H)$ with $\eta(x) = \emptyset$, the ``\hyperref[mov:IsolatedVS]{removing an isolated vertex with an empty set}'' step
    would apply.
    If there were a vertex $x \in V(H)$ of degree one, the ``\hyperref[mov:Sup]{suppressing a degree-1 vertex}'' step would apply. 
    If there were an edge $xy \in E(H)$ with empty $\eta(xy)$, $\eta(xy,x)$, or $\eta(xy,y)$, the ``\hyperref[mov:EmptyEI]{removing an edge with an empty interface}'' step would apply.
    % \mipil{Why is the second to last sentence needed?}\mpil{Because rigid means nonempty edge sets and interfaces.}
    This concludes the proof.
\end{proof}

Recall that in the context of Lemma~\ref{lem:esd:step-pyramid}, we have access to a rigid extended strip decomposition $(H,\eta)$ of $G-v$
with all particles of weight at most $0.5\tau$. We want to add $v$ to the extended strip decomposition; on the way there
we can identify an induced $\Pttt$ or a $0.99\tau$-balanced separator that is dominated by a small number of vertices. 

If $\w(V(G) \setminus N_G[v]) \leq 0.99\tau$, then we can return $N_G[v]$ as the promised $0.99\tau$ balanced separator.
Hence, we assume $\w(V(G) \setminus N_G[v]) \geq 0.99\tau$, that is, 
\begin{equation}\label{eq:esd:tau-small}
\tau \leq 0.99^{-1} \w(V(G) \setminus N_G[v]). 
\end{equation}
Let $S \subseteq N_G[v]$ with $v \in S$. 
A \emph{local cleaning operation applied to $S$} consists of the following:
\begin{enumerate}
    \item Computing an extended strip decomposition $(H_S,\eta_S)$ of $G-S$ by restricting each set $\eta(\cdot)$ from $(H,\eta)$ to $V(G) \setminus S$.
    (Note that in this step $(H_S,\eta_S)$ may not be rigid, as some sets $\eta_S(x)$, $\eta_S(xy)$ or interfaces
    $\eta_S(xy,x)$ may be empty.)
    \item Iteratively, while possible, apply the first applicable local cleaning operation to $(H_S,\eta_S)$.
    \item If at any moment of the process there exists a particle of $(H_S,\eta_S)$ whose weight is
    at least $0.01 \w(V(G) \setminus S)$, apply Lemma~\ref{lem:substantial-particle} to it, obtaining
    a $0.99\w(V(G) \setminus S)$-balanced separator $X$ of $G-S$ 
    equal to $\bigcup_{x \in F} \bigcup_{y \in N_{H_S}(x)} \eta_S(xy,x)$ for some $F \subseteq V(H_S)$ of size at most $2$.
    Since $V(H_S) \subseteq V(H)$ and $\eta_S(xy,x) \subseteq \eta(xy,x)$ for every $xy \in E(H_S)$ by Lemma~\ref{lem:clean-step-ok},
    $X \subseteq \bigcup_{x \in F} \bigcup_{y \in N_H(x)} \eta(xy,x)$. Hence, by Lemma~\ref{lem:dom-potato}, 
    $X$ is dominated by at most four vertices in $G$ (not necessarily in $G - S$). 
    By~\eqref{eq:esd:tau-small} and since $\tau \geq \w(G)$, every connected component of $G-S-X$ has weight at most 
    \[ \w(V(G) \setminus S) - 0.01 \w(V(G) \setminus S) = 0.99 \w(V(G) \setminus S) \leq 0.99\tau.\]
    Thus, we return $X \cup S$ as a $0.99\tau$-balanced separator of $G$ dominated by at most five vertices. 
\end{enumerate}
Initially, every particle of $(H_S,\eta_S)$ has weight at most $0.5\tau$ which, by~\eqref{eq:esd:tau-small}, is upper bounded by $\frac{0.5}{0.99}\w(V(G)\setminus N_G[v]) \leq 0.55\w(V(G) \setminus S)$. 
Thus, if Lemma~\ref{lem:substantial-particle} is triggered before the first cleaning operation, its assumptions are satisfied. 
Later in the process, we apply a local cleaning operation to an extended strip decomposition whose every particle
is of weight at most $0.01 \w(V(G) \setminus S)$. 
Since every local cleaning operation moves a subset of one set $\eta(\cdot)$ to another, after a single local cleaning operation
every particle is of weight at most $0.02 \w(V(G) \setminus S)$. This justifies the assumptions
of Lemma~\ref{lem:substantial-particle} if triggered in later steps. 

We conclude with the following straightforward summary of the properties of the result of local cleaning (cf. Lemma~\ref{lem:locally-cleaned-ok}).
\begin{lemma}\label{lem:local-cleaning}
Let $S \subseteq V(G)$ with $v \in S$ and assume that the local cleaning operation applied to $S$ 
finished with an  extended strip decomposition $(H_S,\eta_S)$ of $G-S$. Then:
\begin{itemize}
    \item $(H_S,\eta_S)$ is a rigid extended strip decomposition of $G-S$.
    \item $V(H_S) \subseteq V(H)$ and $E(H_S) \subseteq E(H)$.
    \item For every $x \in V(H_S)$ we have $\eta_S(x) \supseteq \eta(x)$.
    \item For every $xy \in E(H_S)$ we have $\eta_S(xy) \subseteq \eta(xy)$, $\eta_S(xy,x) \subseteq \eta(xy,x)$, and $\eta_S(xy,y) \subseteq \eta(xy,y)$.
    \item For every $xyz \in T(H_S)$ we have $\eta_S(xyz) \subseteq \eta(xyz)$.
    \item Every particle of $(H_S,\eta_S)$ is of weight at most $0.01 \w(V(G) \setminus S) \leq 0.01 \tau$.
    \item Every vertex of $H_S$ is of degree at least $2$.
\end{itemize}
\end{lemma}

We start the algorithm of Lemma~\ref{lem:esd:step-pyramid} with applying the local cleaning operation to $(H,\eta)$, that is, 
the case $S = \{v\}$.
We either return a $0.99\tau$-balanced separator dominated by at most five vertices
or (we reuse the name $(H,\eta)$ for the obtained extended strip decomposition, slightly abusing the notation)
ensure the following two properties.
\begin{equation}\label{eq:esd:H-small-particles}
\mathrm{Every\ particle\ of\ }(H,\eta)\mathrm{\ is\ of\ weight\ at\ most\ }0.01\tau.
\end{equation}
\begin{equation}\label{eq:esd:H-degree-two}
\mathrm{Every\ }x \in V(H)\mathrm{\ is\ of\ degree\ at\ least\ }2.
\end{equation}

We henceforth proceed assuming~\eqref{eq:esd:H-small-particles} and~\eqref{eq:esd:H-degree-two}.
Note that Lemma~\ref{lem:dom-potato} implies the following.
\begin{equation}\label{eq:esd:H-dom-separator}
\mathrm{For\ every\ }F \subseteq V(H)\mathrm{\ the\ set\ }\bigcup_{x \in F} \bigcup_{y \in N_H(x)} \eta(xy,x)\mathrm{\ is\ dominated\ by\ at\ most\ }2|F|\mathrm{\ vertices\ in\ }G.
%\forall_{F \subseteq V(H)} \exists_{Z \subseteq V(G)} \left(|Z| \leq 2|F|\right) \wedge \left(\bigcup_{x \in F} \bigcup_{y \in N_H(x)} \eta(xy,x) \subseteq N_G[Z]\right).
\end{equation}

\newcommand{\tanglethreshold}{\sigma}

\subsection{A wall avoiding $N[v]$}

An edge $xy \in E(H)$ is \emph{$v$-safe} if 
there exists a path in $G[\eta(xy)]-N[v]$ between a vertex of $\eta(xy,x)$ and a vertex of $\eta(xy,y)$.
A subgraph of $H$ is \emph{$v$-safe} if all its edges are $v$-safe.

We now show the following.
\begin{lemma}\label{lem:find-safe-wall}
For every constant $\tanglethreshold$ there exists a constant $\tanglethreshold'$ such that
in polynomial time we can either find a $0.99\tau$-balanced separator dominated by at most $\tanglethreshold'$ vertices in $G$
or a $v$-safe wall $W$ in~$H$ of sidelength $3\tanglethreshold$ with $\mathcal{T}_W$, the tangle governed by $W$, having the following property:
% \TM{I would suggest having something else than $\tau$ for the threshold as the letter is similar to $\mathcal{T}$. But if I'm the only one, I'm fine with that.}
\begin{equation}\label{eq:W-tangle}
\forall_{(A,B) \in \mathcal{T}_W} \w(\preimage{(H,\eta)}{B \setminus A}) \geq 0.99\tau.
\end{equation}
\end{lemma}
\begin{proof}
  % Let $S \coloneqq N[v]$.
  Apply the local cleaning operation to $(H_S,\eta_S)$ where $S \coloneqq N[v]$. 
If a $0.99\tau$-balanced separator $X$ is found, return $X \cup S$ as the promised $0.99\tau$-balanced separator.
Otherwise, we have an extended strip decomposition $(H_S,\eta_S)$  of $G-N[v]$
that satisfies the properties of Lemma~\ref{lem:local-cleaning}.

Let $k \coloneqq f_{\mathrm{KTW20}}(\tanglethreshold)$
where $f_{\mathrm{KTW20}}$ comes from Theorem~\ref{thm:tangle2wall}. 
Note that $k$ is still a constant.

Assume that there exists a separation $(A,B)$ in $H_S$ of order less than $k$ with
both \[\w(\preimage{(H_S,\eta_S)}{A \setminus B}) \leq 0.99\tau \quad \text{  and  } \quad \w(\preimage{(H_S,\eta_S)}{B \setminus A}) \leq 0.99\tau.\]
By Lemma~\ref{lem:sepH2G}, the set $X \coloneqq \bigcup_{x \in A \cap B} \bigcup_{y \in N_{H_S}(x)} \eta_S(xy,x)$ is
a $0.99\tau$-balanced separator of ${G-S}$.
Due to Lemma~\ref{lem:local-cleaning}, we have $X \subseteq \bigcup_{x \in A \cap B} \bigcup_{y \in N_H(x)} \eta(xy,x)$, which
is in turn dominated by at most ${2|A \cap B| \leq 2k-2}$ vertices in $G$ due to Lemma~\ref{lem:dom-potato}. 
Since $S = N[v]$, $G$ admits a $0.99\tau$-balanced separator dominated by at most $2k-1$ vertices. 
Since $k$ is a constant, we can find such a separator in polynomial time and return it.
Thus, henceforth we assume that such a separation $(A,B)$ does not exist. 

Let $\mathcal{T}'$ be a set consisting of every separation $(A,B)$ in $H_S$ of order less than $k$ with 
\[\w(\preimage{(H_S,\eta_S)}{B \setminus A}) \geq 0.99\tau.\]
Since $\w(G) \leq \tau$ and due to our assumption from the previous paragraph,
for every separation $(A,B)$ in $H_S$ of order less than $k$, exactly one of $(A,B)$ and $(B,A)$ belongs to $\mathcal{T}'$. 
Also, for every $(A,B) \in \mathcal{T}'$ it holds that $\w(\preimage{(H_S,\eta_S)}{A \setminus B}) \leq 0.01\tau$.

Assume that $\mathcal{T}'$ is not a tangle of order $k$. Then, there exist $(A_1,B_1),(A_2,B_2),(A_3,B_2) \in \mathcal{T}'$
with $A_1 \cup A_2 \cup A_3 = V(H_S)$. 
Let $F = (A_1 \cap B_1) \cup (A_2 \cap B_2) \cup (A_3 \cap B_3)$. We have $|F| \leq 3k-3$. 
Let $X = \bigcup_{x \in F} \bigcup_{y \in N_{H_S}(x)} \eta_S(xy,y)$.
Since $\w(\preimage{(H_S,\eta_S)}{A_i \setminus B_i}) \leq 0.01\tau$ for $i=1,2,3$ and every particle of $(H_S,\eta_S)$
is of weight at most $0.01\tau$, $X$ is a $0.99\tau$-balanced separator in $G-S$.
% \mipil{Is it even a $0.01\tau$-balanced separator?}
% \mpil{I think so, but it's irrelevant, isn't it?}
Similarly as before, Lemmas~\ref{lem:local-cleaning}
and~\ref{lem:dom-potato} imply that $X$ is dominated by at most $6k-6$ vertices in $G$.
Hence, $X \cup S$ is a $0.99\tau$-balanced separator in $G$ dominated by at most $6k-5$ vertices in $G$. 
Since $k$ is a constant, we can check in polynomial time if such a separator exists. 
Thus, henceforth we continue with the assumption that $\mathcal{T}'$ is a tangle of order $k$ in $H_S$.

We now apply Theorem~\ref{thm:tangle2wall} to obtain a wall $W$ in $H_S$ of sidelength $3\tanglethreshold$ such that
$\mathcal{T}_W'$, the tangle of order $\tanglethreshold$ governed by $W$ in $H_S$, is the restriction of $\mathcal{T}'$ to order $\tanglethreshold$.
We observe that as $\tanglethreshold$ is a constant, $W$ can be computed in polynomial time for example by
first guessing its pegs, and then applying an algorithm for \textsc{Disjoint Paths} of~\cite{GM13}.
Since $E(H_S) \subseteq E(H)$, the wall $W$ exists also in $H$. 
Observe that, due to the local cleaning operation, every edge $xy \in E(H_S)$ is $v$-safe in $H$. 
Hence, $W$ is $v$-safe.

Let $(A,B)$ be a separation in $H$ of order less than $k$.
Since $E(H_S) \subseteq E(H)$, we observe that $(A \cap V(H_S), B \cap V(H_S))$ is a separation in $H_S$ of order less than $k$.
Let then $\mathcal{T}$ be the set of all separations $(A,B)$ of $H$ of order less than $k$ such that
$(A \cap V(H_S), B \cap V(H_S))$ belongs to~$\mathcal{T}'$.
Similarly, let $\mathcal{T}_W$ be the set of all the separations $(A,B)$ of $H$ of order less than $\tanglethreshold$
such that $(A \cap V(H_S), B \cap V(H_S))$ is in $\mathcal{T}_W'$.
Then, $\mathcal{T}$ is a tangle of order $k$ in $H$ that satisfies
\[
\forall_{(A,B) \in \mathcal{T}} \;\; \w(\preimage{(H_S,\eta_S)}{B \setminus A}) \geq 0.99\tau.
\]
Moreover, $\mathcal{T}_W$ is the tangle of order $\tanglethreshold$ governed by $W$ in $H$ and
is equal to the restriction of $\mathcal{T}$ to order~$\tanglethreshold$ satisfying~\eqref{eq:W-tangle}.
\end{proof}

We fix the wall $W$ obtained via Lemma~\ref{lem:find-safe-wall} for the remainder of the proof.
In subsequent steps we are going to obtain more and more structural properties of $(H,\eta)$ and $W$, at the cost
of gradually shrinking $W$. The actual value of $\tanglethreshold$ will be fixed at the end of the proof,
so that the final remnants of $W$ are still substantial.

\subsection{Finding a $\Pttt$ in a $v$-safe wall}
\newcommand{\Projection}{\Pi}
\newcommand{\StartV}[1]{\overleftarrow{\mathbf{u}}(#1)}
\newcommand{\EndV}[1]{\overrightarrow{\mathbf{u}}(#1)}
\newcommand{\IthV}[2]{\mathbf{u}_{#1}(#2)}

In what follows, we will be thinking of every path $P$ as a path with an orientation, so that $P$ has a starting
vertex $\StartV{P}$ and an ending vertex $\EndV{P}$, and for an integer $1 \leq i \leq |V(P)|$ we can speak
of the $i$-th vertex $\IthV{i}{P}$ of a path. Clearly, $\StartV{P} = \IthV{1}{P}$ and $\EndV{P} = \IthV{|V(P)|}{P}$. Additionally, we will call a set of paths $P_1, P_2, \ldots, P_i$ {\em almost vertex-disjoint} if they all share the same final vertex and are otherwise vertex-disjoint. More formally, if there is a vertex $u$ such that for all $1 \leq j \leq i$, $\EndV{P_j} = u$, and the set of paths $P_1 - u, P_2 - u, \ldots, P_i-u$ is vertex-disjoint set of paths, then we the paths are almost vertex-disjoint.

For a moment, let us get out of the context of the proof of Lemma~\ref{lem:esd:step} and introduce an auxiliary tool for finding long almost disjoint paths in some parts of $H$ and subwalls of $W$.

\begin{lemma}[Finding paths of length $t$ in a wall]\label{lem:pathInWall}
    Let $\widetilde{W}$ be a $2t$-subdivided wall of sidelength at least $3$ in a graph $\widetilde{H}$. 
    Let $P_1$, $P_2$, $P_3$ be vertex-disjoint paths such that $\EndV{P_i}$ is a peg of $\widetilde{W}$ for $i=1,2,3$.
    Let $\widetilde{H}'$ be the subgraph of $\widetilde{H}$ consisting of $\widetilde{W}$ and all paths $P_i$, $i=1,2,3$.

    Then, there exist three almost vertex-disjoint paths $P_1'$, $P_2'$, $P_3'$ in $\widetilde{H}'$ such that for every $i=1,2,3$ there
    exists an integer $1 \leq k_i \leq \min(|V(P_i')|,|V(P_i)|)$ so that
    \begin{itemize}
        \item the prefix of $P_i'$ up to $\IthV{k_i}{P_i'}$ equals the prefix of $P_i$ up to $\IthV{k_i}{P_i}$; and
        \item the suffix of $P_i'$ from $\IthV{k_i}{P_i'}$ is an induced path in $\widetilde{W}$ of length greater than $t$.
    \end{itemize}
\end{lemma}
\begin{proof}
    A \emph{finishing touch} for a vertex $v$ in $\widetilde{W}$ is a path defined as follows:
    \begin{itemize}
        \item if $v$ is a peg of $\widetilde{H}$, then a finishing touch is a zero-length path consisting of $v$ only;
        \item otherwise, if $Q$ is the basic path of $\widetilde{W}$ containing $v$,
          then a finishing touch of $v$ is a subpath of $Q$ between $v$ and one of its endpoints (which is always a peg). Note that if $v$ is not a peg then $v$ has two finishing touches.
    \end{itemize}
    Thus, a finishing touch connects $v$ with a peg of $\widetilde{W}$, without any other peg on the way. 
    
    Consider the following modification step. Let $i \in \{1,2,3\}$ and assume $P_i$ contains a vertex $v$
    such that the suffix of $P_i$ starting in $v$ is not a finishing touch for $v$, but there is a finishing touch of $v$
    whose intersection with $V(P_1) \cup V(P_2) \cup V(P_3)$ is contained in the suffix of $P_i$ starting at $v$.
    Then, modify $P_i$ by replacing the suffix starting at $v$ with the said finishing touch of $v$. In the case where this can be done with either finishing touch, choose the shorter one, so that the other finishing touch (the one not chosen) contains over $t$ vertices.   
    Note that if we find the paths $P_1'$, $P_2'$, and $P_3'$ for the modified paths, then they are also good for the original
    paths, as one can assume that $\IthV{k_i}{P_i'}$ is not later than $v$ on $P_i'$. 

    As this modification either strictly decreases the number of edges of $\widetilde{H}'$ that are not in $\widetilde{W}$
    or strictly decreases the number of pegs on the paths $P_1$, $P_2$, $P_3$ while keeping $\widetilde{H}'$ intact,
    without loss of generality we can assume that the modification step is not possible.
    This, in particular, implies that for $i=1,2,3$, the only peg on $P_i$ is $\EndV{P_i}$.

    Let $Q$ be a basic path in $\widetilde{W}$.
    Assume that there is an internal vertex of $Q$ which is on one of the paths $P_i$. We claim that either
    \emph{both} endpoints of $Q$ are endpoints of two out of three paths $P_1$, $P_2$, $P_3$, or 
    the intersection of $Q$ with the union of paths $P_1$, $P_2$, $P_3$ is a suffix of one of those paths and is equal to a finishing touch of some internal vertex, $v$, of $Q$ and is the shorter of $v$'s two finishing touches (and therefore has at most $t$ vertices).

    To this end, assume that an endpoint $u$ of $Q$ is not an endpoint of any of the paths $P_i$, $i=1,2,3$. 
    Since the endpoints are the only pegs on the paths $P_i$, $i=1,2,3$, $u$ does not lie on either of the paths $P_i$, $i=1,2,3$.
    Let $w$ be the closest to $u$ vertex of $Q$ that lies on one of the paths $P_i$, $i=1,2,3$, and let $i \in \{1,2,3\}$ be such
    that $w$ lies on $P_i$. By our assumption, $w$ is an internal vertex of $Q$.

    Note that the modification step is applicable and we could replace the suffix of $P_i$ starting at~$w$ with the finishing touch being the subpath of $Q$ from $w$ to $u$. The only reason this modification step is invalid is that the suffix of $P_i$ starting from $w$ is the second (and shorter) finishing touch of~$w$, that is, the suffix is the subpath of $Q$ from $w$ to the second endpoint, and that finishing touch is the shorter of the two. This proves the claim.

    %Hence, if $Q$ is a basic path in $\widetilde{W}$ such that exactly one of its endpoints is an endpoint of one of the paths $P_1, P_2, P_3$, denote this path as $P_j$, then since $Q$ has at least $2t$ vertices, the intersection of $P_j$ with $Q$ is a suffix of $P_j$ that is a finishing touch for a vertex, $x$, of $Q$ and the other finishing touch of $x$ contains over $t$ vertices. 

    It is straightforward to verify that for any wall, $H$, of sidelength at least 3 and for any three pegs, $x,y,z$ of the wall, there is a connected component, $C$, of $H-\{x,y,z\}$ such that $x,y,z \in N(C)$. Hence, there exists a connected component, $C$, of $\widetilde{W} - \{ \EndV{P_1}$, $\EndV{P_2}$, $\EndV{P_3}\}$ such that $\EndV{P_1}$, $\EndV{P_2}$, $\EndV{P_3}$ $\in N(C)$. Let $P$ be a shortest path from $\EndV{P_1}$ to $\EndV{P_2}$ with internal vertices in $C$, and let $P'$ be $P$ along with a shortest path from $\EndV{P_3}$ to an internal vertex of $P$, with internal vertices in $C$.

    First, we show that $P'$ is a $S_{t_1,t_2,t_3}$, where $t_1,t_2,t_3 > 2t$, with $\EndV{P_1}$, $\EndV{P_2}$, and $\EndV{P_3}\}$ as its leaves. That $P'$ is not a path follows from the fact that none of $\EndV{P_1}$, $\EndV{P_2}$, nor $\EndV{P_3}$ belong to $C$. Furthermore, pegs of $\widetilde{W}$, that is all vertices of degree 3, are of distance greater than $2t$ from one another in $\widetilde{W}$ by the definition of a $2t$-subdivided wall. It follows that $P'$ must be a $S_{t_1,t_2,t_3}$. 

    Let $c$ be the center of the claw (degree 3 vertex) that is $P'$ and for the leaves $\EndV{P_i}$, $i=1,2,3$, let $P_i^*$ be the path of the claw from $\EndV{P_i}$ to $c$. Let $Q_i$ denote the basic path of $\widetilde{W}$ contained in $P_i^*$ that has $\EndV{P_i}$ as one of its endpoints. Note that it follows from our previous claim that because the only vertex of $P_i^*$ that is the endpoint of $P_1, P_2,$ or $P_3$ is $\EndV{P_i}$, we may conclude that for $j \neq i$, $P_j$ does not intersect $P_i^*$. Furthermore, from the same claim it follows that $P_i$ only intersects $P_i^*$ in $Q_i$ and this intersection is a suffix of $P_i$, which we denote as $R_i$, and $R_i$ is a finishing touch for vertex $v_i$ of $Q_i$ where the subpath of $Q_i$ from $v_i$ to the other endpoint of $Q_i$ length greater than $t$. Finally, to obtain the desired path $P_i'$ to complete the proof, we replace $R_i$ in $P_i$ with the subpath of $P_i^*$ from $v_i$ to $c$, which has length greater than $t$. This completes the proof.    
\end{proof}

Now we get back to the context of the proof of Lemma~\ref{lem:esd:step-pyramid}, i.e., we work with an extended strip decomposition $(H,\eta)$ of $G-v$. First, consider a collection $P_1, P_2, \ldots, P_i$ of induced paths, where path $P_j$ has endpoints $x_j, y_j$. We say this collection is a {\em pyramid base} if the paths $P_1-y_1, P_2-y_2, \ldots, P_i-y_i$ are pairwise anti-complete, and $y_1, y_2, \ldots, y_i$ form a clique of size $i$.
We make the following simple observation that we will later use multiple times to find $\Pttt$s within our proofs.
\begin{lemma}\label{lem:wall-path-to-real-path}
%Let $\widetilde{G}$ be a graph and let $(\widetilde{H},\widetilde{\eta})$ be an extended strip decomposition of $G-v$.
Let $\widetilde{H}$ be a $v$-safe subgraph of $H$ and let $\mathcal{P}$ be a family of almost vertex-disjoint paths in $\widetilde{H}$, each of length at least one. 
Then one can find a family $\mathcal{Q} = \{Q_P~|~P \in \mathcal{P}\}$ of induced paths in $G-N[v]$ which form a pyramid base and for every $P \in \mathcal{P}$, the path $Q_P$:
\begin{itemize}
    \item is of length at least $|V(P)|-2$;
    \item starts in a vertex of $\eta(\StartV{P}\IthV{2}{P}, \StartV{P})$;
    \item ends in a vertex of $\eta(\EndV{P}\IthV{|V(P)|-1}{P}, \EndV{P})$; and
    \item has all internal vertices contained in 
    \[ \bigcup_{e \in E(P)} \eta(e) \setminus (\eta(\StartV{P}\IthV{2}{P}, \StartV{P}) \cup \eta(\EndV{P}\IthV{|V(P)|-1}{P}, \EndV{P})). \]
\end{itemize}
\end{lemma}
\begin{proof}
    Fix $P \in \mathcal{P}$. Since $\widetilde{H}$ is $v$-safe, for every $1 \leq i < |V(P)|$ there exists a path 
    $Q_{P,i}$ in $G[\eta(\IthV{i}{P}\IthV{i+1}{P})]-N[v]$ with endpoints in $\eta(\IthV{i}{P}\IthV{i+1}{P}, \IthV{i}{P})$
    and in $\eta(\IthV{i}{P}\IthV{i+1}{P}, \IthV{i+1}{P})$ and all internal vertices
    in $\eta(\IthV{i}{P}\IthV{i+1}{P}) \setminus (\eta(\IthV{i}{P}\IthV{i+1}{P}, \IthV{i}{P}) \cup \eta(\IthV{i}{P}\IthV{i+1}{P}, \IthV{i+1}{P}))$. 
    By the properties of an extended strip decomposition, the ending vertex of $Q_{P,i}$ is adjacent to the starting vertex
    of $Q_{P,i+1}$ for $1 \leq i < |V(P)|-1$. Thus, the concatenation of those paths gives a path $Q_P$ in $G-N[v]$
    with the starting and ending vertices placed as desired and with at least $|V(P)|-1$ vertices.
    Finally, the properties of an extended strip decomposition, together with the assumption that the paths $\mathcal{P}$
    are almost vertex-disjoint imply that the paths $\{Q_P~|~P \in \mathcal{P}\}$ form a pyramid base. 
    This completes the proof.
\end{proof}

By the properties of an extended strip decomposition, every connected component of $G-v-\bigcup_{e \in E(H)} \eta(e)$ lies
in a single set $\eta(x)$ for some $x \in V(H)$ or in a single set $\eta(xyz)$ for some $xyz \in T(H)$.
We will be thinking of vertices that are reachable from $v$ without visiting $\bigcup_{e \in E(H)} \eta(e)$ as vertices
close to $v$ in the following sense.
\begin{definition}[projection]
The \emph{projection of $v$}, denoted $\Projection$, is the set of those vertices $u \in \bigcup_{e \in E(H)} \eta(e)$
for which there exists a path $P^\Projection_u$ with endpoints $v$ and $u$ and no internal vertex in $\bigcup_{e \in E(H)} \eta(e)$. 
\end{definition}
Note that $P^\Projection_u$ is either a single edge,
or a path whose all internal vertices lie in a single set $\eta(x)$ for some $x \in V(H)$, or in a single set $\eta(xyz)$ for some $xyz \in T(H)$. In particular, $\Projection \cap N[v] = \bigcup_{e \in E(H)} \eta(e) \cap N[v]$.

We need one more tool that will help us exhibit induced $\Pttt$s.
\begin{lemma}\label{lem:connect-via-projection}
Let $P$ be a path in $H$ of length at least one with $x = \StartV{P}$, $y = \IthV{2}{P}$, and $\Projection \cap \eta(xy) \neq \emptyset$.

Then there is a path $Q$ in $G$ that starts in $v$, ends in a vertex of $\eta(\EndV{P}\IthV{|V(P)|-1}{P}, \EndV{P})$,
and has all internal vertices in
\[ \eta(x) \cup \eta(y) \cup \bigcup_{e \in E(P)} \eta(e) \cup \bigcup_{z \in V(H)~|~xyz \in T(H)} \eta(xyz). \]
Furthermore, if $\Projection \cap (\eta(xy) \setminus \eta(xy,x)) \neq \emptyset$, then $Q$ can be chosen with all internal vertices in
\[ \eta(y) \cup \bigcup_{e \in E(P)} \eta(e). \]
\end{lemma}
\begin{proof}
    Pick $u \in \Projection \cap \eta(xy)$, preferably not in $\eta(xy,x)$ if possible. 
    The path $P^\Projection_u$ is either a direct edge, has all internal vertices in $\eta(y)$, all internal vertices
    in $\eta(x)$, or all internal vertices in $\eta(xyz)$ for some $z \in V(H)$ such that $xyz \in T(H)$.
    Furthermore, if $u \notin \eta(xy,x)$, then the last two options are impossible.

    Since $(H,\eta)$ is locally cleaned, in particular, the ``\hyperref[mov:DisconnectedVI]{moving a disconnected vertex of an interface}'' and ``\hyperref[mov:DisconncetedES]{moving a disconnected component of an edge set}'' steps are inapplicable, there is a path $Q_u$ in $\eta(xy)$ from $u$ to a vertex $w \in \eta(xy,y)$
    with all internal vertices not in $\eta(xy,x)$. 
    By the properties of a locally cleaned extended strip decomposition, there is a path $R_u$ from $w$ to a vertex
    of $\eta(\EndV{P}\IthV{|V(P)|-1}{P}, \EndV{P})$ via $\bigcup_{e \in E(P) \setminus \{xy\}} \eta(e)$.
    By concatenating $P^\Projection_u$, $Q_u$, and $R_u$, and possibly shortcuting it to an induced path we obtain the desired
    path $Q$.
\end{proof}

We are ready to find our first $\Pttt$ of the proof.
\begin{lemma}\label{lem:filter-projection-from-wall}
For every constant $t_1$ there is a constant $t_2$ such that if $W$ is a $v$-safe wall in $H$
of sidelength at least $t_2$, then either $G$ contains an induced $\Pttt$ or
there exists a subwall $W'$ of $W$ of sidelength $t_1$
such that for every $e \in E(H)$ with $\Projection \cap \eta(e) \neq \emptyset$, at most one endpoint of $e$
lies in~$W'$.
\end{lemma}
\begin{proof}

    Consider the natural plane embedding of the wall $W$ as given in Figure \ref{fig:wall}. Assume that the columns and rows of $W$ have been sequentially labeled 1 through $t_2$. Remove rows $t_2-t$ through $t_2$ of $W$, if it now holds that every $e \in E(H)$ with $\Projection \cap \eta(e) \neq \emptyset$, at most one endpoint of $e$ lies in $W$, then we are done, so assume that this does not happen. So there is an $x_1$ and $y_1$ in $W$ such that $\Projection \cap \eta(x_1y_1) \neq \emptyset$. Let $C_{x_1}$ and $C_{y_1}$ be the columns of $W$ that $x_1$ and $y_1$ lie in. If $x_1$ and $y_2$ lie on the same column, then choose $x_1$ to be the vertex that is closer to row $t_2$ in $W$. Remove $C_{x_1}$ and $C_{y_1}$ from $W$. Again, if it now holds that every $e \in E(H)$ with $\Projection \cap \eta(e) \neq \emptyset$, at most one endpoint of $e$ lies in $W$, then we are done, so assume that this does not happen. Repeat this process to find a $x_2$, $y_2$ and $x_3$, $y_3$.

    Now consider the original wall $W$, without any of the columns or rows removed. By how $x_1$, $x_2$, and $x_3$ were chosen, there exists paths $P_1$, $P_2$, and $P_3$ each of length at least $t+1$ contained in columns $C_{x_1}$, $C_{x_2}$, $C_{x_3}$ from $x_1$, $x_2$, and $x_3$ respectively down to row $t_2$ which are vertex-disjoint and do not contain $y_1$, $y_2$, nor $y_3$. Assume, without loss of generality, that the end vertex of $P_2$ that lies in row $t_2$, call it $z_2$, (which must be a peg) is in-between the end vertices of $P_1$ and $P_3$ that lie in row $t_2$. Extended $P_1$ and $P_3$ to reach $z_2$ using vertices in row $t_2$. Now $P_1$, $P_2$, and $P_3$ are three almost vertex-disjoint paths of lengths at least $t+1$.

    Now, for $i=1,2,3$, proceed as follows. 
    Apply Lemma~\ref{lem:connect-via-projection} to a path consisting of the edge $x_iy_i$ only,
    obtaining a path $Q_i$ from $v$ to a vertex $u_i \in \eta(x_iy_i, x_i)$ with all internal vertices in $\eta(x_i) \cup \eta(y_i) \cup (\eta(x_iy_i) \setminus \eta(x_iy_i,x_i)) \cup \bigcup_{z \in V(H)~|~x_iy_iz \in T(H)} \eta(x_iy_iz)$.
    Since the vertices $x_1,y_1,x_2,y_2,x_3,y_3$ are pairwise distinct, the paths
    $\{Q_i-\{v\}~|~1 \leq i \leq 3\}$ are anti-adjacent.
    Hence, $V(Q_1) \cup V(Q_2) \cup V(Q_3)$ induce a tree with three leaves $u_1,u_2,u_3$ and $v$ being the unique vertex of
    degree $3$. 
    We now extend this tree with paths $Q_{P_i}$ for $i=1,2,3$, obtained from paths $P_i$, $i=1,2,3$, using Lemma~\ref{lem:wall-path-to-real-path} for $\widetilde{H}=W$. This gives the desired $\Pttt$.
\end{proof}

A wall $W'$ in $H$ is \emph{$v$-pure} if it is $v$-safe, it is $2t$-subdivided, and for every $e \in E(H)$
with $\eta(e) \cap \Projection \neq \emptyset$, at most one endpoint of $e$ lies in $W'$.

The following statement follows directly from Lemma~\ref{lem:filter-projection-from-wall} and the fact that by leaving only every $(2t+1)$-th column and row of a wall, we can extract a $2t$-subdivided subwall.
\begin{lemma}\label{lem:safe-to-pure}
For every constant $t_1$ there exists a constant $t_2$ such that if $W$ is a $v$-safe wall in $H$
of sidelength at least $t_2$, then either $G$ contains an induced $\Pttt$ or $W$ contains a $v$-pure
subwall $W'$ of sidelength $t_1$.
\end{lemma}
%Note that since all walls in Lemma~\ref{lem:safe-to-pure} have constant size,
%there is only a constant number of subwalls to consider and the computability statement is trivial.

\subsection{The case of $\Projection$ being well-connected to a $v$-pure wall}\label{sub:noSmallCut}

Lemma~\ref{lem:safe-to-pure} allows us to find a large $v$-pure wall $W$ in $H$. 
We now observe that if there is a substantial connection between edges of $H$ containing elements of $\Projection$
in their sets and $W$, then $G$ admits an induced $\Pttt$.

\begin{lemma}[Pyramid rooted at $v$ without a small cut]\label{lem:clawAtv}
Let $W$ be a $v$-pure wall in $H$ of sidelength at least $3$.
Assume that $H$ contains three vertex-disjoint paths $P_1$, $P_2$, and $P_3$ such that for every $i=1,2,3$,
the first edge $e_i$ of $P_i$ is such that $\eta(e_i) \cap \Projection \neq \emptyset$ and the ending vertex of $P_i$
is a peg of~$W$. Then $G$ contains an induced $\Pttt$. 
\end{lemma}
\begin{proof}
Apply Lemma~\ref{lem:pathInWall} to $P_1$, $P_2$, $P_3$, and $W$, obtaining paths $P_1'$, $P_2'$, $P_3'$.
For every $i=1,2,3$ define the path $P_i''$ as follows. If $e_i$ is the first edge of $P_i'$ (i.e., $k_i > 1$ from Lemma~\ref{lem:pathInWall}), 
then keep $P_i'' \coloneqq P_i'$. 
The other case is only possible if $e_i=x_iy_i$ with $x_i \in V(W)$, $y_i \notin V(W)$, $x_i = \StartV{P_i}$,
and the path $P_i'$ is fully contained in $W$.
Note that, because $P_j$ and $P_j'$ differ only inside $W$, the vertex $x_i$ is not used by any other path $P_j'$. 
Define $P_i''$ to be the path $P_i'$ prepended with the edge $e_i$ (so now $\StartV{P_i''} = y_i$).
(Note that $e_i$ cannot be an edge of $W$ as $y_i \notin V(W)$.)
In this manner, $(P_i'')_{i=1,2,3}$ are almost vertex-disjoint, each starts with $e_i$ and contains a suffix of length at least $t$ contained in the wall~$W$. 

Split each $P_i''$ into the said suffix $R_i$ of length $t$ and the remaining prefix $Q_i$. 
Lemma~\ref{lem:wall-path-to-real-path} applied to $\{R_i~|~1 \leq i \leq 3\}$ and $\widetilde{H}=W$
induced paths $R_i'$ on at least $t$ vertices each which forms a pyramid base.
Lemma~\ref{lem:connect-via-projection} applied to $Q_i$ gives a path $Q_i'$ from $v$ to a vertex of
$\eta(\EndV{Q_i}\IthV{|V(Q_i)|-1}{Q_i}, \EndV{Q_i})$. Because the paths $P_i''$ for $i=1,2,3$ are almost vertex-disjoint,
the paths $Q_i'-\{v\}$ for $i=1,2,3$ are anti-adjacent. 
Then, $Q_1' \cup Q_2' \cup Q_3' \cup R_1' \cup R_2' \cup R_3'$ contains an induced $\Pttt$ with $v$ as the apex (the unique degree three vertex in $\Pttt$ with an independent neighborhood).
\end{proof}

% \mpil[inline]{Important modification: I think this lemma still works if $x \in V(W)$ and it is easier later
% if we allow it.}

\begin{lemma}[Pyramid in wall rooted at $v$ where paths in $H$ start at a single vertex]\label{lem:clawAtv2}
Let $W$ be a $v$-pure wall of sidelength at least $3$.
Assume that $H$ contains three vertex-disjoint paths $P_1$, $P_2$, $P_3$ and a vertex $x \in V(H)$ such that 
for every $i=1,2,3$ the ending vertex of $P_i$ is a peg of $W$ while $x\StartV{P_i} \in E(H)$
and $\Projection \cap (\eta(x\StartV{P_i}) \setminus \eta(x\StartV{P_i},x)) \neq \emptyset$. 
Then, $G$ contains an induced $\Pttt$. 
\end{lemma}
\begin{proof}
    We start by applying Lemma~\ref{lem:pathInWall} to $P_1$, $P_2$, and $P_3$, obtaining almost vertex-disjoint paths $P_1'$, $P_2'$, and $P_3'$
    and indices $k_1$, $k_2$, and $k_3$.
    We remark that $x$ may appear on one of the paths $P_i'$ and even one of those paths can start with the edge $x\StartV{P_i}$.

    Fix $i \in \{1,2,3\}$. Denote $y_i = \StartV{P_i} = \StartV{P_i'}$.
    Lemma~\ref{lem:connect-via-projection}, applied to the single edge $xy_i$
    gives a path $Q_i$ from $v$ to $v_i \in \eta(y_ix,y_i)$ with all 
    internal vertices in $\eta(y_i) \cup (\eta(xy_i) \setminus (\eta(xy_i,x) \cup \eta(xy_i,y_i)))$.
    Observe that because vertices $y_i$ are pairwise distinct,
    the paths $Q_i-\{v\}$ are pairwise disjoint and anti-adjacent for $i=1,2,3$.

    If $k_i > 1$, extend $Q_i$ from $v_i$ to a vertex of $\eta(\IthV{k_i}{P_i'}\IthV{k_i-1}{P_i'}, \IthV{k_i}{P_i'})$
    via sets $\eta(e)$ for $e$ lying on the prefix of $P_i'$ till $\IthV{k_i}{P_i'}$,
    and shorten the walk to an induced path in the end.
    % \tm{It does not need to be a path: in case  that we used $\eta(xy_i) \setminus (\eta(xy_i,x) \cup \eta(xy_i,y_i))$ to get to $v_i$ and $P_1$ start with $y_1x$ edge, we might be forced to reuse the vertices. But we can as well shortcut to an induced path, so the proof is OK.}
    Let $Q_i'$ be the resulting path; set $Q_i' = Q_i$ if $k_i=1$ and observe that if $k_i=1$ then the first edge of $P_i'$ is not $xy_i$
    as $xy_i \notin E(W)$ due to $W$ being $v$-pure. 
    To obtain the desired $\Pttt$ with apex $v$, extend every path $Q_i'$ with a path obtained
    from Lemma~\ref{lem:wall-path-to-real-path} applied to the suffix of $P_i'$ from $\IthV{k_i}{P_i'}$ and $\widetilde{H}=W$. 
\end{proof}

We will use Lemmas~\ref{lem:clawAtv} and~\ref{lem:clawAtv2} to find a separation that separates $W$
from $\Projection$. The precise meaning of ``separate'' is encapsulated in the following statement.

\begin{definition}[capturing a projection]
Let $(A,B)$ be a separation in $H$ and let $Z \subseteq A \cap B$. 
We say that \emph{$(A,B)$ captures $\Projection$ with backdoor set $Z$}
if for every $e \in E(H)$ with $\Projection \cap \eta(e) \neq \emptyset$, either
$e \subseteq A$ or there is an endpoint $x \in e$ such that $x \in Z$ and $\Projection \cap \eta(e) \subseteq \eta(e,x)$. 
\end{definition}

\begin{lemma}\label{lem:capture-projection}
Let $W$ be a $v$-pure wall in $H$ of sidelength at least $55$. 
Then either $G$ contains an induced $\Pttt$ or there exists a separation $(A,B) \in \mathcal{T}_W$ of order less than $19$
that captures $\Projection$ with a backdoor set of size at most $2$.
\end{lemma}
\begin{proof}
    Let 
    \[ X_1 = \{x \in V(H)~|~\exists_{e_x = xy_x \in E(H)} \Projection \cap (\eta(e_x) \setminus \eta(e_x,y_x)) \neq \emptyset \}. \]
    For each $x \in X_1$ fix one edge $e_x$ (and thus also the endpoint $y_x$) as in the definition above.
    Let
    \[ F = \{e = xy \in E(H)~|~x,y \notin X_1 \textrm{ and } \Projection \cap \eta(e,x) \cap \eta(e,y) \neq \emptyset \}.\]
    Let $F_2$ be a maximal matching in $F$ and let $X_2 \coloneqq V(F_2)$.
    For an edge $e = xy \in F_2$ we denote $e_x = e_y = e$.
    Furthermore, for every $x \in X_2$, if $xy$ is the unique edge of $F_2$ containing $x$, then
    we denote $y_x = y$. 
    Let $X \coloneqq X_1 \cup X_2$. 
    Note that $e_x$ and $y_x$ has been defined for all $x\in X$.
    Observe that, by the definition of $X$,
    \begin{equation}\label{eq:capture-X}
    \forall_{e \in E(H)} \left(\eta(e) \cap \Projection \neq \emptyset\right) \Longrightarrow
    \left(e \subseteq X\right) \textrm{ or } \left( e \cap X = \{x\} \textrm{ and } \eta(e) \cap \Projection \subseteq \eta(e,x)\right).
    \end{equation}

    Assume first that there exists a family $\mathcal{P}$ of 17 vertex-disjoint paths in $H$ with starting points in $X$
    and ending points in pegs of $W$. Let $X' \subseteq X$ be the set of the starting points of 
    the paths in $\mathcal{P}$ and for $x \in X'$ let $P_x \in \mathcal{P}$ be the path starting at $x$.
    We say that $x \in X'$ \emph{kills} $x' \in X'\setminus\{x\}$ if $y_x$ lies on $P_{x'}$.
    Observe that $x$ kills at most one $x'\in X'$ as paths in $\mathcal{P}$ are pairwise vertex-disjoint.

    Consider an auxiliary bipartite graph $K$ with sides being two copies of $X'$ and an edge $(x,x') \in E(K)$ if $x$ kills $x'$. 
    We consider two subcases.
    In the first subcase, $K$ has a matching of size $9$.
    Then, $K$ contains three edges $(x_1,x_1')$, $(x_2,x_2')$, and $(x_3,x_3')$ such that $x_1,x_1',x_2,x_2',x_3,x_3'$ are six pairwise distinct vertices of $X'$.
    Then, the paths $P_{x_i}$ prepended 
    with the edge $e_{x_i}$ are vertex-disjoint and satisfy the requirements of 
    Lemma~\ref{lem:clawAtv}, giving an induced $\Pttt$ in $G$.
    In the other case, $K$ has a vertex cover of size at most $8$. By deleting these vertices from $X'$, we obtain 
    a subset $X'' \subseteq X'$ of size $9$ where no vertex kills another one. 

    Consider now a graph $L$ on the vertex set $X''$ where $xx' \in E(L)$ if $y_x = y_{x'}$. By Ramsey's Theorem,
    $L$ has an independent set of size $3$ or a clique of size $4$.
    If $I$ is an independent set of size $3$ in $L$, then for every $x \in I$ obtain 
    a path $P_x'$ as follows: if $y_x$ does not lie on $P_x$, set $P_x'$ to be $P_x$ prepended with the edge $xy_x$, and
    otherwise set $P_x'$ to be the edge $xy_x$ with the suffix of $P_x$ starting in $y_x$; note that $\{P_x'~|~x \in I\}$
    satisfy the requirements of Lemma~\ref{lem:clawAtv}.
    If $I$ is a clique of size $4$ in $L$, denote by $z$ the vertex that is equal to $y_x$ for every $x \in I$.
    Note that because $F_2$ is a matching, at most one element of $I$ is in $X_2$.
    Hence, $z$ and $\{P_x~|~x \in I \setminus X_2\}$ satisfy the requirements of Lemma~\ref{lem:clawAtv2}.
    This finishes the case where the family of paths $\mathcal{P}$ exists.

    In the other case, by Menger's theorem, there is a separation $(A,B)$ in $H$ of order less than $17$ such that
    $X \subseteq A$ but all pegs of $W$ lie in $B$. 
    By~\eqref{eq:capture-X}, $(A,B)$ captures $\Projection$, but the set of backdoors can be as large as $A \cap B$.
    Our goal is now to modify $(A,B)$ a bit to restrict the set of backdoors.
    To this end, consider a subgraph $H'$ of $H$ with $V(H') = B$ and $e \in E(H')$ if $e \subseteq B \setminus A$
    or $|e \cap A \cap B| = 1$ and $\eta(e) \cap \Projection \neq \emptyset$. 
    
    We consider two subcases. In the first subcase, $H'$ contains a family $\mathcal{Q}$
    of three vertex-disjoint paths from $A \cap B$ to the set of pegs of $W$ that are in $B \setminus A$.
    By the construction of $H'$, each path $Q \in \mathcal{Q}$ starts with a vertex $x \in A \cap B$
    and an edge $e_x = xy_x$, $y_x \in B \setminus A$, and $\Projection \cap \eta(e_x) \neq \emptyset$.
    Then, Lemma~\ref{lem:clawAtv} applied to $\mathcal{Q}$ yields an induced $\Pttt$ in $G$.

    In the second subcase, there is a separation $(A',B')$ in $H'$ of order less than $3$ such that
    $A \cap B \subseteq A'$ while all pegs of $W$ that lie in $B \setminus A$ belong to $B'$. 
    Let $A'' \coloneqq A \cup A'$ and $B'' \coloneqq B' \cup (A \cap B)$. 
    Then, $(A'', B'')$ is a separation in $H$ with $X \subseteq A''$ and all pegs of $W$ lying in $B''$.
    Furthermore, $A'' \cap B'' = (A' \cap B') \cup (A \cap B)$.
    In particular, the order of $(A'',B'')$ is less than $19$ and hence $(A'',B'') \in \mathcal{T}_W$
    as $W$ has sidelength at least $55$ (so $\mathcal{T}_W$ has order at least $19$) and all pegs of $W$ lie in $B''$.

    Consider now $z \in A'' \cap B''$ such that there exists an edge $zy \in E(H)$ with $y \in B'' \setminus A''$ and
    $\Projection \cap \eta(zy) \neq \emptyset$. 
    By~\eqref{eq:capture-X}, $\{y,z\} \cap X \neq \emptyset$. However, as $X \subseteq A$ while $y \notin A$,
    we have $z \in X$, $z \in A$, and, by~\eqref{eq:capture-X} again, $\Projection \cap \eta(zy) \subseteq \eta(zy,z)$. 
    The edge $zy$ belongs to $H'$. Since $y \in B'' \setminus A'' \subseteq B' \setminus A'$, we have $z \in B'$.
    Since $z \in A\cap B \subseteq A'$, we have $z \in A' \cap B'$. 

    Hence, $(A'',B'') \in \mathcal{T}_W$ is a separation of order less than 19
    that captures $\Projection$ with backdoor set $Z \subseteq A' \cap B'$, which is of size at most $2$. 
    This finishes the proof.
\end{proof}

\subsection{Cleaning the backdoors}
Lemma~\ref{lem:capture-projection} allows us to find a separation of small order that captures $\Projection$
with at most two backdoors. Our goal in this section is to further clean the situation with regards to how exactly
the neighbors of $v$ can appear around $\eta(z)$ and $\eta(zx,x)$ for a backdoor vertex $z$. 
On the way there, we will need to sacrifice small parts of the wall $W$ defining the tangle 
or slightly increase the size of the allowed separation.

We start with the following straightforward observation from the definition of capturing $\Projection$
and the fact that $(H,\eta)$ is locally cleaned.
\begin{lemma}\label{lem:triangle-entries}
Let $(A,B)$ be a separation in $H$ that captures $\Projection$ with backdoor set $Z$. Suppose $xyz \in T(H)$ is such that
$\eta(xyz) \cap N(v) \neq \emptyset$. Then either $x,y,z \in A$,
or two of the vertices $x,y,z$ belong to $Z$ and the third one is in $B \setminus A$.
\end{lemma}

\newcommand{\PotatoZ}{\mathfrak{P}_{z}}

We need a few definitions that describe how a separation behaves with respect to the triangles.
Let $(A,B)$ be a separation in $H$ that captures $\Projection$ with backdoor set~$Z$.
\begin{definition}[triangle-safe separation]
We say that $(A,B)$ is \emph{triangle-safe} if for every triangle $T=xyz \in T(H)$ with
$\eta(T) \cap N(v) \neq \emptyset$ either $x,y,z \in A$ or, assuming without loss of generality $y,z \in A \cap B$ and $x \in B \setminus A$, we have that $v$ is complete to both $\eta(xy,y)$ and $\eta(xz,z)$.
\end{definition}

Assume $(A,B)$ is triangle-safe separation in $H$ that captures $\Projection$ with backdoor set~$Z$.
For $z \in Z$, let \[\PotatoZ \coloneqq \bigcup_{x \in N_H(z) \cap (B \setminus A)} \eta(zx,z).\]
% Fix a backdoor vertex $z\in Z$.
Observe that for every $u \in \Projection \cap \PotatoZ$, the
path $P^\Projection_u$ is either a direct edge or goes via $\eta(z)$:
it cannot go through $\eta(xyz)$ for a triangle with say $x \in B \setminus A$ as, thanks to triangle-safeness,
$v$ is then complete to $\eta(xz,z)$, where $u$ resides.
A \emph{$z$-entry point} is a vertex $v' \in \{v\} \cup \eta(z)$ that has a neighbor in $\PotatoZ$
and admits a path $Q^\Projection_{v'}$ from $v$ to $v'$ (it can be of zero length when $v=v'$)
whose internal vertices have no neighbors in $\PotatoZ$.
A backdoor vertex $z\in Z$ is \emph{pure} if every $z$-entry point is complete to $\PotatoZ$. 

In the next two lemmas we first ensure that we have a separation in $H$ that captures $\Projection$ and
is triangle-safe, and then we ensure that all backdoors are pure.
\begin{lemma}[triangle cleanup]\label{lem:clean-triangles}
Let $W$ be a $v$-pure wall in $G$ and let $(A,B) \in \mathcal{T}_W$ be a separation capturing $\Projection$
with a set of backdoors $Z$ of size at most $2$ such that the sidelength of $W$
is $k_W \geq 4|A \cap B| + 7$.
Then, either $G$ admits an induced $\Pttt$ or there exists a subwall $W'$ of $W$ of sidelength at least $k_W - |A \cap B|$ 
and a separation $(A',B') \in \mathcal{T}_{W'}$ in $H$ of order at most $|A \cap B| + 2$ that captures $\Projection$ with a backdoor set
of size at most $2$, $A \subseteq A'$, $B' \subseteq B$, and is triangle-safe.
\end{lemma}
\begin{proof}
If $(A,B)$ is already triangle-safe, we can return $W = W'$ and $(A',B') = (A,B)$, so assume otherwise.
By Lemma~\ref{lem:triangle-entries}, this is only possible if $|Z| = 2$, say $Z = \{z_1,z_2\}$
and there is a triangle $z_1z_2x \in T(H)$ with $\eta(xyz) \cap N(v) \neq \emptyset$ and $x\in B\setminus A$, but
$v$ is not complete to $\eta(zx,z)$ for some $z \in \{z_1,z_2\}$. 
Let us call such a triangle a \emph{violating triangle}.
Let $X$ be the set of those $x \in B \setminus A$ for which $z_1z_2 x$ is a violating triangle.
 
 Since $(A,B) \in \mathcal{T}_W$, $B \setminus A$ contains $k_W - |A \cap B|$ full rows and $k_W - |A \cap B|$
 full columns of $W$; let $W'$ be a subwall of $W$ completely contained in $B \setminus A$ of sidelength
 at least $k_W - |A \cap B|$.

 We consider two cases. In the first case, there are three vertex-disjoint paths $P_1$, $P_2$, $P_3$ from $X$ to 
 the pegs of $W'$ in the graph $H[B \setminus A]$. Let $x_i = \StartV{P_i}$ for $i=1,2,3$.
 In this case, we exhibit an induced $\Pttt$ in $G$. 

 To this end, apply Lemma~\ref{lem:pathInWall} to $P_1$, $P_2$, $P_3$, and $W'$, obtaining paths $P_1'$, $P_2'$, and $P_3'$. 
 We note that as all paths $P_1$, $P_2$, and $P_3$, as well as the wall $W'$ are in $H[B \setminus A]$, the paths 
 $P_1'$, $P_2'$, and $P_3'$ are also contained in $H[B \setminus A]$; in particular, they do not contain vertices $z_1$ nor $z_2$.
 Since $(A,B)$ captures $\Projection$, $H[B \setminus A]$ is $v$-safe. 
 Thus, Lemma~\ref{lem:wall-path-to-real-path} applied to $\{P_1', P_2', P_3'\}$ in $H[B \setminus A]$ yields induced paths $R_1, R_2, R_3$ which form a pyramid base.

 For $i=1,2,3$, let $C_i$ be a component of $G[\eta(x_iz_1z_2)]$ that contains a neighbor of $v$.
 Since $z_1$ and $z_2$ are symmetric so far, as $z_1 z_2 x_2$ is a violating triangle,
 we can assume that $v$ is not complete to $\eta(z_1 x_2, z_1)$;
 pick $y_2 \in \eta(z_1x_2, z_1)$ nonadjacent to $v$. 
 Since $(H,\eta)$ is locally cleaned, $C_1$ has a neighbor $y_1 \in \eta(z_1x_1, z_1) \cap \eta(z_1x_1, x_1)$ and
 $C_3$ has a neighbor $y_3 \in \eta(z_2x_3,z_2) \cap \eta(z_2x_3,x_3)$. 
 (We emphasize the intended lack of symmetry in the choice of $z_1$ vs $z_2$ in the last two sentences.)
 By the properties of an extended strip decomposition, $y_1y_2 \in E(G)$,
 $y_1y_3,y_2y_3 \notin E(G)$, and $C_1$ and $C_3$ are anti-adjacent to $y_2$. 
 Let $Q$ be a shortest path from $y_1$ to $y_3$ via $C_1$, $v$, and $C_3$ 
 (note that it may go via direct edges $vy_1$ or $vy_3$ if they exist). 
 Let
 % $R_1'$ be a shortest path
 % from $y_1$ to a vertex of $\eta(z_1 x_1, x_1)$ with all internal vertices
 % in $\eta(z_1 x_1) \setminus (\eta(z_1x_1, z_1) \cup \eta(z_1x_1, x_1))$,
 $R_2'$ be a shortest path from $y_2$ to a vertex of $\eta(z_1 x_2, x_2)$ with all internal vertices
 in $\eta(z_1 x_2) \setminus (\eta(z_1x_2, z_1) \cup \eta(z_1x_2, x_2))$. %,
 % and $R_3'$ be a shortest path from $y_3$ to a vertex of $\eta(z_2 x_3, x_3)$ with all internal vertices
 % in $\eta(z_2 x_3) \setminus (\eta(z_2x_3, z_2) \cup \eta(z_2x_3, x_3))$.
 % By appending $R_1'$ and $R_1$ to $y_1$, $R_2'$ and $R_2$ to $y_2$, and $R_3'$ and $R_3$ to $y_3$
 % and connecting $y_1$ and $y_3$ via $Q$, we obtain an induced $S_{t,t,t}$
 % in $G$ with center in $y_1$.
 By appending $R_1$ to $y_1$, $R_2'$ and $R_2$ to $y_2$, and $R_3$ to $y_3$
 and connecting $y_1$ and $y_3$ via $Q$, we obtain an induced $\Pttt$
 in $G$ with apex in $y_1$.

 %\mpil[inline]{The paragraph above contains a possibly serious bug: the construction does not work 
 %if $v$ is adjacent to both $y_1$ and $y_2$. At the moment I don't know how to fix it.}

 In the second case, there is a separation $(A_1, B_1)$ in $H[B \setminus A]$ of order less than $3$ with $X \subseteq A_1$
 and all pegs of $W'$ lying in $B_1$.
 Define $A' = A \cup A_1$ and $B' = B_1 \cup (A \cap B)$. 
 Clearly, $(A',B')$ is a separation of order at most $|A \cap B| + 2$ in $H$ with $A \subseteq A'$ and $B' \subseteq B$.
 Since the sidelength of $W'$ is at least $3|A \cap B| + 7$ and all pegs of $W'$ lie in $B'$,
 we have $(A',B') \in \mathcal{T}_{W'}$. 
 Since $A \subseteq A'$, any backdoor vertex of $(A',B')$ is also a backdoor vertex of $(A,B)$, and thus is in the set
 $Z$ of size $2$. 
 Finally, since $X \subseteq A_1$, for every $xyz \in T(H)$ with $\eta(xyz) \cap N(v) \neq \emptyset$ we have
 $x,y,z \in A'$. 
 Thus, the wall $W'$ and the separation $(A',B')$ is the desired outcome.
\end{proof}

\begin{lemma}[backdoor cleanup]\label{lem:cleanup-backdoor}
Let $W$ be a $v$-pure wall in $G$ and let $(A,B) \in \mathcal{T}_W$ be a separation capturing $\Projection$
with a set of backdoors $Z$ of size at most $2$ that is triangle-safe and such that the sidelength of $W$
is $k_W \geq 4|A \cap B| + 7$.

Suppose there exists a backdoor vertex $z \in Z$ that is not pure.
%and a path $Q$ in $G[\{v\} \cup \eta(z)]$ 
%with endpoints $v$ and $v'$ (possibly of zero length, i.e., we allow $v=v'$)
%such that $v'$ is the only vertex of $Q$ that has neighbors in
%$\PotatoZ \coloneqq \bigcup_{x \in N_H(z) \cap (B \setminus A)} \eta(zx,z)$,
%but $v'$ is not complete to $\PotatoZ$.
%
Then, one of the following holds:
\begin{itemize}
    \item $G$ admits an induced $\Pttt$;
    \item $H$ admits a separation $(A_0,B_0) \in \mathcal{T}_W$ of order $1$ that captures $\Projection$ with backdoor set $A\cap B$;
    \item $H$ admits a subwall $W'$ of $W$ of sidelength at least $k_W - |A \cap B|$
and a separation $(A',B') \in \mathcal{T}_{W'}$ of order at most $|A \cap B| + 2$ with $A \subseteq A'$ and $B' \subseteq B$
capturing $\Projection$ with backdoor set contained in $Z \setminus \{z\}$;  or
   \item $H$ admits a separation $(A^\ast, B^\ast) \in \mathcal{T}_W$ of order at most $|A \cap B|$
   with $A \subseteq A^\ast$, $B^\ast \subseteq B$, and $(A,B) \neq (A^\ast, B^\ast)$.
\end{itemize}
\end{lemma}
\begin{proof}
    Recall that for $z\in A\cap B$, $\PotatoZ = \bigcup_{x \in N_H(z) \cap (B \setminus A)} \eta(zx,z)$.
    Fix a $z$-entry point $v' \in \{v\} \cup \eta(z)$
    that causes $z$ not to be pure: $v'$ has some neighbors in $\PotatoZ$, but is not complete
    to $\PotatoZ$.

    Since $(A,B) \in \mathcal{T}_W$ while the sidelength of $W$ is at least $4|A \cap B| + 7$, 
    we define $W'$ to be a subwall of $W$ of sidelength at least $k_W - |A \cap B| \geq 3|A \cap B| + 7$
    that is fully contained in $H[B \setminus A]$.

    We define $X_1$ to be the set of those vertices $x \in A$ for which either $xz \in E(H)$ and $\Projection \cap (\eta(xz) \setminus \eta(xz,z)) \neq \emptyset$, or there exists a neighbor $y \in N_H(x) \cap (A \setminus \{z\})$ with $\Projection \cap \eta(xy) \neq \emptyset$.
    Let $X_\bot$ be the set of those vertices $y \in B \setminus A$ for which there exists $x_y \in X_1$ and
    a path $P_y$ in $H$ from $x_y$ to $y$ that avoids $z$ and whose only vertex outside $A$ is $y$; denote the last 
    edge of $P_y$ by $e_y$ and the penultimate vertex of $P_y$ by $z_y$ (so that $e_y = z_yy$). 
    Note that Lemma~\ref{lem:connect-via-projection} implies that for every $y \in X_\bot$ there exists an induced
    path $R_y$ from $v$ to a vertex $\eta(e_y,y)$ whose all internal vertices belong to 
    \begin{equation}\label{eq:backdoor:Rv} (\eta(e_y) \setminus \eta(e_y,y)) \cup \bigcup_{x \in N_H(z) \cap A} (\eta(zx) \setminus \eta(zx, x))
      \cup \bigcup_{e \in E(H[A \setminus \{z\}])} \eta(e) \cup \bigcup_{x \in A \setminus \{z\}} \eta(x) \cup \bigcup_{T \in T(H[A])} \eta(T). 
      \end{equation}

    We construct an auxiliary graph $H_1$ as follows. 
    Start with $H_1 \coloneqq H[B \setminus A]$.
    Add all vertices of $\PotatoZ$ to $H_1$ and for a vertex $u \in \PotatoZ$,
    if $u \in \eta(zy,z)$ for $y \in B \setminus A$, make $u$ adjacent to $y$ in $H_1$.
    Add three new vertices $a_\circ$, $a_\bullet$, and $a_\bot$.
    Make $a_\bot$ adjacent to all vertices of $X_\bot$.
    Make $a_\circ$ adjacent to all vertices of $\PotatoZ$ that are nonadjacent to $v'$ in $G$
    and make $a_\bullet$ adjacent to all vertices of $\PotatoZ$ that are adjacent to $v'$ in $G$.

    We consider two cases. In the first case, there are three vertex-disjoint paths $P_\circ$, $P_\bullet$, and $P_\bot$
    in $H_1$ from $a_\circ$, $a_\bullet$, and $a_\bot$, respectively, to the set of pegs of $W'$. 
    In this case, we will exhibit an induced $\Pttt$ in $G$.

    Since every vertex of $\PotatoZ$ is of degree $2$ in $H_1$, 
    the path $P_\circ$ starts in $a_\circ$, continues via $u_\circ \in \eta(zx_\circ,z)$
    to $x_\circ \in B \setminus A$ and then stays in $B \setminus A$ till the end. Similarly,
    the path $P_\bullet$ starts in $a_\bullet$, continues via $u_\bullet \in \eta(zx_\bullet,z)$
    to $x_\bullet \in B \setminus A$ and then stays in $B \setminus A$ till the end.
    By construction, $v'u_\circ \notin E(G)$ while $v'u_\bullet \in E(G)$.
    Since $P_\circ$ and $P_\bullet$ are vertex-disjoint, $x_\circ \neq x_\bullet$ and thus
    $u_\circ u_\bullet \in E(G)$ by the properties of an extended strip decomposition.
    Let $a_\bot x_\bot$ be the first edge of $P_\bot$ for some $x_\bot \in X_\bot$.
    
    Apply Lemma~\ref{lem:pathInWall} to $P_\circ$, $P_\bullet$, $P_\bot$, and the wall $W'$, obtaining almost vertex-disjoint paths $P_\circ'$,
    $P_\bullet'$, and $P_\bot'$. Note that $a_\circ u_\circ$ and $u_\circ x_\circ$ remain the first two edges of $P_\circ'$,
    $a_\bullet u_\bullet$ and $u_\bullet x_\bullet$ remain the first two edges of $P_\bullet'$,
    and $a_\bot x_\bot$ remains the first edge of $P_\bot'$.

    The suffixes of paths $P_\circ'$, $P_\bullet'$, and $P_\bot'$ from $x_\circ$, $x_\bullet$, and $x_\bot$, respectively,
    stay in $H[B \setminus A]$, which is $v$-safe due to the assumption that $(A,B)$ captures $\Projection$.
    Apply Lemma~\ref{lem:wall-path-to-real-path} to these three suffixes and $H[B \setminus A]$, obtaining induced
    paths $R_\circ$, $R_\bullet$, and $R_\bot$ which form a pyramid base.

    We now exhibit an induced $\Pttt$ in $G$ with apex $u_\bullet$. For the first leg,
    as $(H,\eta)$ is locally cleaned (in particular, the ``\hyperref[mov:DisconnectedVI]{moving a disconnected vertex of an interface}'' and
    ``\hyperref[mov:DisconncetedES]{moving a disconnected component of an edge set}'' operations are inapplicable), there is a path in $\eta(zx_\bullet)$ 
    (possibly of zero length) from $u_\bullet$ to a vertex of $\eta(zx_\bullet, x_\bullet)$ whose only
    vertex of $\eta(zx_\bullet, z)$ is $u_\bullet$; concatenate this path with $R_\bullet$.
    For the second leg, start with an edge $u_\bullet u_\circ$, continue via an analogous path
    in $\eta(zx_\circ)$ from $u_\circ$ to a vertex of $\eta(zx_\circ, x_\circ)$ whose only vertex of $\eta(zx_\circ,z)$
    is $u_\circ$ and append $R_\circ$ at the end.
    For the third leg, go with a shortest path from $u_\bullet$ to the starting vertex of $R_\bot$ via
    $v'$, $Q^\Projection_{v'}$, $v$, $R_{x_\bot}$, and finish with $R_\bot$. 
    Condition~\eqref{eq:backdoor:Rv} ensures that this is indeed an induced $\Pttt$ in $G$.

    We are left with the second case where $H_1$ contains a separation $(A_1,B_1)$ of order less than three
    with $a_\circ, a_\bullet, a_\bot \in A_1$ but all pegs of $W'$ in $B_1$. 
    Pick such $(A_1,B_1)$ with $B_1$ inclusion-wise minimal. 

    Observe that no vertex of $\PotatoZ$ is in $A_1 \cap B_1$, as if
    some $u \in A_1 \cap B_1$ with $u \in \eta(zx, z)$, $x \in B \setminus A$, then as $a_\circ, a_\bullet \in A_1$,
    $(A_1 \cup \{x\}, B_1 \setminus \{u\})$ is also a separation of order less than $3$, contradicting the choice of $(A_1,B_1)$.
    Let $B'$ consist of the pegs of $W'$ and $N_H[(B \setminus A) \cap (B_1 \setminus A_1)]$ and
    $A' = V(H) \setminus ((B \setminus A) \cap (B_1 \setminus A_1))$.
    Clearly, $(A', B')$ is a separation in $H$. 
    Observe that $A' \cap B' \subseteq (A \cap B) \cup (A_1 \cap B_1)$, and thus
    the order of $(A',B')$ is at most $|A \cap B| + |A_1 \cap B_1| \leq |A \cap B| + 2$.
    As all pegs of $W'$ are in $B'$, $(A',B') \in \mathcal{T}_{W'} \subseteq \mathcal{T}_W$.  
    We have $B' \subseteq B$ and thus $A \subseteq A'$, and hence $(A',B')$ is triangle-safe and captures $\Projection$ 
    with a backdoor set being a subset of $Z$. 

    If there is no path in $H[B \setminus A]$ from $X_\bot$ to a peg of $W'$, then define a separation as follows.
    Let $B_0$ consist of $z$ and all vertices of $H$ reachable in $H-\{z\}$ from a peg of $W'$, and let
    $A_0 = (V(H) \setminus B_0) \cup \{z\}$. Clearly, $(A_0, B_0)$ is a separation of order $1$ in $H$.
    Since all pegs of $W'$ lie in $B_0$, we have $(A_0,B_0) \in \mathcal{T}_{W'} \subseteq \mathcal{T}_W$.
    By the assumption, $X_\bot \subseteq A_0$. Hence, $(A_0,B_0)$ is $\Projection$-capturing. This gives the second outcome
    of the lemma.

    If for some $a \in \{a_\circ, a_\bullet\}$ there is no path in $H_1$ from $a$ to a peg of $W'$, define a separation 
    $(A^\ast, B^\ast)$ as follow. Let $C$ be all vertices of $B \setminus A$ reachable from $a$ in $H_1$. 
    Since $v'$ has at least one neighbor in $\PotatoZ$, but is not complete to $\PotatoZ$, we have that $a$ is
    not an isolated vertex of $H_1$ and thus $C \neq \emptyset$. 
    Let $A^\ast = A \cup C$ and $B^\ast = B \setminus C$. Then, 
    $(A^\ast, B^\ast)$ is also a separation in $H$ with $A \cap B = A^\ast \cap B^\ast$ 
    and all pegs in $W'$ lying in $B^\ast$ (thus $(A^\ast, B^\ast) \in \mathcal{T}_{W'} \subseteq \mathcal{T}_W$)
    but with $A \subsetneq A^\ast$, $B^\ast \subsetneq B$. This gives the last outcome of the lemma.

    If $z$ is not a backdoor vertex of $(A',B')$, then we are done with the penultimate outcome.
    Otherwise, there exists $x \in B '\setminus A'$ with $\Projection \cap \eta(zx,z) \neq \emptyset$.
    We have $x \in B_1 \setminus A_1$ and thus either $a_\circ$ or $a_\bullet$ belongs to $A_1 \cap B_1$. 
    Since $|A_1 \cap B_1| < 3$ and for every $a \in \{a_\circ, a_\bullet, a_\bot\}$ 
    there is a path from $a$ to a peg of $W'$ in $H_1$, 
    $A_1 \cap B_1$ is of size $2$ and consists of $a \in \{a_\circ, a_\bullet\}$ and a vertex $z' \in B \setminus A$.
    In particular, $a_\bot \in A_1 \setminus B_1$ hence $X_\bot \subseteq A_1$. 
    Observe that for every $y \in X_\bot$, we have $N_H(z_y) \cap (B \setminus A) \subseteq X_\bot$. Hence,
    for every $y \in X_\bot$, the vertex $z_y$ lies in $A' \setminus B'$. 
    As there is a path from $X_\bot$ to a peg of $W'$ in $H[B \setminus A]$, in particular we have $X_\bot \neq \emptyset$,
    so at least one $z_y$ in $(A' \setminus B') \cap (A \cap B)$ exists. 
    As $A_1 \cap B_1 = \{a, z'\}$, we have $A' \cap B' \subseteq (A \cap B) \cup \{z'\}$. 
    As $(A' \setminus B') \cap (A \cap B) \neq \emptyset$, we have $|A' \cap B'| \leq |A \cap B|$ and $(A', B') \neq (A,B)$. 
    This proves that $(A',B')$ satisfies the properties of the last outcome of the lemma.
\end{proof}

We conclude this section with a lemma summarizing what we obtained so far.
\begin{lemma}\label{lem:esd:final-cut}
   For every constant $t$, there exists a constant $c_t$ and a polynomial-time algorithm that, given input as in Lemma~\ref{lem:esd:step},
   either returns one of the promised outputs of Lemma~\ref{lem:esd:step}
   or outputs all the following objects:
   \begin{itemize}
       \item a locally clean extended strip decomposition $(H,\eta)$ of $G-v$ 
   where every vertex of $H$ has degree at least~$2$ and every particle of has weight at most $0.01\tau$,
       \item a separation $(A,B)$ in $H$ of order at most $c_t$ that is triangle-safe, captures $\Projection$ with 
       backdoor set $Z$ of size at most $2$ such that every $z \in Z$ is pure, and satisfies
           $\w(\preimage{(H,\eta)}{B \setminus A}) \geq 0.99\tau$, and
        \item for every $X \subseteq A \cap B$ of size three, 
          a family $\mathcal{Q}^X = \{Q_x^X~|~x \in X\}$ of three  induced paths in $G$ on $t$ vertices each which induce a pyramid base,
        where for every $x \in X$ there exists ${y_x \in N_H(x) \cap (B \setminus A)}$ such that
        the path $Q_x^X$ starts in a vertex of $\eta(xy_x,x)$
        and has all remaining vertices in $\bigcup_{e \in E(G[B \setminus A])} \eta(e)$
        or in $\eta(xy_x) \setminus \eta(xy_x, x)$.
    \end{itemize}
\end{lemma}
\begin{proof}
    We compute $(H,\eta)$ as discussed in Section~\ref{sub:localCleaning}, by exhaustively
    applying the local cleaning operation to the input extended strip decomposition of $G-v$ and discarding isolated
    vertices of $H$ with empty vertex sets. Either a $0.99\tau$-balanced separator dominated by a constant number of vertices
    is returned, or $(H,\eta)$ satisfies properties~\eqref{eq:esd:H-small-particles}, \eqref{eq:esd:H-degree-two},
    and~\eqref{eq:esd:H-dom-separator}.
    
    We fix $\tanglethreshold$ to be large enough constant depending on $t$, emerging from the proof.
    We start searching for $W$ and $(A,B)$ by applying Lemma~\ref{lem:find-safe-wall}. Thus, we can either already conclude, or get
    a $v$-safe wall $W$ of sidelength $3\tanglethreshold$ satisfying~\eqref{eq:W-tangle}.
    
    We proceed with $W$ with a series of 
    lemmas that each either concludes the reasoning,
    or exhibit a subwall $W'$ of $W$
    with some additional property. 
    The new subwall will be large in the following
    sense: if for every constant $\tanglethreshold'$
    there exists a constant $\tanglethreshold$
    such that if $W$ is of sidelength at least
    $\tanglethreshold$, then $W'$ is guaranteed
    to be of sidelength at least $\tanglethreshold'$.
    After each step, we rename $W'$ back to $W$.
    For the sake of clarity, we will not 
    follow the exact computation of the dependencies
    of $\tanglethreshold$ on $\tanglethreshold'$,
    but only refer to them as a ``decrease''
    or ``losing on the sidelength''.
    
    We start with applying Lemma~\ref{lem:safe-to-pure} to $W$; by losing a bit on the sidelength of $W$, we can assume that $W$ is actually
    $v$-pure. 
    We then apply Lemma~\ref{lem:capture-projection}. This yields either an induced $\Pttt$ in $G$ (which we can return)
    or a separation ${(A,B) \in \mathcal{T}_W}$ of order less than $19$ that captures $\Projection$
    and has backdoor set of size at most $2$.
    We pass it to Lemma~\ref{lem:clean-triangles} that either finds $\Pttt$ or, at the cost of a slight decrease (of at most $|A\cap B|$)  of the sidelength 
  of $W$ and an increase of the order of $(A,B)$ by at most $2$, upgrades $(A,B)$ to be triangle-safe. (Here, we slightly abuse the notation
    by denoting the output of Lemma~\ref{lem:clean-triangles} by $W$ and $(A,B)$, again.)

    We now iterate on improving $W$ and $(A,B)$ using Lemma~\ref{lem:cleanup-backdoor}.
    While $(A,B)$ has a backdoor vertex $z$ that is not pure, apply Lemma~\ref{lem:cleanup-backdoor} to $W$, $(A,B)$, and $z$.
    If the third outcome happens, replace $W$ with $W'$ and $(A,B)$ with $(A',B')$ and repeat; note that this can happen
    only twice as initially $(A,B)$ has at most two backdoors. 
    If the fourth outcome happens, replace $(A,B)$ with $(A^\ast, B^\ast)$ and repeat; note that this step can happen
    $\Oh(|V(H)|)$ times, but does not degrade the order of $(A,B)$ nor the size of the wall $W$.
    
    If we reach $(A,B)$ with all backdoor vertices being pure, we proceed as follows.
    First, note that we can assume that there is no $(A',B') \in \mathcal{T}_W$ with $A \subseteq A'$ but
    $|A' \cap B'| < |A \cap B|$, as then we can replace $(A,B)$ with $(A',B')$, because such an operation cannot turn a
    pure backdoor into a non-pure backdoor. 
    We would like to return $(H,\eta)$ and $(A,B)$ as the third outcome of the lemma; to this end, we need to construct
    the path families $\mathcal{Q}^X$. 
    Fix a subwall $W'$ of $W$ contained in $H[B \setminus A]$; assuming $\tanglethreshold$ is large enough, we can choose $W'$  of sidelength
    at least $3|A \cap B| + 4$.
    Since $A \cap B$ and the pegs of $W'$ cannot be separated in $H[B]$ by a separation of order less than $|A \cap B|$, 
    by Menger's theorem, there is a family $\mathcal{P} = \{P_x~|~x \in A \cap B\}$ of vertex-disjoint paths in $H[B]$
    such that every $P_x \in \mathcal{P}$ starts in $x$ and ends in a peg of $W'$. 
    For every $X \subseteq A \cap B$ of size three, construct $\mathcal{Q}^X$ as follows:
    Apply Lemma~\ref{lem:pathInWall} to $\{P_x~|~x \in X\}$ and $W'$, obtaining paths $\{P_x'~|~x \in X\}$.
    Let $y_x$ be the second vertex of a path $P_x'$ for $x \in X$. 
    As $H[B \setminus A]$ is $v$-pure, apply Lemma~\ref{lem:wall-path-to-real-path} to the paths $P_x'$ with the first edges removed,
    obtaining paths $R_x$ for $x \in X$. Prepend every path $R_x$ with a shortest path in $\eta(xy_x)$ from $\eta(xy_x, x)$ to $\eta(xy_x, y_x)$, obtaining the desired path $Q_x^X$.

    The remaining case is when one of the applications of Lemma~\ref{lem:cleanup-backdoor} returns a separation
    $(A_0,B_0) \in \mathcal{T}_W$ of order $1$ that captures $\Projection$.
    Let $\{z\} = A_0 \cap B_0$.
    Compute an extended strip decomposition $(H',\eta')$ of $G$ as follows.
    Start with $H' \coloneqq H[B_0]$ and $\eta'(\alpha) \coloneqq \eta(\alpha)$ for every $\alpha \in V(H') \cup E(H') \cup T(H')$.
    Let $V'$ be the set of vertices of $G$ that are already in some set of $\eta'$.
    Add every vertex of $V(G) \setminus V'$ to the vertex set $\eta'(z)$. 

    We first observe that $(H',\eta')$ is indeed an extended strip decomposition of $G$.
    Indeed, since $z$ is a cutvertex of $H$, every vertex of $V(G) \setminus (V' \cup \{v\})$ has only neighbors in
    $V(G) \setminus V'$ and $\bigcup_{y \in N_H(z) \cap B_0} \eta(zy,z)$. 
    Furthermore, since $(A_0,B_0)$ captures $\Projection$, all neighbors of $v$ are also only in $V(G) \setminus V'$
    and $\bigcup_{y \in N_H(z) \cap B_0} \eta(zy,z)$. 
    Since $(A_0,B_0) \in \mathcal{T}_W$, we have $\w(\preimage{(H,\eta)}{B_0 \setminus A_0}) \geq 0.99\tau$ by~\eqref{eq:W-tangle}.
    Hence, as $\w(G) \leq \tau$, we have $\w(\eta'(z)) \leq 0.01\tau$.
    Since every particle of $(H,\eta)$ is of weight at most $0.01\tau$ by~\eqref{eq:esd:H-small-particles},
    every particle of $(H',\eta')$ is of weight at most $(0.01 + 0.01)\tau = 0.02\tau$. 
    Hence, we can output $(H',\eta')$ as the third output of Lemma~\ref{lem:esd:step}.

    To finish the proof, we observe that after getting the first wall $W$ from Lemma~\ref{lem:find-safe-wall},
    all further steps tackle separations of constant order and subwalls of a wall of constant sidelength. 
    Hence, all computations in later steps can be done in polynomial time naively.
\end{proof}

\subsection{Applying three-in-a-tree}
Lemma~\ref{lem:esd:final-cut} allows us to conclude the proof of Lemma~\ref{lem:esd:step}
by applying the three-in-a-tree theorem (Theorem~\ref{thm:three-in-a-tree}) in the following way.

Given the input to Lemma~\ref{lem:esd:step}, we apply Lemma~\ref{lem:esd:final-cut} for the constant $t$.
Unless we are already done, we obtained the last output,
consisting of $(H,\eta)$, $(A,B)$, and $W$. Recall that $(A,B)$ captures $\Projection$ with a backdoor set $Z$
of size at most $2$, it is triangle-safe, and all $z \in Z$ are pure backdoors.

We will need the following notation. Fix $z \in Z$. Recall that
    $\PotatoZ = \bigcup_{x \in N_H(z) \cap (B \setminus A)} \eta(zx,z)$.
    Let $V_z \subseteq \eta(z)$ 
    %\tm{we can make $V_z \subseteq \eta(z)\cup \{v\}$ and cut out, two plus sentnces later.}\mpil{I am not sure I understand this comment.} 
    % Marcin: discussed on discord, this shortening is not worth it
    be the set of those vertices $u \in \eta(z)$
that are reachable from $v$ by a path contained in $G[\{v\} \cup \eta(z)]$ whose all vertices, except for possibly the last one,
are anti-adjacent to $\PotatoZ$. 
Note that as $z$ is a pure backdoor, every vertex $u \in V_z$
is either completely adjacent to $\PotatoZ$ or completely anti-adjacent to $\PotatoZ$;
we denote by $V_z^\bullet$ and $V_z^\circ$ the sets of vertices of $u \in V_z$ that are completely adjacent 
and completely anti-adjacent to $\PotatoZ$, respectively.
Furthermore, observe that as $z$ is a pure backdoor, $v$ is either completely adjacent to $\PotatoZ$ or completely
anti-adjacent to $\PotatoZ$ and, in the first case, $V_z = \emptyset$.

We construct an auxiliary graph $G_A$ as follows. Let $X_A \subseteq V(G)$ consist of the following:
\begin{itemize}
    \item the vertex $v$;
    \item every $\eta(e)$ for $e \in E(G[A])$;
    \item every $\eta(xyz)$ for $xyz \in T(G[A])$;
    \item every $\eta(x)$ for $x \in A \setminus B$; and
    \item for every $z \in Z$, the set $V_z$.
\end{itemize}
Observe that $X_A$ is disjoint with $\preimage{(H,\eta)}{B \setminus A}$ and thus
$\w(X_A) \leq 0.01\tau$.

We start with $G_A \coloneqq G[X_A]$ and then, for every $x \in A \cap B$ we add two new adjacent vertices
$a_{x,1}$ and $a_{x,2}$ to $G_A$, adjacent to every vertex of $\bigcup_{y \in N_H(x) \cap A} \eta(xy,x)$.
Furthermore, if $x \in Z$, we make $a_{x,1}$ and $a_{x,2}$ fully adjacent to 
$V_z^\bullet$ and, if $v$ is complete to $\PotatoZ$, also adjacent to $v$. 
Observe that for $x \in A \cap B$, the vertices $a_{x,1}$ and $a_{x,2}$ are true twins in $G_A$.
Let $\mathcal{Z} = \{a_{x,i}~|~x \in A \cap B, i \in \{1,2\}\}$. 

We apply Theorem~\ref{thm:three-in-a-tree} to $G_A$ and the set $\mathcal{Z}$. There are two possible outcomes:
either an induced tree in $G_A$ containing at least three elements of $\mathcal{Z}$ or
a rigid extended strip decomposition $(H_A,\eta_A)$ of $(G_A,\mathcal{Z})$. We deal with these cases separately.

\paragraph{An induced tree in $G_A$.}
Let $K$ be an induced tree in $G_A$ containing at least three elements of~$\mathcal{Z}$. 
Without loss of generality, we can assume that $K$ contains exactly three elements of $\mathcal{Z}$,
and $K$ is either a path with two endpoints in $\mathcal{Z}$ or a tree with exactly three leaves being the elements 
of $\mathcal{Z}$.
For $x \in A \cap B$, as $a_{x,1}$ and $a_{x,2}$ are true twins while $K$ contains three elements of $\mathcal{Z}$
and is an induced tree, $K$ can contain only one of $a_{x,1}$ and $a_{x,2}$. Without loss of generality, we can assume that
$V(K) \cap \mathcal{Z} = \{a_{x,1}~|~x \in X\}$ for some $X \subseteq A \cap B$ of size three.

Consider the family $\mathcal{Q}_X = \{Q_x^X~|~x \in X\}$ promised by Lemma~\ref{lem:esd:final-cut} and let 
$u_x$ be the starting vertex of $Q_x^X$ for $x \in X$. Observe that 
$u_x$ has exactly the same neighbors in $X_A$ as $a_{x,1}$, while all other vertices of $Q_x^X$ are anti-adjacent to $X_A$. 
Hence, $(V(K) \setminus \{a_{x_1}~|~x \in X\}) \cup \{V(Q_x^X)~|~x \in X\}$ induces a tree in $G$ that contains an $\Pttt$, as desired.

\paragraph{An extended strip decomposition of $G_A$.}
We will now show that one can merge $(H_A,\eta_A)$ and $(H,\eta)$ into an extended strip decomposition $(H^\ast,\eta^\ast)$
of the whole graph $G$.

Recall that for every $a_{x,i} \in \mathcal{Z}$ there is a degree-$1$ vertex $\xi_{x,i} \in V(H_A)$ with
a unique neighbor $\zeta_{x,i} \in V(H_A)$ and $\eta_A(\xi_{x,i}\zeta_{x,i}, \xi_{x,i}) = \{a_{x,i}\}$.
For every $x \in A \cap B$, since $a_{x,1}$ and $a_{x,2}$ are adjacent, we have $\zeta_{x,1} = \zeta_{x,2}$ (and we henceforth
denote this vertex by $\zeta_x$)
and $a_{x,i} \in \eta_A(\xi_{x,i}\zeta_x, \zeta_x)$ for $i=1,2$. Thus, we can assume without loss of generality that
actually $\eta_A(\xi_{x,i} \zeta_x) = \{a_{x,i}\}$ and $\eta_A(\xi_{x,i}) = \emptyset$,
as all vertices of $\eta_A(\xi_{x,i} \zeta_x) \cup \eta_A(\xi_{x,i})$ except for $a_{x,i}$ can be moved to $\eta_A(\zeta_x)$.
Furthermore,
since $a_{x,i}$ and $a_{y,j}$ are nonadjacent for $x \neq y$, $i,j \in \{1,2\}$, we have that $\zeta_x$ and $\zeta_y$ are distinct.

Denote $H_A' \coloneqq H_A - \{\xi_{x,i}~|~x \in A \cap B, i = 1,2\}$ and observe
that $H_A'$ with $\eta_A'$ defined as $\eta_A$ restricted to the vertices, edges, and triangles present in $H_A'$
is an extended strip decomposition of $G[X_A]$. 

Note that every vertex of $V(G) \setminus X_A$ appears in $(H,\eta)$ either in $\eta(x)$ for some $x \in B$,
in $\eta(e)$ for an edge $e$ with at least one endpoint in $B \setminus A$ and both endpoints in $B$,
or in $\eta(xyz)$ for a triangle $xyz$ with $x,y,z \in B$ and at least one vertex in $B \setminus A$.
Hence, $H_B' \coloneqq H[B]$ with $\eta_B'$ defined as $\eta$ with the domain 
restricted to the objects present in $H[B]$ and every value restricted to vertices of $V(G) \setminus X_A$ is an extended strip 
decomposition of $G-X_A$.
We note that $\eta_B'(e) = \emptyset$ for an edge $e$ with both endpoints in $A \cap B$, but we retain
the edge $e$ in $H_B'$ as there may be triangles of $H[B]$ involving it.

We construct $H^\ast$ by taking a disjoint union of $H_A'$ and $H_B'$,
discarding edges $e$ of $H_B'$ that have both endpoints in $A \cap B$,
and identifying $\zeta_x$ with $x$ for every $x \in A \cap B$.
We define $\eta^\ast$ as follows: 
\begin{itemize}
    \item $\eta^\ast(e)$ for $e \in E(H^\ast)$ equals $\eta_A'(e)$ or $\eta_B'(e)$, depending on whether $e$ came from $H_A'$ or $H_B'$;
    \item $\eta^\ast(xyz)$ for $xyz \in T(H^\ast)$ equals $\eta_A'(xyz)$, $\eta_B'(xyz)$, or $\emptyset$,
    depending on whether the triangle $xyz$
    is present in $H_A'$, or at least one of $x,y,z$ is in $B \setminus A$, or none of these options happen;
    \item $\eta^\ast(x)$ for $x \in V(H^\ast)$ equals $\eta_A'(x)$ for $x \in V(H_A') \setminus \{\zeta_x~|~x \in A \cap B\}$, 
    equals $\eta_B'(x)$ for $x \in B \setminus A$, equals $\eta_A'(\zeta_x) \cup \eta_B'(x)$ for $x \in A \cap B$.
\end{itemize}
Finally, we discard from $H^\ast$ every edge $e$ for which $\eta^\ast(e) = \emptyset$ and $e$ does not participate in any 
triangle with a nonempty set in $\eta^\ast$, and every isolated vertex with an empty vertex set.

We now check that $(H^\ast,\eta^\ast)$ is indeed a rigid extended strip decomposition of $G$. 
Recall that $(H_A',\eta_A')$ and $(H_B',\eta_B')$ are extended strip decompositions of $G[X_A]$ and $G-X_A$, respectively.
Furthermore, $(H_A', \eta_A')$ is rigid; for $(H_B',\eta_B')$, we have $\eta_B'(e) = \eta(e)$ and $\eta_B'(e,x) = \eta(e,x)$ for every
$e$ with at least one endpoint in $B \setminus A$ and any $x \in e$, so $(H_B',\eta_B')$ can violate the requirements
of being rigid only at vertex sets and edges contained in $A \cap B$.
We infer that it remains to check the following three properties:
\begin{enumerate}
    \item For every $uw \in E(G)$ with $u \in X_A$ and $w \notin X_A$,  $u$ and $w$ are placed in $(H^\ast, \eta^\ast)$
    in a way allowing the edge $uw$, that is, both $u$ and $w$ either 
    \begin{itemize}
        \item belong to one set $\eta^\ast(e)$ for some $e \in E(H^\ast)$, or 
        \item belong to one set $\eta^\ast(x) \cup \bigcup_{y \in N_{H^\ast}(x)} \eta(xy,x)$ for some $x \in V(H^\ast)$, or 
        \item belong to one set $\eta^\ast(xyz) \cup (\eta^\ast(xy,x) \cap \eta^\ast(xy,y)) \cup (\eta^\ast(yz,y) \cap \eta^\ast(yz,z)) \cup (\eta^\ast(zx,z) \cap \eta^\ast(zx,x))$
    for some triangle $xyz \in T(H^\ast)$.
    \end{itemize}
    \item For every $u \in X_A$ and $w \in V(G) \setminus X_A$ such that for some $x \in A \cap B$
    it holds that $u,w \in \bigcup_{y \in N_{H^\ast}(x)} \eta^\ast(xy,x)$, we have $uw \in E(G)$.
    \item For every edge $xy \in E(H)$ with $x,y\subseteq A \cap B$, if there exists a triangle $xyz$ with $z \in B \setminus A$
    and $\eta(xyz) = \eta_B'(xyz) \neq \emptyset$, the edge $\zeta_x \zeta_y$ exists in $H_A'$. 
\end{enumerate}

For the first property, let $u \in X_A$ and $w \notin X_A$ be adjacent. We observe that, since $(H,\eta)$ is an extended
strip decomposition of $G-v$, we can break into the following subcases.
\begin{itemize}
    \item $u = v$. Then, as $(A,B)$ is triangle-safe and captures $\Projection$, 
    we have two options when $w$ exists:
    \begin{itemize}
        \item There exists $z \in Z$
    with $v$ complete to $\PotatoZ$ and $w \in \PotatoZ$. 
    Then, $v$ is adjacent to $a_{z,1}$ and $a_{z,2}$. Hence, in $(H_A',\eta_A')$ we have that
    $v \in \eta_A'(\zeta_z) \cup \bigcup_{y \in N_{H_A'}(\zeta_z)} \eta(\zeta_z y, \zeta_z)$. 
    Consequently, in $(H^\ast, \eta^\ast)$ we have that both $u=v$ and $w$ belong to 
    $\eta^\ast(z) \cup \bigcup_{y \in N_{H^\ast}(z)} \eta^\ast(zy,z)$, as desired.
    \item There exists a triangle $xyz \in T(H)$ with $x,y \in Z$ and $z \in B \setminus A$ with
    $w \in \eta(xyz)$. Then, as $(A,B)$ is triangle-safe, we have that $v$ is complete to $\eta(xz,x) \cup \eta(yz,y)$.
    Hence, $v$ is adjacent to $a_{x,1}$, $a_{x,2}$, $a_{y,1}$, and $a_{y,2}$ in $G_A$. 
    The only way how $(H_A,\eta_A)$ can accommodate this is if $\zeta_{x} \zeta_{y} \in E(H_A)$ and
    $v \in \eta_A(\zeta_{x} \zeta_{y}, \zeta_x) \cap \eta_A(\zeta_x \zeta_y, \zeta_y)$. 
    Hence, $v \in \eta^\ast(xy,x) \cap \eta^\ast(xy,y)$ while $w \in \eta^\ast(xyz)$, as desired.
    \end{itemize}
    \item $u \in \eta(z)$ for some $z \in V(H)$. By the definition of $X_A$, we have actually $z \in A$.
    By the existence of $w$, we have $z \in A \cap B$. Since the only parts of sets $\eta(x)$ for $x \in A \cap B$
    that are in $X_A$ are sets $V_x$ for $x \in Z$, we have $z \in Z$ and $v$ is anti-complete to $\PotatoZ$.
    Furthermore, as $V_z^\circ$ is nonadjacent to $\PotatoZ$ and to $\eta(z) \setminus V_z$, we have $u \in V_z^\bullet$
    and $w \in \PotatoZ \cup (\eta(z) \setminus V_z)$.
    Then, $u$ is adjacent to $a_{z,1}$ and $a_{z,2}$ in $G_A$. Thus, $u \in \eta_A'(\zeta_z) \cup \bigcup_{y \in N_{H_A'}(z)} \eta(\zeta_z y, \zeta_z)$. 
    Hence, $u,w \in \eta^\ast(z) \cup \bigcup_{y \in N_{H^\ast}(z)} \eta^\ast(zy,z)$, as desired.
    \item $u \in \eta(e)$ for some $e \in E(H)$. By the definition of $X_A$, both endpoints of $e$ are in $A$. 
    Due to the existence of the vertex $w \in N_G(u) \setminus X_A$, we have the following options:
    \begin{itemize}
        \item There exists an endpoint $x$ of $e$ that lies in $A \cap B$, $u \in \eta(e,x)$, and 
        $w \in \eta(x)$ or $w \in \bigcup_{y \in N_H(x) \cap (B \setminus A)} \eta(xy,x)$.
        Then, $u$ is adjacent to $a_{x,1}$ and $a_{x,2}$ in $G_A$, and therefore
        $u \in \eta_A'(\zeta_x) \cup \bigcup_{y \in N_{H_A'}(\zeta_x)} \eta_A'(\zeta_x y, \zeta_x)$.
        Thus, both $u$ and $w$ belong to $\eta^\ast(x) \cup \bigcup_{y \in N_{H^\ast}(x)} \eta^\ast(xy,x)$.
        \item Both endpoints $x,y$ of $e$ lie in $A \cap B$, $u \in \eta(xy,x) \cap \eta(xy,y)$ and there exists $z \in B \setminus A$
        with $xz,yz \in E(H)$ and $w \in \eta(xyz)$. 
        We infer that $u$ is adjacent to $a_{x,1}$, $a_{x,2}$, $a_{y,1}$, and $a_{y,2}$ in $G_A$. 
        The only way how $u$ is accommodated in $(H_A, \eta_A)$ is that $\zeta_x \zeta_y \in E(H_A)$
        and $u \in \eta_A(\zeta_x \zeta_y, \zeta_x) \cap \eta_A(\zeta_x\zeta_y, \zeta_y)$. 
        Hence, $u \in \eta^\ast(xy,x) \cap \eta^\ast(xy,y)$ while $w \in \eta^\ast(xyz)$. 
    \end{itemize}
    \item $u \in \eta(xyz)$ for some $xyz \in T(H)$. By the definition of $X_A$, we have $x,y,z \in A$.
     Then, $N_G(\eta(xyz)) \subseteq X_A$, contradicting the existence of $w$. Hence, this case is impossible.
\end{itemize}

For the second property, let $x \in A \cap B$ and $u,w \in \bigcup_{y \in N_{H^\ast}(x)} \eta^\ast(xy,x)$
with $u \in X_A$ and $w \notin X_A$. 
By the way we obtained $(H^\ast, \eta^\ast$), 
there exists $y_u \in N_{H_A'}(\zeta_x)$ with $u \in \eta_A(\zeta_x y_u, \zeta_x)$
and $y_w \in N_{H}(x) \cap (B \setminus A)$ with $w \in \eta(x y_w, x)$.
In particular, $u$ is adjacent to $a_{x,1}$ and $a_{x,2}$ in $G_A$. 
Therefore, by the way we constructed $G_A$, we have that either
$u \in \bigcup_{y \in N_H(x) \cap A} \eta(xy,x)$, 
or for $z \in Z$, either $u \in V_z^\bullet$, or $u=v$ and $v$ is completely adjacent to $\PotatoZ$. 
On all cases, $u$ is completely adjacent to $\bigcup_{y \in N_H(x) \cap (B \setminus A)} \eta(xy,x)$,
where $w$ resides. This proves $uw \in E(G)$, as desired. 

For the third property, consider $xyz \in T(H)$ with $x,y \in A \cap B$, $z \in B \setminus A$, and $\eta(xy) \neq \emptyset$.
Since $(H,\eta)$ is locally cleaned, there is an edge $uw \in E(G)$ with $w \in \eta(xyz)$ and $u \in \eta(xy,x) \cap \eta(xy,y)$. 
Then, $u \in X_A$ and $u$ is adjacent to $a_{x,1}$, $a_{x,2}$, $a_{y,1}$, and $a_{y,2}$ in $G_A$. 
Hence, the only option to accommodate $u$ in $(H_A,\eta_A)$ is that $\zeta_x \zeta_y \in E(H_A)$ and $u \in \eta_A(\zeta_x \zeta_y, \zeta_x) \cap \eta_A(\zeta_x \zeta_y, \zeta_y)$. This proves the third property.

Hence, $(H^\ast, \eta^\ast)$ is indeed a rigid extended strip decomposition of $G$. 
Since $\w(X_A) \leq 0.01\tau$ and every particle of $(H,\eta)$ has weight at most $0.01 \tau$, 
every particle of $(H^\ast, \eta^\ast)$ has weight at most $0.5\tau$. 
This proves that $(H^\ast, \eta^\ast)$ satisfies the requirements for the last outcome of Lemma~\ref{lem:esd:step}.
This concludes the proof of Lemma~\ref{lem:esd:step}, and thus also the proof of Lemma~\ref{lem:esd}.

\section{Algorithm}
%\pg{define recursion DAG. Give high level description of the algorithm. Find a nicer proof that we cant pack boosted balanced separators. Import \matching algorithm from QPTAS paper}
%\prz[inline]{For the reduction to matching, see Lemma 5.2 here: https://arxiv.org/pdf/2107.05434.pdf I believe we extracted this in a nicely citable version.}

In this section, we define $\log(n) = \max(2, \lceil \log_2(n) \rceil)$. Let $t$ be a positive integer, throughout this section, we will use $c_t$ to denote the constant given in Lemma~\ref{lem:esd} of the same name. In order to make dealing with constants easier (in particular the constants that arise from Definition~\ref{def:esd bs generation}), we will assume that $c_t \geq 34t$. Additionally, in this section we will assume that all graphs, $G$, come equipped with a weight function $\w : V(G) \to [0,+\infty)$. If $G'$ is an induced subgraph of $G'$, then we assume that $G$ inherits its weight function from $G$, that is the weight function for the vertices of $G'$ is the weight function for the vertices of $G$ when restricted to the vertices of $G'$. For a subset $X \subseteq V(G)$, $\w(X)$ = $\sum_{x \in X} \w(x)$.

\subsection{Definitions and Observations}

In this subsection we collect most of the definitions we will use for this section and immediate observations about these definitions.

The interpretation of $G$ and $G'$ in the coming definitions will be as follows: we are running the algorithm in order to find the maximum size independent set of $G$. The algorithm is recursive, and only makes recursive calls on induced subgraphs of the input graph. Suppose we want to analyze a recursive call in which the current induced subgraph of $G$ that is being considered is $G'$. When arguing about the behavior of the algorithm on input $G'$, it is useful to be able to conclude that the bigger graph $G$ contains an $\Sttt$.

%\mpil[inline]{The notion of a $c$-balanced separator is used also in Section~\ref{sec:esd} and should be probably moved
%to preliminaries.}
%\pg{done. I believe what I wrote in the prelims matches how we both use the definition, but please verify it works for your section.}

%\begin{definition}[Balanced Separator, Core]\label{def:bs and bbs}
%Let $G$ be a graph, $G'$ an induced subgraph of $G$, and $Y,Z \subseteq V(G')$ such that no component of $G'-Y$ contains over $c$ vertices of $Z$. Then we say that $Y$ is a {\em $c$-balanced separator for $(G',Z)$}. If there is a set $X \subseteq V(G)$ such that $Y = N_G^{G'}[X]$ then we say that {\em $Y$ has a core $X$ originating in $G$}. When $Z = V(G')$ then we say that $Y$ is a {\em $c$-balanced separator for $G'$ with a core $X$ originating in $G$} and when $G = G'$ then we say that $Y$ is a {\em $c$-balanced separator for $G'$ with a core $X$}.

Throughout our algorithm, balanced separators as defined in Section~\ref{sec:prelims} %Definition~\ref{def:bs and bbs} 
will often be readily available. However we will sometimes need balanced separators with even stronger properties; in particular we will need the amount that the separator disconnects the graph to depend super-linearly in the size of its core. We will call such balanced separators {\em boosted} (see Definition~\ref{def:bbs} below), and a substantial part of our algorithm will consist of trying to reduce the input graph $G$ so that some vertex set becomes boosted (in the sense of Definition~\ref{def:bbs}). Unlike for normal balanced separators, the following definition is always used with $Z=V(G')$, so reference to $Z$ is dropped in the following definition.

%\mpil[inline]{I was confused here a bit that this definition stops to use $Z$ as before. I would comment that we will use it only for $Z=V(G')$, so we drop the notation.}
\begin{definition}[$s$-boosted balanced separator]\label{def:bbs}
Let $G$ be a graph, $G'$ an induced subgraph of $G$, and $s$ be a positive integer.
A vertex set $Y \subseteq V(G')$ is an {\em $s$-boosted balanced separator for $G'$ with a core $X$ originating in $G$} 
if $Y$ is a $c$-balanced separator for $G'$ with core $X$ originating in $G$, $|X| \leq s$, and $c \leq \frac{|C|}{16s^2}$, where $C$ is a largest component of $G'$. 
When $G$ and $G'$ are clear from context, we may say that $X$ is a core of the boosted balanced separator, $Y$.
\end{definition}

The algorithm will often work with a graph $G$, a vertex set $X$ in $G$ and an induced subgraph $G'$ of $G$. The aim is to ensure that $X$ is a core of an $s$-boosted balanced separator (for an appropriately chosen $s$) of $G'$. 
If $X$ doesn't already satisfy this, it is because $G' - Y$, where $Y = N_G^{G'}[X]$, has some connected components that are too big. The next definition zooms in on the neighborhood of these connected components into $Y$ (the constants in the formal definition don't quite match the intuition above for book keeping reasons). 

\begin{definition}[relevant set]\label{def:relevant}
Let $t$ be a positive integer, $G$ an $\Sttt$-free graph, $X \subseteq V(G)$, $G'$ an induced subgraph of $G$, and $N$ a positive integer. We define \rel$_G(G', X, N)$ to be the subset of vertices of $N_G^{G'}[X]$ that have at least one neighbor in at least one component of $G'-N_G^{G'}[X]$ that contains over $\frac{N}{32c_t^2\log^2(N)}$ vertices.  
\end{definition}

We make the following important observation about \rel$_G(G', X, N)$ that follows directly from the fact that if $G''$ is an induced subgraph of $G'$, then every component of $G''-N_G^{G''}[X]$ of size at least $\frac{N}{32c_t^2\log^2(N)}$ vertices is contained in some component of $G'-N_G^{G'}[X]$ of size at least $\frac{N}{32c_t^2\log^2(N)}$.

\begin{observation}
Let $t$ be a positive integer, $G$ an $\Sttt$-free graph, $X \subseteq V(G)$, $G'$ an induced subgraph of $G$, $G''$ and induced subgraph of $G'$, and $N$ a positive integer. Then \rel$_G(G'', X, N)$ $\subseteq$ \rel$_G(G', X, N)$.
\end{observation}
 
Note that in an $\Sttt$-free graph, Theorem~\ref{thm:ICALP:weight} will always return a family ${\cal F}$ satisfying the second bullet point of the theorem statement. 
%Since the size of ${\cal F}$ is $n^{\Oh(\log(n))}$ we can in $n^{\Oh(\log(n))}$ time pick out the vertex set $X$ and corresponding extended strip decomposition $(H, \eta)$ of $G - N_G[X]$ such that no particle of $(H, \eta)$ has over $|G|/2$ vertices. 
This motivates the following definition.

\begin{definition}[\esdthm~ and inferred extended strip decomposition]\label{def:esd bs generation}
Let $t$ be a positive integer and $G$ an $n$-vertex $\Sttt$-free graph. We define \esdthm$(G)$ to be a subroutine that uses Theorem~\ref{thm:ICALP:weight} to return a set $X \subseteq V(G)$, $|X| \leq (t+1)(11\log(n) + 6) \leq 34t\log(n)$ $\leq$ $c_t\log(n)$ such that $G-N_G[X]$ has a rigid extended strip decomposition, $(H, \eta)$, where no particle of $(H, \eta)$ has over $|G|/2$ vertices. Furthermore, this subroutine runs in time polynomial time. 

Additionally, for any induced subgraph $G'$ of $G$, we define the extended strip decomposition {\em inferred by $(X,G')$}, call it $(H', \eta')$. For each component, $C$, of $G'$ that does not contain a vertex of $N^{G'}_{G}[X]$, $H'$ contains an isolated copy $H_c$ of $H$, and for all vertices, edges, and triangles, ${\cal R}_c$, of $H_c$ let ${\cal R}$ be the corresponding vertex, edge, or triangle in $H$; we set $\eta'({\cal R}_c)$ = $\eta({\cal R}) \cap C$. For each component $C^*$ of $G'$ that contain at least one vertex of $N^{G'}_{G}[X]$, $H'$ contains an isolated vertex $v_{c^*}$ and $\eta'(v_{c^*})$ = $C^*$. 
\end{definition}

It follows from the definition of extended strip decomposition that $(H', \eta')$ is a valid extended strip decomposition of $G'$. Note that the extended strip decomposition inferred by $(X, G')$ can be computed in polynomial time since we have access to the extended strip decomposition, $(H, \eta)$, by Theorem~\ref{thm:ICALP:weight} and since $H$ has only $\Oh(n)$ vertices (because it is rigid, see the discussion after Theorem~\ref{thm:ICALP:weight}), $H'$ has $n^{\Oh(1)}$ vertices and therefore $n^{\Oh(1)}$ particles. Furthermore, note that for any particle $P$ of $(H', \eta')$ either $P$ is equal so some component $C^*$ that contains at least one vertex of $N^{G'}_{G}[X]$ (when $P = \eta'(v_{c^*})$) or $P$ is equal to to $P' \cap C$ where $P'$ is a particle of $(H, \eta)$ and $C$ is a component of $G'$ that does note contain any vertices of $N^{G'}_{G}[X]$. This leads to the next two observation. But we first give one additional definition related to extended strip decompositions which is meant to capture when each particle of our extended strip decomposition is ``small enough'' so that we make enough progress when we recursively call the algorithm on each particle.

Let $t$ be a positive integer, $G$ an $\Sttt$-free graph, $N$ a natural number, and $(H, \eta)$ an extended strip decomposition of $G$. We say that $(H, \eta)$ is an {\em $N$-good} extended strip decomposition of $G$ if no particle of $(H, \eta)$ has over $(1-\frac{1}{32c_t^2\log^2(N)})N$ vertices of $G$. Note that these are the same constants used in the definition of $\rel$, the reason for this will become apparent in Observation~\ref{obs:bbs or good esd}.

The first observation follows from the fact that the size of the largest particle of the extended strip decomposition inferred by $(X,G')$ is bounded by the size of the largest component of $G'$.

\begin{observation}\label{obs:no large componet give good esd}
Let $t$ be a positive integer, $N$ a natural number, $G$ an $\Sttt$-free graph, $G'$ an induced subgraph of $G$, and $X \subseteq V(G)$. If no component of $G'$ has over $(1-\frac{1}{32c_t^2\log^2(N)})N$ vertices then the extended strip decomposition inferred by $(X, G')$ is $N$-good.
\end{observation}

\begin{observation}\label{obs:bbs or good esd}
Let $t$ be a positive integer, $G$ an $n$-vertex $\Sttt$-free graph, $N$ a natural number, $G'$ an induced subgraph of $G$, and $X \subseteq V(G)$ such that $G-N_G[X]$ has an extended strip decomposition, $(H, \eta)$, where no particle contains over $N/2$ vertices. If $\rel_G(G',X, N) = \emptyset$ then either $N_G^{G'}[X]$ is an $\frac{N}{32c_t^2\log^2(N)}$-balanced separator for $G'$ or the extended strip decomposition inferred by $(X, G')$ is $N$-good. 
\end{observation}

\begin{proof}
Let $t$, $N$, $G$, $G'$, $X$, and $(H, \eta)$ be as in the statement of the lemma and $(H', \eta')$ be the extended strip decomposition inferred by $(X, G')$. If there are no components of $G'-N_{G}^{G'}[X]$ that contain at least $\frac{N}{32c_t^2\log^2(N)}$ of the vertices of $G$ then we are done (as $N_{G}^{G'}[X]$ would be an $\frac{N}{32c_t^2\log^2(N)}$ balanced separator). So we may assume that there is a component $C$ of $G'-N_{G}^{G'}[X]$ that contains at least $\frac{N}{32c_t^2\log^2(N)}$ vertices of $G$ and $N_{G'}[C] \cap N_{G}^{G'}[X] = \emptyset$ by the assumption that $\rel_G(G',X, N) = \emptyset$. It follows that $C$ is a component of $G'$ and therefore any component of $G'$ that contains at least one vertex of $N_{G}^{G'}[X]$ has at most $(1-\frac{1}{32c_t^2\log^2(N)})N$ vertices. This combined with the note about particles made just after Definition~\ref{def:esd bs generation} (that every particle of $(H', \eta')$ is either a component, $C$, of $G$ that contains at least one vertex of $N_{G}^{G'}[X]$ - which can have size at most $(1-\frac{1}{32c_t^2\log^2(N)})N$ - or a subset of a particle of $(H, \eta)$ - which by assumption has at most $N/2$ vertices) proves the observation.
\end{proof}

%\begin{proof}
%Let $G$, $X$, $(H, \eta)$, $G'$, and $(H', \eta')$ be as in the statement of this lemma. We prove $(H', \eta')$ is an extended strip decomposition of $G'$. Let $u,v$ in $G'$, we must first verify that if $uv \notin E$ then $u$ and $v$ do not belong to $\eta(ab, a)$ and $\eta(ac, a)$ respectively for distinct vertices $a,b,c \in H'$. If either $u$ or $v$ belongs to a component, $C$, that contains a vertex of 
%\end{proof}

\subsection{Preliminary Lemmas}

We will now present a few lemmas that will be useful to have in hand before describing the algorithm. Given a graph $G$ and an extended strip decomposition $(H, \eta)$ for $G$ the following lemma shows that solving independent set on $G$ can be reduced to solving independent set on each particle of $G$. This reduction first appears in \cite{DBLP:conf/soda/ChudnovskyPPT20}, the version we cite here is derived from \cite{ACDR21} (Lemma 5.2).

\begin{lemma}[\cite{DBLP:conf/soda/ChudnovskyPPT20, ACDR21}]\label{lem:particle matching}
Let $G$ be an $n$-vertex graph and let $(H, \eta)$ be an extended strip decomposition of $G$ where $H$ has $N$ vertices. Furthermore, assume that for each particle, $P$, of $(H, \eta)$, we know the weight of a maximum weight independent set of $G[P]$. Then in time polynomial in $n+N$ we can compute the weight of a maximum weight independent set for $G$.
\end{lemma}

Let $G$ be a graph and $(H,\eta)$ an extended strip decomposition for $G$ such that for each particle, $P$, of $(H,\eta)$ the weight of a maximum weight independent set of $G[P]$ is known. We use $\matching(H, \eta)$ to denote the output (the weight of a maximum weight independent set of $G$) of running the algorithm of Lemma~\ref{lem:particle matching}.

%\pg{Daniel thinks that lemma~\ref{lem: balanced sep or esd} should probably be the main theorem of section \ref{sec:esd} since it is "more citable". It follows with a short proof from the main lemma of section \ref{sec:esd}.  (move to the section \ref{sec:esd}?}

%\mpil[inline]{I would not say it should be the main theorem of Section~\ref{sec:esd}, because the current one sends clear message
%and is comparable to the ICALP paper, but we can move this boosting to Section~\ref{sec:esd} a variant of the main result
%that in some contexts is more useful.}

%Marcin: I won't be doing this change for the FOCS submission, we can do it for arxiv later.

\begin{lemma}\label{lem: balanced sep or esd}
Let $t$ be an positive integer, $G$ an $n$-vertex $\Sttt$-free graph, $A \subseteq V(G)$, and $i \leq \log(n)$ a natural number. Either $G$ contains a set $C$ such that $N_G[C]$ is an $(|A|/2^i)$-balanced separator for $(G, A)$ and $|C| \leq (c_t)(70)2^{i+1}\log(n)$ or $G$ has a rigid extended strip decomposition, $(H, \eta)$, such that no particle contains over $(1-1/2^{i+2})|A|$ vertices of $A$. Furthermore, either $C$ or $(H, \eta)$ can be found in polynomial time.
\end{lemma}

\begin{proof}
Let $t$, $G$, $n$, and $A$ be as in the statement of this lemma. We first claim that either $G$ contains a set $C$ such that $N_G[C]$ is an $(|A|/2)$-balanced separator for $(G, A)$ and $|C| \leq 70c_t\log(n)$ or $G$ has an extended strip decomposition, $(H, \eta)$, such that no particle contains over $(1-1/4)|A|$ vertices of $A$ and either $C$ or $(H, \eta)$ can be found in polynomial time. 

In order the prove this we will consider a process consisting of at most 70 step. At the $j^{th}$ step we will assume we have a set $C_j$ such that $N_G[C_j]$ is an $(|A|\cdot 0.99^j)$-balanced separator for $(G,A)$ and $|C| \leq c_tj\log(n)$ (we have $C_0 = \emptyset$ for the base case). Given such a  $C_j$, we show how to find $C_{j+1}$ or find a rigid extended strip decomposition $(H, \eta)$ of $G$ such that no particle has over $(1-1/4)|A|$ vertices of $A$. If $N_G[C_j]$ is already an $|A|/2$-balanced separator for $(G,A)$ then we are done, so assume this does not happen, let $X$ be the component of $G-N_G[C_j]$ that contains over half the vertices of $A$ and let $X_A = X \cap A$. 

We apply Lemma~\ref{lem:esd} to $G$ where all vertices of $X_A$ have weight 1 and all other vertices have weight 0. Outcome $(1)$ cannot occur as $G$ is $\Sttt$-free. If outcome $(2)$ occurs, then we get a set $X_C$ such that $N_G[X_C]$ is an $(|X_A|\cdot 0.99)$-balanced separator for $(G,X_A)$ and $|X_C| \leq c_t\log(n)$. We set $C_{j+1} = C_j \cup X_C$, it holds then that $N_G[C_{j+1}]$ is an $(|A|\cdot 0.99^{j+1})$-balanced separator for $(G,A)$ and $|C_{j+1}|$ $\leq$ $|C_t| + |X_c| \leq c_t(j+1)\log(n)$, as desired. If outcome $(3)$ occurs then we get a rigid extended strip decomposition $(H, \eta)$ for $G$ such that no particle of $(H, \eta)$ contains over half of $X_A$. Since $|X_A| \geq |A|/2$ it follows that no particle of $(H, \eta)$ contains over $(1-1/4)|A|$ vertices of $A$. Since $0.99^{70} < .5$ this process must end by the $70^{th}$ step. Since Lemma~\ref{lem:esd} runs in polynomial time, this process runs in polynomial time.

We now prove the full statement of this lemma in a similar manner. Fix some natural number $i \leq \log(n)$. %, we show either $G$ contains a set $C$ such that $N_G[C]$ is an $(|A|/2^i)$-balanced separator for $(G, A)$ and $|C| \leq (c_t)(70)2^{i+1}$ or $G$ has an extended strip decomposition, $(H, \eta)$, such that no particle contains over $(1-1/2^{i+2})|A|$ vertices of $A$.
In order the prove this we will consider a process consisting of at most $i$ step. At the $j^{th}$ step, $j < i$, we will assume we have a set $C_j$ such that $N_G[C_j]$ is an $|A|/2^j$-balanced separator for $(G,A)$ and $|C| \leq 70c_t\cdot 2^{j+1}\log(n)$ (we have $C_0 = \emptyset$ for the base case). So, assume $C_j$ satisfies these properties, we show how to find $C_{j+1}$ or find a rigid extended strip decomposition $(H, \eta)$ of $G$ such that no particle has over $(1-1/2^{j+2})|A|$ vertices of $A$ (we have $C_0 = \emptyset$ for the base case). 

Consider each component, $X$, of $G-N_G[C_j]$ that contain at least $|A|/2^{j+1}$ vertices of $|A|$, there are at most $2^{j+1}$ such components, set $X_A = X \cap A$. For each $X$ and corresponding $X_A$ we apply the claim from the first paragraph of this proof to $G$ where all vertices of $X_A$ have weight 1 and all other vertices have weight 0. The first possibility is for each $X$ and $X_A$ we get a set $X_C$ such that $N_G[X_C]$ is an $|X_A|/2$-balanced separator for $(G,X_A)$ and $|X_C| \leq 70c_t\log(n)$. Then we set $C_{j+1} = C_j \cup \bigcup\limits_{X} X_C$, it holds then that $N_G[C_{j+1}]$ is an $(|A|/2^{j+1})$-balanced separator for $(G,A)$ and $$|C_{j+1}| \leq |C_j| + \sum_{X} |X_C|\leq70c_t\cdot 2^{j+1}\log(n) + 70c_t\cdot 2^{j+1}\log(n)=70c_t\cdot 2^{j+2}\log(n),$$ as desired. The other possibility is that for at least one $X$ we get a rigid extended strip decomposition $(H, \eta)$ for $G$ such that no particle of $(H, \eta)$ contains over half of $X_A$. Since $|X_A| \geq |A|/2^{j+1}$ it follows that no particle of $(H, \eta)$ contains over $(1-1/2^{j+3})|A|$ $\leq$ $(1-1/2^{i+2})$ vertices of $A$ (since $j < i$).

Repeating this $i \leq \log(n)$ times (or until we get a desired extended strip decomposition) then yields the result. Since each step applies Lemma~\ref{lem:esd} less than $n$ time  and  Lemma~\ref{lem:esd} runs in polynomial time and there are at most $\log(n)$ steps, this process runs in polynomial time.
\end{proof}

\subsubsection{Cannot Pack Many Balanced Separators}

Recall the following notation from Section~\ref{sec:prelims}: Let $G$ be a graph, $G'$ an induced subgraph of $G$, and $X \subseteq$ $V(G)$. We define $N_G^{G'}[X]$ to mean $N_G[X] \cap V(G')$. Furthermore, let ${\cal F} = \{F_1, F_2, \ldots, F_k\}$ be a list of vertex sets of $G$. Then $N_G^{G'}[{\cal F}] = \{N_G^{G'}[F_1], N_G^{G'}[F_2], \ldots, N_G^{G'}[F_k]\}$. 

\begin{lemma}\label{lem:cant pack strong bs}
Let $t$ be a positive integer, $G$ an $n$-vertex $\Sttt$-free graph, $N \leq |G|$ a natural number, $G'$ an induced subgraph of $G$, $X \subseteq V(G)$, and ${\cal F}$ a list of $G$. Assume that $\rel_G(G',X, N) \neq \emptyset$ and that all sets of $N_{G}^{G'}[{\cal F}]$ are $\frac{|\rel_G(G',X,N)|}{100c_t^2\log^2(N)|X|}$-balanced separators for $(G',$ $\rel_G(G'$,$X,N))$ such that no vertex of $G'$ belongs to over $\bar{c}$ sets of $N_{G}^{G'}[{\cal F}]$ for some positive integer $\bar{c}$. Assume $|{\cal F|} \geq 10t\bar{c}$. Then $G$ contains an induced $\Sttt$.
\end{lemma}

\begin{proof}
Let $t, G, G', N, X$, ${\cal F}$, and $\bar{c}$ have the same meaning as in the statement of this lemma. Among all components of $G'-N_{G}^{G'}[X]$ that have at least $\frac{|G'|}{32c_t^2\log^2(N)}$ vertices, let $C$ denote the one such that the size of $C^* = N_{G'}[C] \cap \rel_G(G',X,N)$ is maximized. Since there are at most $32c_t^2\log^2(N)$ components of $G'-N_{G}^{G'}[X]$ that have at least $\frac{|G'|}{32c_t^2\log^2(N)}$ vertices and by definition all vertices of $\rel_G(G',X, N)$ have at least one neighbor in a component of $G'-N_{G}^{G'}[X]$ of size at least $\frac{|G'|}{32c_t^2\log^2(N)}$, it holds that $|C^*| \geq \frac{|\rel_G(G',X,N)|}{32c_t^2\log^2(N)}$. Next for all vertices in $X$ let $x$ be one such that the size of $C^{*,x} = N_{G}^{G'}[x] \cap C^*$ is maximized, since $N_G^{G'}[X]$ dominates $C^*$ it holds that $|C^{*,x}| \geq \frac{|C^*|}{|X|} \geq \frac{|\rel_G(G',X,N)|}{32c_t^2\log^2(N)|X|}$ $>$ $0$ (the last inequality following from the assumption $\rel_G(G',X, N) \neq \emptyset$).

\begin{claim}\label{claim:good sublist}
There exist a triple of vertices $a, b, c \in C^{*,x}$ and a subfamily of ${\cal F}$, call it ${\cal F}^*$,
of size greater than $2t\bar{c}$, such that no two vertices among $a, b,$ and $c$ are in the same component in $G-N^{G'}_G[F^*]$ for all $F^* \in {\cal F}^*$ (but possibly with some or all of $a,b,$ and $c$ belonging to $F^*$).
\end{claim}

In order to prove Claim \ref{claim:good sublist}, we let $F \in {\cal F}$ and let $a,b,c$ be three independently and uniformly at random (with replacement) chosen vertices of $C^{*,x}$ (so it is possible that, for instance, $a$ = $b$). We first calculate the probability that no two vertices among $a,b,$ and $c$ belong to the same component in $G' - N_{G}^{G'}[F]$. 
Since $|C^{*,x}| \geq \frac{|\rel_G(G',X,N)|}{32c_t^2\log^2(N)|X|}$, $N_{G}^{G'}[F]$ is a  $\frac{|\rel_G(G',X,N)|}{100c_t^2\log^2(N)|X|}$-balanced separator for $(G'$, $\rel_G(G',X,N))$, and $C^{*,x} \subseteq \rel_G(G',X,N)$, we have that $N_{G}^{G'}[F]$ is a $\frac{|C^{*,x}|}{3}$-balanced separator for $(G', C^{*,x})$. So, since no component of $G-N_{G}^{G'}[F]$ has over $\frac{|C^{*,x}|}{3}$ vertices of $C^{*,x}$ there is at least a $\frac{2}{3}$ probability that $a$ and $b$ do not belong to the same component, either because $a$ and $b$ are in different components of $G-N_{G}^{G'}[F]$ or at least one of $a$ and $b$ is in $F$. Furthermore, conditioned on $a$ and $b$ not being in the same component of $G-N_{G}^{G'}[F]$, there is at least a $\frac{1}{3}$ probability that $c$ is not in the same component as $a$ or $b$. It follows that there is at least a $\frac{2}{3} \cdot \frac{1}{3} = \frac{2}{9}$ probability that no two vertices among $a,b,$ and $c$ belong to the same component in $G' - N_{G}^{G'}[F]$ (again, possibly with some or even all of $a,b,$ and $c$ belong to $N_{G}^{G'}[F]$).

Hence if $X_F$ represents the random variable that is 1 if no two of the independently and uniformly at random chosen $a,b,c \in C^{*,x}$ (with replacement) are in the same component in $G' - N_{G}^{G'}[F]$ and 0 otherwise, the expected value $\mathbb{E}[X_F] \geq \frac{2}{9}$. Then by the linearity of expectation, we have that $\mathbb{E}[\sum_{F \in {\cal F}} X_F]$ $\geq$ $\frac{2}{9} \cdot 10t\bar{c} > 2t\bar{c}$. Thus, there must exists a triple, $a,b,c \in C^{*,x}$, such that for a subset of ${\cal F}$, call it ${\cal F}^*$, of size greater than $2t\bar{c}$, no two of $a,b$, and $c$ are in the same component in $G' - N_{G}^{G'}[F^*]$ for all $F^* \in {\cal F}^*$. This completes the proof of Claim \ref{claim:good sublist}.

Now, let $a, b, c \in C^{*,x}$ and ${\cal F}^*$ be as in the statement of Claim \ref{claim:good sublist}. We have that for any path $P$ in $G'$ with $a$ and $b$ as its endpoints must have over $2t$ vertices because for all $F^* \in {\cal F}^*$ $N_{G}^{G'}[F^*] \cap P \neq \emptyset$ (or else $a$ and $b$ would be in the same component of $G' - N_{G}^{G'}[f^*]$) and if $P$ had at most $2t$ vertices, since $|{\cal F}^*|>2t\bar{c}$, that would force some vertex of $P$ to belong to over $\bar{c}$ sets in $N_G^{G'}[{\cal F}^*]$, contrary to assumption. Similarly, all paths with $a$ and $c$ or $b$ and $c$ with its endpoint must have over $2t$ vertices as well.

Now we show there exists three anti-complete induced paths, $P_a$, $P_b$ and $P_c$, each with $t$ vertices such that $a,b,$ and $c$ are one of the endpoints of $P_a$, $P_b$ and $P_c$ respectively, and all other vertices of these paths belong to $C$ (recall from the first paragraph of this proof that $C$ is the component of $G'-N_{G}^{G'}[X]$ such that $C^* = N_{G'}[C] \cap \rel_G(G',X,N)$. Since $C^{*,x} \subseteq C^*$, all vertices of $C^{*,x}$, which includes $a,b,$ and $c$, have a neighbor in $C$). To locate $P_a$, take a shortest path from $a$ to $b$ with all internal vertices in $C$ (since $a$ and $b$ both have neighbors in $C$ and $C$ is connected, such a path exists). By the previous paragraph this path must have at least $2t$ vertices, so let $P_a$ be the first $t$ vertices of this path, so $P_a$ is a path with $t$ vertices such that $a$ is one endpoint of the path and all other vertices are in $C$. Identical arguments show there are induced paths $P_b$ and $P_c$ which have $b$ and $c$ as their endpoints, respectively, and all other other vertices are in $C$. Furthermore, if $P_a$, $P_b$ and $P_c$ were not anti-complete, then that would imply that there exists paths between two vertices of $\{a,b,c\}$ with at most $2t$ vertices, which contradicts the conclusion of the previous paragraph.

Lastly, note that since $C$ is a component of $G'-N_{G}^{G'}[X]$ and $x \in X$, $x$ has no neighbors in $C$. So, since $x$ is neighbors with $a$, $b$, and $c$ and $P_a$, $P_b$ and $P_c$ are anti-complete, $x$ along with $P_a$, $P_b$ and $P_c$ form an $\Sttt$.
\end{proof}

\subsubsection{Cannot Pack Many Boosted Balanced Separators}\label{sec:cannot pack bbs}

%\todo[inline]{DL: need to adapt boosted to $s$-boosted in rest of the paper: core has size at most $s$ and no component of $G-Y$ has at least $\frac{C}{16s^2}$ vertices of $C$, where $C$ is largest component.
%}
%\mpil[inline]{I think this is done now, isn't it?}

\begin{lemma}\label{lem:cant pack bbs}
Let $G$ be a graph and $G'$ be an induced subgraph of $G$, and $s, t, c$ be positive integers. If there exists a list ${\cal F}$ of $s$-boosted balanced separators of $G'$ originating in $G$ such that
$|{\cal F}| \geq 80 \cdot s\cdot t\cdot c$, and no vertex of $G'$ belongs to over $c$ sets of ${\cal F}$,
%. If 
%the size of the largest set in ${\cal F}$ is $d$ and 
%$|{\cal F|}$ $\geq$ $t^?c^?$ 
then $G$ contains an $\Sttt$.
\end{lemma}

%\begin{lemma}[OLD STATEMENT]\label{lem:cant pack bbs OLD}
%Let $G$ be a graph and $G'$ be an induced subgraph of $G$. Assume that there is a list, ${\cal F}$, of boosted balanced separators of $G'$ with cores originating in $G$ such that no vertex of $G'$ belongs to over $c$ sets of ${\cal F}$. If 
%%%%%%the size of the largest set in ${\cal F}$ is $d$ and 
%$|{\cal F|}$ $\geq$ $t^?c^?$ then $G$ contains an $\Sttt$.
%\end{lemma}

This subsection is devoted to the proof of Lemma~\ref{lem:cant pack bbs}. Thus, within this subsection we will assume that the premise of Lemma~\ref{lem:cant pack bbs} holds. 
Towards the proof of Lemma~\ref{lem:cant pack bbs}, we set ${\cal F} = \{Y_1, Y_2, \ldots Y_\ell\}$. For every $Y_i$ in ${\cal F}$ we let $X_i$ be a core of $Y_i$ originating in $G$. %\todo{am I saying this correctly? PG: $Y_i$ should be 'a' core not 'the' core. also you should say that $Y_i$ originates in $G$ (i.e. $Y_i \subseteq V(G)$) since $Y_i$ might not be in $G'$}. 
In other words, $Y_i = N_G^{G'}[X_i]$, and $|X_i| \leq s$. Additionally, $C$ is the largest component of $G'$, and $r = 4s$.

\begin{lemma}\label{lem:getSplitSetR}
There exists a set $R \subseteq C$ of size $r$ and a subfamily ${\cal F}_1 \subseteq {\cal F}$ such that
$|{\cal F}_1| > |{\cal F}|/2$, and for every $Y_i \in {\cal F}_1$, each connected component of $G' - Y_i$ contains at most one vertex of $R$.
\end{lemma}

\begin{proof}
We pick a tuple $R = v_1, v_2, \ldots, v_r$ of $r = 4s$ vertices from $C$ uniformly at random (with repetition).
Consider an arbitrary $Y_i \in {\cal F}$.
For each choice of $1 \leq p < q \leq r$, the probability that $v_q$ is in the same component of $G' - Y_i$ as $v_p$ is at most $\frac{|C|}{16s^2}$.
The union bound over all choices of $p$, $q$ yields that the probability that no component of $G' - Y_i$ contains at least two vertices of $R$ is at least $1 - {r \choose 2} \cdot \frac{1}{16s^2} > 1/2$.
Define ${\cal F}_1$ to be the family of all $Y_i$'s in ${\cal F}$ such that no component of $G' - Y_i$ contains at least two vertices of $R$.
The expected size of ${\cal F}_1$ is strictly larger than $|{\cal F}|/2$, and there exists at least one choice of $R$ that achieves expectation, proving the statement of the lemma. 
\end{proof}

We will use the following well-known facts about trees, that we will state without proof.

% Marcin: commenting out for FOCS submission.
% I think this is actually OK to have it here,
% the reader has it in their context.
%\todo{(preliminaries?)}
%\mpil[inline]{Actually, I would not state them even as statements, just use inline in the argumentation.}
%\pg{daniel wrote this subsection. hoping he sees this and will figure out what the best things to do is}
%\begin{observation}[folklore]
\begin{observation}\label{obs:deg3inTree} A tree with $k$ leaves has at most $k-1$ vertices of degree at least $3$.
\end{observation}

\begin{observation}\label{obs:delPath}
Let $T$ be a tree and $P$ be a path in the tree such that all vertices on $P$ have degree $2$ in $T$. Then $T - V(P)$ has precisely two connected components.
\end{observation}

For the rest of the proof of Lemma~\ref{lem:cant pack bbs}, let $R$ and ${\cal F}_1$ be as given in the statement of Lemma~\ref{lem:getSplitSetR}, $G^*$ be an inclusion minimal connected induced subgraph of $G'[C]$ containing $R$, and $T^*$ be a spanning tree of $G^*$.
Define $M$ as $R$ plus all the vertices of $T^*$ that have degree at least three in $T^*$. Finally, set $M^* = N_{T^*}[M]$. In the next lemma we collect a few simple observations about $G^*$, $T^*$, $M$, and~$M^*$.

\begin{lemma}\label{lem:sillyObsAboutGstar}
$G^*$, $T^*$, $M$ and $M^*$ have the following properties:
\begin{enumerate}\setlength\itemsep{-.7pt}
    \item\label{itm:leavesInR} All leaves of $T^*$ are in $R$,
    \item\label{itm:smallM} $|M| \leq 2|R|$,
    \item\label{itm:fewComponentsTminusM} there are at most $|M|-1$ connected components of $T^*-M$,
    \item\label{itm:smallMstar} $|M^*| \leq 6|R|$,
    \item\label{itm:noDumbEdges} each edge $uv$ of $G^*$ is an edge of $T^*$ or has at least one endpoint in $M^*$.
\end{enumerate}
\end{lemma}

\begin{proof}
For (\ref{itm:leavesInR}), note that $G^* - v$ is connected for every leaf $v$ of $T^*$. Thus $v \notin R$ would contradict minimality of $G^*$.
For (\ref{itm:smallM}) we note that (\ref{itm:leavesInR}) implies that $T^*$ has at most $|R|$ leaves, and therefore (by Observation~\ref{obs:deg3inTree}) at most $|R|-1$ vertices of degree at least $3$.

For (\ref{itm:fewComponentsTminusM}) and (\ref{itm:smallMstar}) we note that every vertex in $V(T^*) - M$ has degree precisely $2$ in $T^*$. Thus every component $P$ of $T^* - M$ is a path (on one or more vertices), only the endpoints of the path are neighbors (in $T^*$) of $M$, and $|N_{T^*}(P)| = 2$.
Let $\kappa$ be the number of connected components of $T^* - M$. Then $\kappa$ applications of Observation~\ref{obs:delPath} implies that $T^*[M]$ has at least $1+\kappa$ components. Hence $1+\kappa \leq |M|$, proving (\ref{itm:fewComponentsTminusM}).
%This yields 
%Hence repeated application of Observation~\ref{} implies that there are  at most $|M|-1$ connected components  of $T^* - M$ (proving \ref{itm:fewComponentsTminusM}).
%
Since each component of $T^*-M$ contains at most two neighbors (in $T^*$) of $M$  it follows that 
$|N_{T^*}(M)| \leq 2|M|$, and hence $|M^*| \leq 3|M| \leq 6|R|$, proving (\ref{itm:smallMstar}).

For~(\ref{itm:noDumbEdges}) suppose for contradiction that there exists an edge $uv$ in $G^*$ that is neither an edge in $T^*$ nor has an endpoint in $M^*$.
Let $P$ be the path in $T^*$ from $u$ to $v$. Since $uv$ is not an edge of $T^*$ the path $P$ has at least one internal vertex. Let $u'$ be the vertex immediately after $u$ on $P$.
Since $u \notin M^*$ and $u \in N_{T^*}(u')$ it follows that $u' \notin M$. Thus $u' \notin R$ and $u'$ has degree precisely $2$ in $T^*$.
But then $(T^* - u') \cup \{uv\}$ is connected (since we can go between $u$ and the successor of $u'$ on $P$ via $uv$ and then $P$) and contains $R$, contradicting the minimality of $G^*$.
\end{proof}

We are now ready to conclude the proof of Lemma~\ref{lem:cant pack bbs}

\begin{proof}[Proof of Lemma~\ref{lem:cant pack bbs}]
We set $M'$ to be the set of all vertices in $T^*$ at distance (in $T^*$) at most $t$ from $M^*$.
By Lemma~\ref{lem:sillyObsAboutGstar} (point~\ref{itm:fewComponentsTminusM})  we have that~$|M'| \leq |M^*| + 2t|M| \leq 6r + 4tr$. 
Since $|{\cal F}_1| > |{\cal F}/2| \geq r(4t + 6)c$ there exists a $Y_i \in {\cal F}_1$ disjoint from $M'$.

Let $z$ be the number of connected components of $T^* - M$ that have nonempty intersection with $Y_i$, and $Z$ be the union of the vertex sets of all such components. 
Since each component of $T^* - M$ is a path of vertices of degree $2$ in $T^*$, $z$ applications of Observation~\ref{obs:delPath} yield that $T^* - Z$ has precisely $z + 1$ connected components.
Since $Y_i \in {\cal F}_1$ we have that no connected component of $G^* - Y_i$ contains two vertices of $R$.
Thus no connected component of $T^* - Y_i$ contains two vertices of $R$, and $Y_i$ is disjoint from $M' \supseteq M$, so $Y_i \cap V(T^*) \subseteq Z$ and no connected component of $T^* - Z$ contains two vertices of $R$ either.
But then $T^* - Z$ has at least $r$ connected components, implying $z + 1 \geq r$. 

Since $Y_i \subseteq N_G[X_i]$ and $|X_i| \leq s$, and $z \geq r-1 = 4s-1 > 2|X_i|$,
%\todo{ok with $r=4s$}
there exists an $x \in X_i$ such that $N_G[x]$ has nonempty intersection with three distinct components $C_1$, $C_2$ and $C_3$ of $T^* - M$.
For each $j \in \{1,2,3\}$ define $P_j$ to be a shortest path in $T^*[C_j]$ from $N_G[x]$ to $N_{T^*}[M^*]$.
Since $Y_i$, and therefore $N_G[x]$, is disjoint from $M'$ it follows that 
$|V(P_j)| \geq t$.
Since $P_j$ is {\em shortest} it follows that $x \notin V(P_j)$, that $x$ is a neighbor (in $G$) of precisely one endpoint of $P_j$ (and no other vertices of~$P_j$), and that $V(P_j) \cap M^* = \emptyset$.
Thus, Lemma~\ref{lem:sillyObsAboutGstar} (point~\ref{itm:noDumbEdges}) yields that each $P_j$ induces a path in $G^*$ (and therefore in $G$) and that there are no edges in $G$ between $P_{j}$ and $P_{j'}$ for $j \neq j'$.
But then $x \cup V(P_1) \cup V(P_2) \cup V(P_3)$ induces an $S_{t_1, t_2, t_3}$ in $G$ with $t_1, t_2, t_3 \geq t$.
\end{proof}

\subsection{Presentation of the Algorithm}

%\mpil[inline]{Here $(G,G')$ changed to $(\mathcal{G},G)$. Is it intentional? I found it a bit confusing, but I can live with it.}
%\pg{i didnt want to use $G'$ notation in the algorithm, if i have time I will probably go back and change to $(G,G')$}

We give one last definition before presenting the algorithm.

\begin{definition}[level sets]
Let $G$ be graph, $G'$ and induced subgraph of $G$, ${\cal F}$ a list of vertex sets of $G$, and $N$ a positive integer. For all natural numbers $j$, the $j^{th}$ {\em level set with respect to $G$, $G'$, and ${\cal F}$}, denoted by ${\mathcal{L}_j}(G, G', {\cal F})$, is defined as the set of vertices of $G'$ that belong to at least $j$ sets (counting multiplicity) of $N_{G}^{G'}[{\cal F}]$.
\end{definition}

In the following recursive algorithm, the input will always consist of a natural number $N$, two lists ${\cal F}_1$ and ${\cal F}_2$ and a graph $G$. Additionally, there will be a global variable ${\cal G}$ which is set to the very first graph the algorithm is called on, so $G$ will always be an induced subgraph of ${\cal G}$. We will say that a vertex, $v$, in $G$ is {\em N-branchable with respect to ${\cal G}$, $G$, ${\cal F}_1$, and ${\cal F}_2$} (or more simply {\em branchable} when the values of $N$, ${\cal G}$, $G$, ${\cal F}_1$, and ${\cal F}_2$ are clear from the context) if there is an natural number $j$ such that either $|N_{G}[v] \cap {\mathcal{L}_j}({\cal G}, G, {\cal F}_1)|$ $\geq \frac{N}{2^j}$ or $|N_{G}[v] \cap {\mathcal{L}_j}({\cal G}, G, {\cal F}_2|$ $\geq \frac{N}{2^j}$.

We now present a quasi-polynomial time algorithm for independent set on $\Sttt$-free graphs which we will refer to as IND. We first give the high level ideas of how IND works, followed by a formal prose-style description of the algorithm, then we give the algorithm in pseudocode. 

\paragraph{Overview.}
At the highest level IND is a recursive algorithm that does three basic operations. When there is an $N$-branchable vertex, $v$, (for $N ~ |G|$) IND will be recursively called on $G-v$ and $G-N_G[v]$, when there is an extended strip decomposition such that no particle contains too much weight, IND will be recursively called on each particle, and when there is a balanced separator that is dominated by few vertices, it adds the balanced separator to a list (either ${\cal F}_1$ or ${\cal F}_2$). The lists (${\cal F}_1$ and ${\cal F}_2$) are what will guide the branching process, and the goal of branching is to (efficiently) reach an instance where the input graph has a desirable extended strip decomposition. Both the ``extended strip decomposition'' and ``add a balanced separator'' operations will come in two distinct flavors. The ``extended
strip decomposition” operation will be either a {\em type 1 extended strip decomposition operation}, with reference to the $N$-good extended strip decomposition of Observation \ref{obs:bbs or good esd},  or a {\em type 2 extended
strip decomposition operation}, with reference to the rigid extended strip decomposition of Lemma \ref{lem: balanced sep or esd}. The ``balanced separator” operation will be either a {\em boosted balanced separator operation}, with reference balanced separators of Observation \ref{obs:bbs or good esd} (which will in fact be $s$-boosted balanced separator), or (simply) a {\em balanced separator operation}, with reference to the balanced separator of Lemma \ref{lem: balanced sep or esd}. 

%\pg{This intuition should be added somewhere...: Either we can keep on finding balances seps to use for branching efficiently to reach a graph that has a good extended strip decomposition, or if there are no bal seps then we immediately have a good extended strip decomposition.}

Let us now be a little more detailed about how IND works.
The algorithm is a recursive algorithm that takes as input an $\Sttt$-free graph $G$, a vertex set $X$ ($X$ may also be set to $\bot$, indicating that a new set for $X$ must be found), an integer $N$, and two lists ${\cal F}_1$ and ${\cal F}_2$. If we wish to know the weight of a maximum weight independent set of the graph $G$, then IND is intended to be initially called on the inputs $G, X = \bot, N = |G|, {\cal F}_1 = \emptyset, {\cal F}_2 = \emptyset$. The algorithm sets a global variable ${\cal G}$ which is set to the first graph that the algorithm is called on, so that in all recursive calls, ${\cal G}$ refers to the initial graph the algorithm is called on. In any given call of IND, the vertex set $X$ along with the vertex sets contained in ${\cal F}_1$ and ${\cal F}_2$ may not be subsets of $V(G)$, but they will always be subsets of $V({\cal G})$. The integer $N$ will be approximately equal to $|G|$ (and will always satisfy $|G| \leq N$), and is used for the sake of making the run time analysis easier. Among other things, this integer is used to determine when a vertex is branchable (i.e. when it is $N$-branchable). 

The set $X$ is obtained using \esdthm$(G)$ (see Definition~\ref{def:esd bs generation}) and will thus have the property that no particle of the corresponding extended strip decomposition of $G-N_G[X]$ will have more than $|G|/2 \leq N/2$ vertices and $|X| \leq c_t\log(N)$. One goal of the branching operation, ``type 2 extended strip decomposition'', and ``type 2 balanced separator'' operation are to efficiently reduce $\rel_{\cal G}(G, X, N)$ to the empty set. Observation~\ref{obs:bbs or good esd} tells us that when $\rel_{\cal G}(G, X, N) = \emptyset$ that either we will get an extended strip decomposition that is $N$-good (in which case we make a lot of progress as each particle now has much less than $N \approx |G|$ vertices, this is the ``type 1 extended strip decomposition'' operation) or we get that $N_{\cal G}^G[X]$ is a $c_t\log(N)$-boosted balanced separator (with $X$ as a core). In either case we find a new $X$ using Theorem~\ref{thm:ICALP:weight} and repeat the process. 

But how do we make progress in the case where $N_{\cal G}^G[X]$ is a $c_t\log(N)$-boosted balanced separator (and we do not have an $N$-good extended strip decomposition)? When $N_{\cal G}^G[X]$ is a $c_t\log(N)$-boosted balanced separator, we place $X$ into ${\cal F}_1$ (this is the ``type 1 balanced separator operation''). An analysis similar to that found in \cite{GartlandL20} (sketched in Section~\ref{subsec:branching}) shows that, because the size of $X$ is at most $c_t\log(N))$, we can collect these cores of $c_t\log(N)$-boosted balanced separators into the list ${\cal F}_1$ and efficiently branch in a manner such that no vertex of $N_{\cal G}^G[{\cal F}_1]$ belongs to over $\log(N)$ of these sets, and therefore by Lemma~\ref{lem:cant pack bbs}, in an $\Sttt$-free graph, ${\cal F}_1$ cannot contain more than $80tc_t\log^2(N)$ sets. It follows that when we repeat the process from the previous paragraph, we can only get back that $N_{\cal G}^G[X]$ is $c_t\log(N)$-boosted balanced separator only a few times (at most $80tc_t\log^2(N)$ times) before get an extended strip decomposition that is $N$-good (or no component of $G$ has many vertices, but by Observation~\ref{obs:no large componet give good esd} this implies the existence of an $N$-good extended strip decomposition), and we make good progress.

Next, let us briefly look at how the algorithm is able to efficiently  reduce $\rel_{\cal G}(G, X, N)$ to the empty set using the branching, ``type 2 extended strip decomposition'', and ``type 2 balanced separator'' operations. The basic idea is based on a combination of Lemmas~\ref{lem: balanced sep or esd} and \ref{lem:cant pack strong bs} and the techniques used in \cite{GartlandL20} (sketched in Section~\ref{subsec:branching}). IND applies Lemma~\ref{lem: balanced sep or esd} to $G$ and $\rel_{\cal G}(G, X, N)$ (with $i = \log(200c_t^3\log^3(N))$). If an extended strip decomposition, $(H, \eta)$, is returned then since each particle, $P$, has much less than $|\rel_{\cal G}(G, X, N)|$ vertices of $\rel_{\cal G}(G, X, N)$ (at most $(1-\frac{1}{800c_t^3\log^3(N)})|\rel_{\cal G}(G, X, N)|$) and because $\rel_{\cal G}(P, X, N) \subseteq \rel_{\cal G}(G, X, N)$, it follows that $|\rel_{\cal G}(P, X, N)| << |\rel_{\cal G}(G, X, N)|$ and good progress is made in reducing the size of $\rel_{\cal G}(P, X, N)$ (this is the type 2 extended strip decomposition operation). Otherwise the lemma returns a set $C$ such that $N_G[C]$ is an $\frac{|\rel_{\cal G}(G, X, N)|}{200c_t^3\log^3(N)}$-balanced separator of ($G$, $|\rel_{\cal G}(G, X, N)|$) and $|C| \leq 28000c_t^4\log^4(N)$.

So how do we make progress in efficiently reducing $|\rel_{\cal G}(G, X, N)|$ when what we get back is a set $C$ such that $N_G[C]$ is an $\frac{|\rel_{\cal G}(G, X, N)|}{200c_t^3\log^3(N)}$-balanced separator of $G$ for $|\rel_{\cal G}(G, X, N)|$? An analysis similar to that found in \cite{GartlandL20} shows that, because $|C| \leq 28000c_t^4\log^4(N)$, we can collect these sets $C$ that we find into the list ${\cal F}_2$  (this is the ``type 2 balanced separator operation'') and efficiently branch in a manner such that no vertex of $N_{\cal G}^G[{\cal F}_2]$ belongs to over $\log(N)$ of these sets. Since ${\cal G}$ is $\Sttt$-free (and $|X| \leq c_t\log(N)$), it follows from Lemma~\ref{lem:cant pack strong bs} that $|{\cal F}_2|$ cannot grow larger than $10t\log(N)$. Hence, after applying Lemma~\ref{lem: balanced sep or esd} a few times (at most $10t\log(N)$ times), it must return an extended strip decomposition and good progress is made in decreasing $\rel_{\cal G}(G, X, N)$.

\paragraph{Formal description.}
We now give a formal description of our independent set algorithm for $\Sttt$-free graphs, which we will refer to as IND. The algorithm is a recursive algorithm that takes as input a graph $G$, a vertex set $X$ ($X$ may also be set to $\bot$), an integer $N$, and two lists ${\cal F}_1$ and~${\cal F}_2$. IND is intended to be initially called on the inputs $G, X = \bot, N = |G|, {\cal F}_1 = \emptyset, {\cal F}_2 = \emptyset$. The algorithm sets a global variable ${\cal G}$ which is set to the graph in the first set of input parameters, $G, X, N, {\cal F}_1, {\cal F}_2$, that the algorithm is called on so that on all recursive calls ${\cal G}$ refers to the initial graph the algorithm is called on. The vertex set $X$ along with the vertex sets contained in ${\cal F}_1$ and ${\cal F}_2$ may not be subsets of $V(G)$, although they will always be subsets of $V({\cal G})$. $N$ will be approximately, but always greater than or equal to, the size of $G$. The point of $N$ is to help in the runtime analysis. Among other things, $N$ is used to determine when a vertex is branchable (i.e. when it is $N$-branchable).

When the algorithm makes a recursive call, some of the elements among the input parameters, $G, X, N, {\cal F}_1, {\cal F}_2$, will be the same in the recursive call as they are in the current instance, while the remaining parameters will be changed. In the following description of the algorithm, when describing the input to the recursive calls, we will only explicitly mention the parameters that are changed from the current call, unmentioned parameters are assumed to remain the same as in the current call. For instance if the graph that the recursive call is made on is different from the graph of the current call but all other elements, $X, N, {\cal F}_1, {\cal F}_2$ remain the same as in the current call of the algorithm, we will only indicate what the new graph is the call is made on and not mention the unchanged elements, $X, N, {\cal F}_1, {\cal F}_2$.

In order to help the runtime analysis of the algorithm, we will label each call of IND that is made based on the first case it satisfies (which in turn determines the recursive calls it will make).

\begin{enumerate}

    \item[Base Case:] For the Base Case (Label: base case call), if $|V(G)| \leq 1$ then IND returns $\w(V(G))$.

    \item[Case 1:] For Case 1 (Label: branch call), if there exists a branchable vertex $v \in G$, then IND is recursively called on two instances, the first instance on $G-v$ and the second on $G-N_G[v]$, and stores the numbers returned by these recursive calls as $I_f$ and $I_s$ respectively. The algorithm then returns the maximum of $I_f$ and $I_s$ + $\w(v)$.

    \item[$\bullet$] If the algorithm has not returned at this point and $X$ is equal to $\bot$, then the algorithm sets $X$~=~\esdthm$(G)$. No recursive call is made here, no label is given here, and the algorithm continues to see which case it satisfies. We say that the set $X$ is {\em discovered} in this call.

    %\pg{move case 2 to case 1 ans case 3 to case 2. if bot then if bot then we should recurse on components |G| is low, then case 3 (now 2) should never be true if bot and case 2 (now 1) is false. So we can make }
    \item[Case 2:] For Case 2 (Label: type 1 extended strip decomposition call) if the extended strip decomposition inferred by $(X,G)$, call it $(H, \eta)$, is an $N$-good extended strip decomposition, then for each particle $P$ of $(H, \eta)$ the algorithm recursively calls itself on $G = P$, $X = \bot$, $N = |P|$, ${\cal F}_1 = \emptyset, {\cal F}_2 = \emptyset$. Then IND returns \matching$(H, \eta)$.

    \item[Case 3:] For Case 3 (Label: boosted balanced separator call), if $N^G_{\cal G}[X]$ is an $\frac{N}{32c_t^2\log^2(N)}$-balanced separator for $G$ then IND is recursively called with $X$ added to ${\cal F}_1$, $X$ set to $\bot$, and ${\cal F}_2$ set to $\emptyset$. Then the algorithm returns the value obtained from this recursive call. Here we say that $X$ is the boosted balanced separator core added in this call.

    \item[$\bullet$] If IND has not returned at this point then note by Observation~\ref{obs:bbs or good esd} and the fact that Case 2 and 3 do not hold implies $\rel_{\cal G}(G, X, N) \neq \emptyset$. The algorithm then applies Lemma~\ref{lem: balanced sep or esd} (with $i = \log(200c_t^3\log^3(N))$) to either obtain, in polynomial time, a rigid extended strip decomposition of $G$, call it $(H, \eta)$, such no particle of $(H, \eta)$ has over $(1-\frac{1}{800c_t^3\log^3(N)})|\rel_{\cal G}(G, X, N)|$ vertices of $\rel_{\cal G}(G, X, N)$ or a $C \subseteq V(G)$ such that $N_G[C]$ is a $\frac{|\rel_{\cal G}(G, X, N)|}{200c_t^3\log^3(N)}$-balanced separator for $\rel(G, X, N)$, and $|C| \leq 28000c_t^4\log^4(|G|) \leq 28000c_t^4\log^4(N)$. No recursive call is made here, no label is given, and the algorithm continues to see which case it satisfies.

    \item[Case 4:] For Case 4 (Label: type 2 extended strip decomposition call), if Lemma~\ref{lem: balanced sep or esd} returned an extended strip decomposition, $(H, \eta)$, then for each particle, $P$, of $(H, \eta)$ IND is recursively called with the graph set to $P$ and ${\cal F}_2$ set to the empty set. IND then returns \matching$(H, \eta)$.

    \item[Case 5:] For Case 5 (Label: balanced separator call), if Lemma~\ref{lem: balanced sep or esd} returned a balanced separator, $N_G[C]$, for $\rel_{\cal G}(G, X, N)$ then IND is recursively called, adding $C$ to ${\cal F}_2$. IND then returns the value obtained from this recursive call. Here we say the set $C$ is the balanced separator core added in this call.

\end{enumerate}

For completeness and ease of reference we give pseudocode for the algorithm IND below. The correctness proofs and running time analysis do not refer to the pseudocode, and so a reader may choose to skip it. 
%, readers though should read the previous description of the algorithm though as it contains definitions and ``labels'' that are not stated in the following pseudocode. 
%
Recall that IND sets a global variable ${\cal G}$ which is set to the graph in the first set of input parameters, $G, X, N, {\cal F}_1, {\cal F}_2$. This step is not explicitly mentioned in the pseudocode. 

\medskip
\noindent \textbf{IND}
\begin{algorithmic}[1]
\STATE \textbf{Input:} $G$, $X$, $N$, ${\cal F}_1$, ${\cal F}_2$
\STATE \textbf{Output:} {\sf mwis}$(G)$.

\IF{$|V(G)| \leq 1$}\label{line:base}
\RETURN $\w(V(G))$
\ENDIF

\IF{exists branchable vertex, $v$,}\label{line:branch}
\RETURN $\max\left(\mbox{IND}(G-v, X, N, {\cal F}_1, {\cal F}_2), \mbox{IND}(G-N_G[v], X, N, {\cal F}_1, {\cal F}_2) + \w(v)\right)$
\ENDIF

\IF{$X$ = $\bot$}
\STATE Set $X$ = \esdthm$(G)$
\ENDIF

\IF{the extended strip decomposition, $(H, \eta)$, inferred by $(X, G)$ is $N$-good}\label{line:good esd}\label{line:type1 esd}
\FORALL{particles $P$ in $(H, \eta)$}
\STATE Get $\mbox{IND}(P, \bot, |P|, \emptyset, \emptyset)$
\ENDFOR
\RETURN \matching$(H, \eta)$
\ENDIF

\IF{$N^G_{\cal G}[X]$ is an $\frac{N}{32c_t^2\log^2(N)}$-balanced separator for $G$}\label{line:bbs}
\RETURN $\mbox{IND}(G, \bot, N, {\cal F}_1 \cup X, \emptyset)$ 
\ENDIF 

\STATE Use Lemma~\ref{lem: balanced sep or esd} with $i = \log(200c_t^3\log^3(N))$ to obtain either a rigid extended strip decomposition, $(H, \eta)$, 
%that is $N$-good for $\rel_{\cal G}(G, X)$ 
such that no particle of $(H, \eta)$ contains over $(1-\frac{1}{800c_t^3\log^3(N)})|\rel_{\cal G}(G, X, N)|$ vertices of $\rel_{\cal G}(G, X, N)$
or a set $C \subseteq V(G)$ such that $N_G[C]$ is a $\frac{|\rel_{\cal G}(G, X,N)|}{200c_t^3\log^3(N)}$-balanced separator for $\rel_{\cal G}(G, X,N)$, and $|C| \leq 28000c_t^4\log^4(|G|) \leq 28000c_t^4\log^4(N)$. %\pg{need to change the lemma to guarantee we get an esd if C can be empty, also need to change this in the prose version of the alg}

\IF{Lemma~\ref{lem: balanced sep or esd} returns $(H, \eta)$}
\FORALL{particles $P$ in $(H,\eta)$}\label{line:type2 esd}
\STATE Get $\mbox{IND}(P, X, N, {\cal F}_1, \emptyset)$
\ENDFOR
\RETURN \matching$(H, \eta)$
\ENDIF

\IF{Lemma~\ref{lem: balanced sep or esd} returns $C$}\label{line:bs}
\RETURN IND$(G, X, N, {\cal F}_1, {\cal F}_2 \cup C)$
\ENDIF

\end{algorithmic}
\medskip

%If the call of IND satisfies the condition of line~\ref{line:base} then it is gets a label of base, if it satisfies the condition of line~\ref{line:branch} then it is gets a label of branch,

%If the algorithm executes the code in line~\ref{line:relavent set is empty 2} then this implies the condition of line~\ref{line:relavent set is empty} must be true and the condition of line~\ref{line:good esd} is false. Hence, by Lemma~\ref{obs:bbs or good esd} we can conclude that $X$ is the core of a boosted balanced separator.

%Similarly, in line \ref{line:good esd 2}, it is not immediately obvious that in this case there must exists an $(H, \eta)$ that is a good extended strip decomposition for some $M' \subseteq \rel(G, X)$ with $|M| \geq (1-1/?)|\rel(G, X)|$. That this must happen follows from the fact that the condition of line~\ref{line:strong balanced separator} is false, that is, $G$ does not have a $\frac{|\rel(G, X)|}{60|X|^2}$-balanced separator for $\rel(G, X)$ with a core $B$, $|B| \leq ?$ along with Lemma~\ref{lem:weak imples strong}.

\subsection{Correctness and Runtime Analysis}

In order to analyze the runtime of the algorithm, we will find it useful to define the recursion tree generated by a run of the algorithm and prove that it has only a quasi-polynomial number of vertices. Because we have yet to prove the algorithm will terminate, the tree in the next definition may be infinite, but we will 
shortly 
%eventually \pg{make sure this is the right word}
prove that IND will terminate that the recursion tree for IND is finite.

%When the current call of IND, call it $P$, makes a recursive call to IND, call it $C$, we say that $P$ {\em invokes} the call of $C$.

\begin{definition}[recursion tree]\label{def:recursion tree}
The {\em recursion tree}, $T$, {\em generated} by IND($G, \bot, |G|, \emptyset, \emptyset$) is the directed rooted tree with a node for each call of IND made in the course of running IND on the initial input ($G, \bot, |G|, \emptyset, \emptyset$), the root node corresponding to the initial call of IND on the input ($G, \bot, |G|, \emptyset, \emptyset$). There is a directed edge from a node $p \in T$ to a node $c \in T$ when the call that corresponds to $p$ invoked the call that corresponds to $c$. Furthermore, we label the vertices $p$ and $c$ as well as the edge $pc$ as follows. The vertices $p$ and $c$ get the same label as the calls they correspond to respectively (replacing ``call'' now with ``node''). If the call that corresponds to $p$ is labeled with anything other than branch call, then the $pc$ edge gets same label as the call that corresponds to $p$ (replacing ``call'' now with ``edge''). If $p$ corresponds to a branch call then let $v$ be the vertex that is branched on in that call and let $G_p$ be the graph given in the input of that call. If $c$ corresponds to the call where the graph $G_p - N_{G_p}[v]$ is used as the input then we label $pc$ as a ``success edge'', and if $c$ corresponds to the call where the graph $G_p - v$ is used as the input we label $pc$ as a ``failure edge''.

%We add two more definitions. First Let $P$ be a path in $T$ beginning at the root vertex and ending at at leaf. We call an induced subpath of $P$ a {\em decedent path} of $T$. Second, 
Furthermore, let $u$ be a node of $T$. Then we let $(G_u, X_u, N_u,$ ${\cal F}_{1,u}, {\cal F}_{2,u})$ denote the tuple that was used for the input of the call $u$ corresponds to. We call this tuple the {\em parameters of $u$}.
\end{definition}

Let $G$ be a graph. We collect a set of observations about IND and the recursion tree generated by IND$(G, \bot, |G|, \emptyset, \emptyset)$ that follow directly from how IND has been defined. We state these observations without a proof (as their proofs follow directly from how the algorithm was defined) and we will use them use in future proofs typically without reference to this observation.

\begin{observation}\label{obs:obvious observations}
Let $G$ be a graph and let $T$ be the recursion tree generated by IND$(G, \bot, |G|, \emptyset, \emptyset)$. Let $p$ and $c$ be nodes of $T$ such that $pc$ is an edge of $T$ and let $(G_p, X_p, N_p,$ ${\cal F}_{1,p}, {\cal F}_{2,p})$ and $(G_c, X_c, N_c,$ ${\cal F}_{1,c}, {\cal F}_{2,c})$ be the parameters of $p$ and $c$ respectively. Then the following hold:

\begin{enumerate}
    \item  $N_p < N_c$ if $pc$ is a type 1 extended strip decomposition edge and $N_p = N_c$ otherwise. Additionally, $N_p, N_c \geq |G|$.

    \item $G_c$ is an induced subgraph of $G_p$, $G_c$ is a proper induced subgraph of $G_p$ if $pc$ is a success, failure, or type 2 extended strip decomposition, and $G_p$ = $G_c$ if $pc$ is a balanced separator or boosted balanced separator edge.

    \item If $X_p = \bot$ and $p$ is not a base case node nor a branch node, then there is a set $X$ discovered in the call that corresponds to $p$.

    \item \label{itm:empty rel set} Assume $p$ is not a base case node nor a branch node. If $X_p \neq \bot$ then let $X = X_p$, else let $X$ be the set that is discovered in the call that corresponds to $p$. Then if $\rel_G(G_p,X, N_p) = \emptyset$ then (using Observation~\ref{obs:bbs or good esd}) $p$ is either a type 1 extended strip decomposition node or a boosted balanced separator node.

    %\item Assume $p$ is not a base case node nor branch node. Then if 
    
    \item If $c$ is a base case node then $c$ is a leaf of $T$.

\end{enumerate}
\end{observation}

The next three lemmas show that IND$(G, \bot, |G|, \emptyset, \emptyset)$ terminates and returns the weight of a maximum weight independent set of $G$.

\begin{lemma}\label{lem:none empty}
Let $t$ be a positive integer, $G$ an $\Sttt$-free graph, $T$ the recursion tree generated by IND$(G, \bot, |G|, \emptyset, \emptyset)$, $u \in T$ such that $u$ is a balanced separator or boosted balanced separator node, and let $(G_u, X_u, N_u, {\cal F}_{1,u}, {\cal F}_{2,u})$ be the parameters of $u$. If $C$ is the balanced separator or boosted balanced separator core added in the call that corresponds to $u$, then $N_{G}^{G_u}[C] \neq \emptyset$.
\end{lemma}

\begin{proof}
Let $t$, $G$, $T$, $u$, $(G_u, X_u, N_u, {\cal F}_{1,u}, {\cal F}_{2,u})$, and $C$ be as in the statement of this lemma, and assume for a contradiction that $N_{G}^{G_u}[C] = \emptyset$.

First assume that $u$ is a boosted balanced separator node, so either $C = X_u$ or $X_u = \bot$ and $C$ was discovered in the call that corresponds to $u$. In either case, since $u$ is a boosted balanced separator node, $N_{G}^{G_u}[C]$ is an $\frac{N_u}{32c_t^2\log^2(N_u)}$-balanced separator for $G_u$ and by assumption $N_{G}^{G_u}[C] = \emptyset$, it follows that no component of $G_u$ contains over $\frac{N_u}{32c_t^2\log^2(N_u)}$ vertices. It follows by Observation~\ref{obs:no large componet give good esd} that the extended strip decomposition inferred by $(C, G_u)$ is $N_u$-good and hence $u$ should have been a type 1 extended strip decomposition node and not a boosted balanced separator node.

Next, assume that $u$ is a balanced separator node, so the empty set is a $\frac{|\rel_G(G_u, X,N_u)|}{200c_t^3\log^3(N_u)}$-balanced separator of $(G_u, \rel_G(G_u, X, N_u)$). If $X_u \neq \bot$ then let $X = X_u$, and if $X_u = \bot$ then let $X$ be the set discovered in the call that corresponds to $u$. Since $u$ is not a type 1 extended strip decomposition node nor a boosted balanced separator node it follows from the \ref{itm:empty rel set} point of Observation~\ref{obs:obvious observations} (or Observation~\ref{obs:bbs or good esd}) that $\rel_G(G_u, X, N_u) \neq \emptyset$. Since $N_{G}^{G_u}[C] = \emptyset$ it follows that no component of $G_u$ contains over $\frac{|\rel_G(G_u, X,N_u)|}{200c_t^3\log^3(N_u)}$ vertices of $\rel_G(G_u, X,N_u)$. Among all components of $G_u-N_G^{G_u}[X]$ that has over $\frac{N_u}{32c_t^2\log^2(N_u)}$ vertices (of which there are at most $32c_t^2\log^2(N_u)$), let $B$ be one such that $|N_{G_u}[B] \cap \rel_G(G_u, X, N_u)|$ is maximized, so since every vertex of $\rel_G(G_u, X, N_u)$ has a neighbor in at least on component of $G_u-N_G^{G_u}[X]$ that has over $\frac{N_u}{32c_t^2\log^2(N_u)}$, it follows that  $|N_{G_u}[B] \cap \rel_G(G_u, X, N_u)| > \frac{|\rel_G(G_u, X, N_u)|}{32c_t^2\log^2(N_u)}$. But then $N_{G_u}[B]$ is a connected set that contains at least $\frac{|\rel_G(G_u, X, N_u)|}{32c_t^2\log^2(N_u)}$ vertices of $\rel_G(G_u, X, N_u)$ and therefore the empty set cannot be a $\frac{|\rel_G(G_u, X,N_u)|}{200c_t^3\log^3(N_u)}$-balanced separator of $(G_u, \rel_G(G_u, X, N_u)$).
\end{proof}

\begin{lemma}\label{lem:bounded depth}
Let $t$ be a positive integer and, $G$ an $\Sttt$-free $n$-vertex graph. Then IND$(G, \bot, |G|, \emptyset, \emptyset)$ terminated and the recursion tree generated by IND$(G, \bot, |G|, \emptyset, \emptyset)$ is finite. 
\end{lemma}

\begin{proof}
Let $t$, $G$, and $n$ be as in the statement of this lemma, $T$ the recursion tree of IND$(G, \bot, |G|, \emptyset, \emptyset)$, and $P$ some path in $T$. In order to show that IND$(G, \bot, |G|, \emptyset, \emptyset)$ terminates and that $T$ is finite, it is sufficient to show that there is a bounded number of each of the 6 types of edges that can appear in~$P$, type 1 and type 2 extended strip decomposition edges, balanced separator and boosted balanced separator edges, and success/failure edges. Let $pc$ be an edge of $T$, and let $(G_p, X_p, N_p, {\cal F}_{1,p}, {\cal F}_{2,p})$ and $(G_c, X_c, N_c, {\cal F}_{1,c}, {\cal F}_{2,c})$ be the parameters of $p$ and $c$ respectively. If $pc$ is a success edge, failure, edge, or type 2 extended strip decomposition edge, then $G_c$ is a proper subgraph of $G_p$, hence there can be at most $n$ of each type of these edges. If $pc$ is a type 1 extended strip decomposition edge, then $N_c < N_p$, so again there can be at most $n$ type 1 extended strip decomposition edges.

Now let $u$ and $w$ be two nodes of $P$ with parameters $(G_u, X_u, N_u, {\cal F}_{1,u}, {\cal F}_{2,u})$ and $(G_w, X_w,$ $N_w, {\cal F}_{1,w}, {\cal F}_{2,w})$ respectively and let $P'$ be the subpath of $P$ that starts at $u$ and ends at $w$. Assume that $P'$ does not contain any type 1 or type 2 extended strip decomposition edges nor success/failure edges (so all edges are balanced separator and boosted balanced separator edges). It follows that $G_u = G_w$, $N_u = N_w = N$, and if $P'$ has $N\log(N)$ boosted balanced separator edges, then ${\cal F}_{1,w}$ must have at least $N\log(N)$ sets $S \in {\cal F}_{1,w}$ (counting multiplicity) such that $N_{G}^{G_w}[S] \neq \emptyset$ (using Lemma~\ref{lem:none empty} to ensure they are nonempty). Hence ${\cal L}_{\log(N)}(G, G_w, {\cal F}_{1,w}, N)$ is none empty and since any vertex $v \in {\cal L}_{\log(N)}(G, G_w, {\cal F}_{1,w}, N)$ is by definition a branchable vertex, $w$ must be a branch node (or a base case node). It follows that $P'$ can have at most $N\log(N)$ boosted balanced separator edges. Since $P$ has a bounded number of type 1 or type 2 extended strip decomposition edges and success/failure edges, there must be a bounded number of boosted balanced separator edges as well.

Now additionally assume that $P'$ contains no boosted balanced separator edges as well, so all edges of $P'$ are balanced separator edges. So, if $P'$ has $N\log(N)$ balanced separator edges, then ${\cal F}_{2,w}$ must have at least $N\log(N)$ sets $S \in {\cal F}_{2,w}$ (counting multiplicity) such that $N_{G}^{G_w}[S] \neq \emptyset$ (using Lemma~\ref{lem:none empty} to ensure they are nonempty). Hence ${\cal L}_{\log(N)}(G, G_w, {\cal F}_{2,w}, N)$ is nonempty and since any vertex $v \in {\cal L}_{\log(N)}(G, G_w, {\cal F}_{2,w}, N)$ is by definition a branchable vertex, $w$ must be a branch node (or a base case node). It follows that $P'$ can have at most $N\log(N)$ balanced separator edges. Since $P$ has a bounded number of all other edge types, there must be a bounded number of balanced separator edges as well, completing the proof.
\end{proof}

Now that we have established the recursion tree is finite, we can prove that IND returns the correct answer.

\begin{lemma}\label{lem: correct return}
Let $G$ be a  graph. Then IND$(G, \bot, |G|, \emptyset, \emptyset)$ returns the weight of a maximum weight independent set of $G$.
\end{lemma}

\begin{proof}
Let $G$ be a graph, $T$ the recursion tree generated by IND$(G, \bot, |G|, \emptyset, \emptyset)$, and $u$ a node of $T$ and assume that for all children, $v$, of $u$, that the call corresponding to $v$ correctly returns the weight of a maximum weight independent set of $G_v$, where $G_v$ is the graph used as input for the call corresponding to $v$. We show that the call corresponding to $u$ then correctly returns the maximum weight independent set of $G_u$ where $G_u$ is the graph used as input for the call corresponding to $u$. 

If $u$ is balanced separator or boosted balanced separator node then $G_u = G_v$ and so $u$ returns the weight of a maximum weight independent set of $G_v = G_u$. If $u$ is a branch node which branches on the vertex $v \in G_u$, then if we set $I_f$ and $I_s$ to be the weight of a maximum weight independent set of $G-v$ and $G-N_{G_u}[v]$ respectively, then $u$ returns the maximum of $I_f$ and $I_s$ + $\w(v)$ which is the weight of a maximum weight independent set of $G_u$. Lastly, if $u$ is a type 1 or type 2 extended strip decomposition node, then $u$ returns \matching$(H, \eta)$, which by Lemma~\ref{lem:particle matching} is the weight of a maximum weight independent set of $G_u$.

It now follows from induction that IND$(G, \bot, |G|, \emptyset, \emptyset)$ returns the weight of a maximum weight independent set of $G$.
\end{proof}

Next, we observe that the degree of recursion trees is at most polynomial.

\begin{observation}\label{obs:bounded degree}
There exists a constant $c$ such that for any positive integer $t$ and $\Sttt$-free $n$-vertex graph $G$, the recursion tree, $T$, of IND$(G, \bot, |G|, \emptyset, \emptyset)$ has a maximum degree of most $n^{c}$.
\end{observation}

\begin{proof}
Let $t$, $G$, and $T$ be as in the statement of the lemma, and let $u \in T$ with parameters $(G_u, X_u, N_u, {\cal F}_{1,u}, {\cal F}_{2,u})$. If $u$ is any type of node other than a type 1 or type 2 extended strip decomposition node, then it is clear the degree of $u$ is at most 2. If $u$ is a type 1 extended strip decomposition node then (by the discussion immediately after Definition~\ref{def:esd bs generation}) the extended strip decomposition inferred by $(X_u, G_u)$ has $n^{\Oh(1)}$ particles, hence $u$ has degree $n^{\Oh(1)}$. If $u$ is a type 2 extended strip decomposition, then each child of $u$ corresponds to a particle of a rigid extended strip decomposition of $G_u$, which has $\Oh(n)$ vertices and therefore $n^{\Oh(1)}$ particles, and therefore $u$ has degree $n^{\Oh(1)}$. 
\end{proof}

Let $G$ be an $n$-vertex graph and $T$ the recursion tree generated by IND$(G, \bot, |G|, \emptyset, \emptyset)$. Recall that the edges of $T$ are labeled (see Definition~\ref{def:recursion tree}). The next set of lemmas show that on any root to leaf path $P$ of $T$, there are at most polylog($n$) edges with any given label other than a failure label, of which there can be at most $n$. We will then be able to use this fact to bound the runtime of IND.

\begin{lemma}\label{lem:few t1 esd edges}
Let $t$ be a positive integer, $G$ an $\Sttt$-free $n$-vertex graph, $T$ the recursion tree generated by IND$(G, \bot,$ $|G|, \emptyset, \emptyset)$, and $P$ a path of $T$. Then $P$ contains at most $64c_t^2\log^3(n)$ type 1 extended strip decomposition edges.
\end{lemma}

\begin{proof}
Let $t$, $G$, $n$, $T$, and $P$ be as in the statement of this lemma. Let $pc$ be a type 1 extended strip decomposition edge of $P$ and let $(G_p, X_p, N_p, {\cal F}_{1,p}, {\cal F}_{2,p})$ and $(G_c, X_c, N_c, {\cal F}_{1,c}, {\cal F}_{2,c})$ be the parameters of $p$ and $c$ respectively. Then $N_c \leq (1-\frac{1}{32c_t^2\log^2(N_p)})N_p \leq (1-\frac{1}{32c_t^2\log^2(n)})N_p$. 

Now, let $u$ and $w$ be two vertices of $P$ with parameters $(G_u, X_u, N_u, {\cal F}_{1,u}, {\cal F}_{2,u})$ and $(G_w, X_w,$ $N_w, {\cal F}_{1,w}, {\cal F}_{2,w})$ respectively. Since for all $c \geq 2$, $(1-1/c)^{2c} \leq 1/2$ it follows that if there are $64c_t^2\log^3(n)$ type 1 extended strip decomposition edges on the subpath of $P$ that starts at $u$ and ends at $w$ then $N_w \leq (1-\frac{1}{32c_t^2\log^2(n)})^{64c_t^2\log^3(n)}N_u$ $\leq (1-1/2)^{\log(n)}N_u$ $\leq 1$. It follows that $|G_w| \leq 1$ and therefore $w$ must be a base node and hence the last vertex of $P$. So, $P$ cannot have over $64c_t^2\log^3(N)$ type 1 extended strip decomposition edges.
\end{proof}

%\begin{lemma}
%Let $G$ be an $n$-vertex $\Sttt$-free graph, $T$ be the recursion tree induced by a run of IND$(G, \bot, |G|, \emptyset, \emptyset)$, and $P$ be a root to leaf path of $T$. Then $P$ contains at most $\log(n)$ extended strip decomposition Nodes.
%\end{lemma}

\begin{lemma}\label{lem:fractional packing}
Let $t$ be a positive integer, $G$ an $\Sttt$-free graph, $T$ the recursion generated by IND$(G, \bot, |G|, \emptyset, \emptyset)$, and $u$ a node of $T$ with parameters $(G_u, X_u, N_u, {\cal F}_{1,u}, {\cal F}_{2,u})$. Then no vertex of $G_u$ belongs to over $\log(N_u)$ sets of $N_G^{G_u}[{\cal F}_{1,u}]$ and no vertex of $G_u$ belongs to over $\log(N_u)$ sets of $N_G^{G_u}[{\cal F}_{2,u}]$, hence ${\cal L}_i(G, G_p, {\cal F}_{1,u}, N_u)$ = ${\cal L}_i(G, G_p, {\cal F}_{2,u}, N_u)$ = $\emptyset$ for $i > \log(N_u)$.
\end{lemma}

\begin{proof}
Let $t$, $G$, $T$, $u$, and $(G_u, X_u, N_u, {\cal F}_{1,u}, {\cal F}_{2,u})$ be as in the statement of this lemma. Assume for a contradiction that the statement of this lemma does not hold for $u$ and that $u$ is the first node on a path from the root to $u$ such that the statement of this lemma does not hold. It follows there is some vertex, call it $v$, of $G_u$ belongs to over $\log(N_u)$ sets of $N_G^{G_u}[{\cal F}_{1,u}]$ or $\log(N_u)$ sets of $N_G^{G_u}[{\cal F}_{2,u}]$. Let $w$ be the parent of $u$ in $T$ and let $(G_w, X_w, N_w, {\cal F}_{1,w}, {\cal F}_{2,w})$ be the parameters of $w$. If the edge $wu$ was a type 1 extended strip decomposition edge, then ${\cal F}_{1,u}$ =  ${\cal F}_{2,u}$ = $\emptyset$ but then $v$ could not belong to over $\log(N_u)$ sets of $N_G[{\cal F}_{1,u}]$, so $wu$ is not a type 1 extended strip decomposition edge, hence $N_w = N_u$. 

As the arguments are nearly identical, with out loss of generality assume $v \in G_u$ belongs to over $\log(N_u)$ sets of $N_G[{\cal F}_{1,u}]$. Since (by how $u$ was chosen) $v$ does not belong to over $\log(N_w) = \log(N_u)$ sets of $N_G[{\cal F}_{1,w}]$ it follows that the edge $wu$ must be a boosted balanced separator edge (hence $G_u = G_w$) and $v$ must belong to $\log(N_u) = \log(N_w)$ sets of $N_G^{G_w}[{\cal F}_{1,w}]$ in order for $v$ to belong to over $\log(N_u)$ sets of $N_G^{G_u}[{\cal F}_{1,u}]$. But then by definition of being a branchable vertex, since $v$ belongs to $\log(N_w)$ sets of $N_G^{G_w}[{\cal F}_{1,w}]$ the call that corresponds to $w$ should have branched on $v$ (or some other branchable vertex) and the edge $wu$ would therefore be a either a success or failure edge, which is a contradiction. Therefore, there can never be a ``first node'' in $T$ such that the statement of this lemma does not hold.
\end{proof}

\begin{lemma}\label{lem:f1 contains bbs}
Let $t$ be a positive integer, $G$ an $\Sttt$-free $n$-vertex graph, $T$ the recursion tree generated by IND$(G, \bot,$ $|G|, \emptyset, \emptyset)$, $u$ a node of $T$ with parameters $(G_u, X_u, N_u, {\cal F}_{1,u}, {\cal F}_{2,u})$, and $C$ the largest component of $G_u$. Then either $|C| < N_u/2$ or all sets of ${\cal F}_{1,u}$ are cores of $c_t\log(N_u)$-boosted balanced separators of $G_u$ originating in $G$.
\end{lemma}

\begin{proof}
Let $t$, $G$, $T$, $u$, $(G_u, X_u, N_u, {\cal F}_{1,u}, {\cal F}_{1,u})$, and $C$ be as in the statement of this lemma. Let $S$ be in ${\cal F}_{1,u}$ and assume that $|C| \geq N_u/2$, we show that $S$ is a core of a $c_t\log(N_u)$-boosted balanced separator for $G_u$ originating in $G$. 

Let $w$ be the closest ancestor of $u$ in $T$ that corresponds to a call where $S$ is the core of a boosted balanced separator added in that call. Let $(G_w, X_w, N_w, {\cal F}_{1,w}, {\cal F}_{2,w})$ be the parameters of $w$. Note that this implies $X_w = S$ or $X_w = \bot$ and $S$ was discovered in the call that corresponds to $w$. Additionally, let $a$ be the closest ancestor of $w$ that corresponds to a call where $S$ was discovered. Let $(G_a, X_a, N_a, {\cal F}_{1,a}, {\cal F}_{2,a})$ be the parameters of $a$, as $S$ was discovered in the call the corresponds to $a$, it follows that $X_a = \bot$.

No edge on the path of $T$ starting from $w$ and ending at $u$ can be a type 1 extended strip decomposition edge or else $S$ would not be in ${\cal F}_{1,u}$ (since type 1 extended strip decompositions ``resets ${\cal F}_{1}$'' to $\emptyset$), and no edge on the path of $T$ starting from $a$ and ending at $w$ can be a type 1 extended strip decomposition edge (since type 1 extended strip decompositions ``resets $X$'' to $\bot$). Hence $N_a = N_{w} = N_u$, so set $N$ = $N_a = N_{w} = N_u$. So, $|S| \leq c_t\log(N)$ and since $N_{G}^{G_{w}}[S]$ is an $\frac{N}{32c_t^2\log^2(N)}$-balanced separator for $G_{w}$ and $G_u$ is an induced subgraph of $G_{w}$, it follows that $N_{G}^{G_{u}}[S]$ is an $\frac{N}{32c_t^2\log^2(N)}$-balanced separator for $G_{u}$. Since by assumption $|C| \geq N/2$, we have that $N_{G}^{G_{u}}[S]$ is an $\frac{|C|}{16c_t^2\log^2(N)}$-balanced separator for $G_{u}$ and therefore $S$ is a core of a $c_t\log(N)$-boosted balanced separator for $G_u$ originating in $G$.
\end{proof}

\begin{lemma}\label{lem:few bbs edges}
Let $t$ be a positive integer, $G$ an $n$-vertex $\Sttt$-free graph, $T$ the recursion tree generated by IND$(G, \bot,$ $|G|, \emptyset, \emptyset)$, and $P$ a  path of $T$. Then $P$ contains at most $5200c_t^4\log^5(n)$ boosted balanced separator edges.

%such that no edges of $P$ are type 1 extended strip decomposition edges. Then $P$ contains at most $\log(n)$ boosted balanced separator edges.
\end{lemma}

\begin{proof}
Let $t$, $G$, $T$, and $P$ be as in the statement if this lemma. Let $P'$ be a subpath of $P$ that does not contain any edges that are type 1 extended strip decomposition edges. We first show that $P'$ has less than $80tc_t\log^2(n)$ boosted balanced separator edges. 

Assume for a contradiction that $P'$ does have $80tc_t\log^2(n)$ boosted balanced separator edges, let $uw$ be the $80tc_t\log^2(n))^{th}$ boosted balanced separator edge, and let $(G_u, X_u, N_u, {\cal F}_{1,u}, {\cal F}_{2,u})$ and $(G_w, X_w, N_w,$ ${\cal F}_{1,w}, {\cal F}_{2,w})$ be the parameters of $u$ and $w$ respectively. 
Since $P'$ has no type 1 extended strip decomposition edges and $80tc_t\log^2(n)$ boosted balanced separator edges it follows that $|{\cal F}_{1,w}|$ $\geq$ $80tc_t\log^2(n)$ $\geq$ $80tc_t\log^2(N_w)$ and $N_u = N_w$.

Since $uw$ is a boosted balanced separator edge and not a type 1 extended strip decomposition edge we can conclude that if $C$ is the largest component of $G_u = G_w$, then $|C| \geq N_u/2 = N_w/2$ (or else by Observation~\ref{obs:no large componet give good esd} the extended strip decomposition inferred by ($X_u, G_u)$ would be $N_u$-good and $uw$ would be a type 1 extended strip decomposition edge). It follows then from Lemma~\ref{lem:f1 contains bbs} that all $80tc_t\log^2(n) \geq 80tc_t\log^2(N_w)$ sets of ${\cal F}_{1,w}$ are cores of $c_t\log(N_w)$-boosted balanced separators for $G_w$ originating in $G$. By Lemma~\ref{lem:fractional packing} no vertex of $G_w$ belongs to over $\log(N_w)$ vertex sets of $N_{G}^{G_w}[{\cal F}_{1,w}]$, it then follows from Lemma~\ref{lem:cant pack bbs} that $G$ contains an $\Sttt$, a contradiction.

Hence, $P'$ has less than $80tc_t\log^2(n)$ boosted balanced separator edges. Since by Lemma~\ref{lem:few t1 esd edges} $P$ has at most $64c_t^2\log^3(n)$ type 1 extended strip decomposition edges, it follows that $P$ contains at most $(80tc_t\log^2(n))\cdot (64c_t^2\log^3(n)+1)$ $\leq$ $5200tc_t^3\log^5(n)$ $\leq$ $5200c_t^4\log^5(n)$ (recall by definition that $c_t \geq t$) boosted balanced separator edges.
\end{proof}

%\begin{lemma}
%Let $G$ be an $n$-vertex $\Sttt$-free graph, $T$ the recursion tree generated by IND$(G, \bot, |G|, \emptyset, \emptyset)$, and $P$ a path of $T$ such that no edges of $P$ are type 1 extended strip decomposition nor Boosted balanced separator edges. Assume that $P$ contains at least $\log(n)$ balanced separator edges. Then there exists nodes $A$ and $B$ on $P$ with parameters $(G_A, X_A, N_A, {\cal F}_{1,A}, {\cal F}_{2,A})$ and $(G_B, X_B, N_B, {\cal F}_{1,B}, {\cal F}_{2,B})$ respectively such $|G_A| \geq |G_B|/2$ and the subpath of $P$ with endpoints $A$ and $B$ contains at least $\log(n)$ Balanced Separator Nodes.
%\end{lemma}

\begin{lemma}\label{lem:few t2 esd edges}
Let $t$ be a positive integer, $G$ an $n$-vertex $\Sttt$-free graph, $T$ the recursion tree generated by IND$(G, \bot, |G|,$ $\emptyset, \emptyset)$, and $P$ a path of $T$. Then $P$ contains at most $10^7c_t^7\log^9(n)$ type 2 extended strip decomposition edges.
\end{lemma}

\begin{proof}
Let $t$, $G$, $T$, and $P$ be as in the statement of this lemma. Let $P'$ be a subpath of $P$ that contains no type 1 extended strip decomposition nor boosted balanced separator edges. It then follows that there exists an integer $N$ and a set $X \subseteq V(G)$ such that for any node $a \in P'$ with parameters $(G_a, X_a, N_a, {\cal F}_{1,a}, {\cal F}_{2,a})$ it holds that $N_a = N$ and $X_a$ = $\bot$ or $X$.
%, furthermore, since $P'$ is maximal, if $a$ is the first node of $P'$ then $X_a = \bot$ and there is some node $b \in P'$ such that $X$ is discovered in the call that is associated with $b$.

We first show that $P'$ has at most $1600c_t^3\log^4(n)+2$ type 2 extended strip decomposition edges. Let $pc$ be a type 2 extended strip decomposition edge of $P'$ and let $(G_p, X_p, N, {\cal F}_{1,p}, {\cal F}_{2,p})$ and $(G_c, X_c, N, {\cal F}_{1,c}, {\cal F}_{2,c})$ be the parameters of $p$ and $c$ respectively. Either $X_p = X_c = X$ or $X_p = \bot$, $X_c = X$ and $X$ was discovered in the call $p$ corresponds to. In either case we have $|\rel_G(G_c, X, N)| \leq (1-\frac{1}{800c_t^3\log^3(N)})|\rel_G(G_p, X, N)| \leq (1-\frac{1}{800c_t^3\log^3(n)})|\rel_G(G_p, X, N)|$. 

Now, let $u$ and $w$ be two vertices of $P'$ with parameters $(G_u, X_u, N, {\cal F}_{1,u}, {\cal F}_{2,u})$ and $(G_w, X_w,$ $N, {\cal F}_{1,w}, {\cal F}_{2,w})$ respectively. Since $(1-1/c)^{2c} \leq 1/2$ for $c \geq 2$, it follows that if there are $1600c_t^3\log^4(n)+2$ type 2 extended strip decomposition edges on the subpath of $P'$ that starts at $u$ and ends at $w$ then \begin{align*}|\rel_G(G_w, X,N)| & \leq (1-\frac{1}{800c_t^3\log^3(n)})^{1600c_t^3\log^4(n)+2}|\rel_G(G_u, X,N)|\\ & \leq (1-1/2)^{\log(n)+1}|\rel_G(G_u, X,N)|< 1.\end{align*} It follows that $|\rel_G(G_w, X,N)| = 0$ and therefore, for any node $z$ in $P'$ (with parameters say $(G_z, X,$ $N, {\cal F}_{1,z}, {\cal F}_{2,z})$) that comes after $w$ it holds that $\rel_G(G_z, X,N) = \emptyset$.
%By Observation~\ref{obs:bbs or good esd} either $N_G^{G_z}[X]$ is an $\frac{N}{32c_t^2\log^2(N)}$-balanced separator for $G_z$ or the extended strip decomposition inferred by $(X, G_z)$ is $N$-good. 
It then follows from the \ref{itm:empty rel set}$^{th}$ point of Observation~\ref{obs:obvious observations} that there cannot be any other type 2 extended strip decomposition edges after $w$ in $P'$.

Since $P'$ has at most $1600c_t^3\log^4(n)+2$ type 2 extended strip decomposition edges and by Lemmas~\ref{lem:few t1 esd edges} and \ref{lem:few bbs edges}, $P$ has at most $64c_t^2\log^3(n)$ type 1 extended strip decomposition edges and at most $5200c_t^4\log(n)^5$ boosted balanced separator edges, it follows that $P$ has at most $(1600c_t^3\log^4(n)+2)$ $(64c_t^2\log^3(n) + 5200c_t^4\log^5(n) + 1)$ $\leq$ $10^7c_t^7\log^9(n)$ type 2 extended strip decomposition edges. 
\end{proof}

\begin{lemma}\label{lem:few bs edges}
Let $t$ be a positive integer, $G$ an $n$-vertex $\Sttt$-free graph, $T$ the recursion tree generated by IND$(G, \bot,$ $|G|, \emptyset, \emptyset)$, and $P$ a  path of $T$ such that no edges of $P$ are type 1 extended strip decomposition, boosted balanced separator, nor type 2 extended strip decomposition edges. Then $P$ contains at most $10^9c_t^8\log^{11}(n)$ balanced separator edges.
\end{lemma}

\begin{proof}
%Let $t$, $G$, $T$, and $P$ be as in the statement of this lemma. Let $P'$ be a maximal subpath of $P$ that contains no type 1 extended strip decomposition, boosted balanced separator, nor type 2 extended strip decomposition edges. Let $P''$ be a maximal subpath of $P$ that contains $P'$ and contains no type 1 extended strip decomposition nor boosted balanced separator edges. It then follows that there exists an integer $N$ and a set $X \subseteq V(G)$ such that for any node $a \in P''$ with parameters $(G_a, X_a, N_a, {\cal F}_{1,a}, {\cal F}_{2,a})$ it holds that $N_a = N$ and $X_a$ = $\bot$ or $X$, furthermore, if $a$ is the first node of $P''$ then $X_a = \bot$ and there is some node $b \in P''$ such that $X$ is discovered in the call that is associated with $b$.
Let $t$, $G$, $T$, and $P$ be as in the statement of this lemma. Let $P'$ be a subpath of $P$ that contains no type 1 extended strip decomposition, boosted balanced separator, nor type 2 extended strip decomposition edges. It then follows that there exists an integer $N$ and a set $X \subseteq V(G)$ such that for any node $a \in P'$ with parameters $(G_a, X_a, N_a, {\cal F}_{1,a}, {\cal F}_{2,a})$ it holds that $N_a = N$ and $X_a$ = $\bot$ or $X$.

Assume for a contradiction that $P'$ has $10t\log^2(n)$ balanced separator edge. Let $u_iw_i$ denote the $(i10t\log(n))^{th}$ balanced separator edge, $u_0$ be the first vertex of $P'$, $(G_{u_i}, X_{u_i}, N, {\cal F}_{1,u_i}, {\cal F}_{2,u_i})$ be the parameters of $u_i$, $i \geq 0$, of $P'$,  and $P^j$ denote the subpath of $P'$ that starts at $u_{j}$ and ends at $u_{j+1}$, $j \geq 0$.

For some $P^j$ it must hold that $|\rel_G(G_{j}, X,N)|/2 \leq |\rel_G(G_{j+1}, X,N)|$ or else \[\frac{|\rel_G(G_0, X,N)|}{2^{\log(n)}} > |\rel_G(G_{\log(n)}, X,N)|\] which implies that $|\rel_G(G_{\log(n)}, X,N)| = 0$. Therefore, by 
%Observation~\ref{obs:bbs or good esd} either $N_G^{G_{\log(n)}}[X]$ is an $\frac{N}{10|X|}$-balanced separator of $V(G_{\log(n)})$ or the extended strip decomposition inferred by $(X, G_{\log(n)})$ is $N$-good. It follows that there cannot be any other balanced separator edge after $u_{\log(n)}$ in $P'$, but that contradicts that $u_{\log(n)}w_{\log(n)}$ is a balanced separator edge.
the \ref{itm:empty rel set}$^{th}$ point of Observation~\ref{obs:obvious observations} it follows that $u_{\log(n)}w_{\log(n)}$ cannot be a balanced separator edge.
So, we conclude that for some $P^j$ it must hold that $|\rel_G(G_{j}, X,N)|/2 \leq |\rel_G(G_{j+1}, X,N)|$, so fix this index as $j$.

%Then it follows there must exists nodes $u_{i-1}$ and $w$ of $P'$ with parameters $(G_u, X_u, N_u, {\cal F}_{1,u}, {\cal F}_{2,u})$ and $(G_w, X_w, N_w,$ ${\cal F}_{1,w}, {\cal F}_{2,w})$ respectively such that the subpath of $P'$, call it $P''$, that has $u$ as its first vertex and $w$ last has at least $\log^{?-1}(n)$ balanced separator edges and $|\rel(G_u, X)|/2 \leq |\rel(G_w, X)|$. 
For each balanced separator node, $v$, on the path $P^j$, let $B_v$ denote the balanced separator core added in call $v$, ${\cal F}$ the list of these $B_v$'s, and $(G_v, X_v, N,$ ${\cal F}_{1,v}, {\cal F}_{2,v})$ the parameters of $v$. Note that $X = X_v$ (accept for possibly the first such $v$ if $X_v = \bot$ and then $X$ would be $X$ discovered in the call the corresponds to $v$) and $|\rel_G(G_v, X,N)|/2 \leq |\rel_G(G_{u_j}, X,N)|/2 \leq |\rel_G(G_{u_{j+1}}, X,N)|$ and $\rel_G(G_{u_{j+1}}, X,N) \subseteq \rel_G(G_v, X,N) \subseteq \rel_G(G_{u_j}, X,N)$. So, since $N_G^{G_v}[B_v]$ is an $\frac{|\rel_G(G_v, X,N)|}{200c_t^3\log^3(N)}$-balanced separator for $\rel_{G}(G_v, X,N)$, it follows that $N_G^{G_{u_{j+1}}}[B_v]$ is an $\frac{|\rel_G(G_{u_{j+1}}, X,N)|}{100c_t^3\log^3(N)}$-balanced separator for $\rel_{G}(G_{u_{j+1}}, X, N)$. By definition of ${\cal F}$, it follows that $|{\cal F}| \geq 10t\log(n)$ $\geq$ $10t\log(N)$, by Lemma~\ref{lem:fractional packing} no vertex of $G_{u_{j+1}}$ belongs to over $\log(N)$ set of $N_{G}^{G_{u_{j+1}}}[{\cal F}_{2,u_{j+1}}]$ and therefore of $N_{G}^{G_{u_{j+1}}}[{\cal F}]$ (as ${\cal F}$ $\subseteq$ ${\cal F}_{2,u_{j+1}}$), so by Lemma~\ref{lem:cant pack strong bs} $G$ contains and $\Sttt$, a contradiction.

We conclude that $P'$ has less that $10t\log^2(n)$ $\leq$ $10c_t\log^2(n)$ (recall by definition $c_t \geq t$) balanced separators. Since by Lemmas \ref{lem:few t1 esd edges}, \ref{lem:few bbs edges}, and \ref{lem:few t2 esd edges} $P$ has at most $64c_t^2\log^3(n)$, $5200c_t^4\log^5(n)$, and $10^7c_t^7\log^9(n)$ type 1 extend strip decomposition edges, boosted balanced separator edges, and type 2 extended strip decomposition edges respectively, it follows that $P$ has at most $10c_t\log^2(n)$ ($64c_t^2\log^3(n)$ + $5200c_t^4\log^5(n)$ + $10^7c_t^7\log^9(n)$ + 1) $\leq$ $10^9c_t^8\log^{11}(n)$ balanced separator edges. 
\end{proof}

Let $G$ be a graph, $T$ the recursion tree generated by IND$(G, \bot, |G|, \emptyset, \emptyset)$, and $p, c \in T$ such that $pc$ is an edge of $T$ with parameters $(G_p, X_p, N_p, {\cal F}_{1,p}, {\cal F}_{2,p})$ and $(G_c, X_c, N_c, {\cal F}_{1,c}, {\cal F}_{2,c})$ respectively. We say that a vertex $v \in G$ is {\em added to level set $i$ during the call corresponding to $p$} if $v \notin$ $({\mathcal{L}_i}(G, G_p, {\cal F}_{1,p}) \cup {\mathcal{L}_i}(G, G_p, {\cal F}_{2,p}))$ and $v \in$ $({\mathcal{L}_i}(G, G_c, {\cal F}_{1,c}) \cup {\mathcal{L}_i}(G, G_c, {\cal F}_{2,c}))$. Note that for a vertex $v$ to be added to level set $i$ during call $p$, the edge $pc$ must be either a balanced separator edge or a boosted balanced separator edge and $v \in ({\mathcal{L}_{i-1}}(G, G_p, {\cal F}_{1,p}) \cup {\mathcal{L}_{i-1}}(G, G_p, {\cal F}_{2,p}))$. Given a path $P$ in $T$ we say that a vertex $v$ is added to level set $i$ in path $P$ if there is at least one node $u$ such that $v$ is added to level set $i$ during the call corresponding to $u$. Note that it is possible for a vertex $v$ to be added to level set $i$ in path $P$ multiple times since the level sets ${\cal F}_1$ and ${\cal F}_2$ can get set to the empty set multiple times.

Additionally, we say that a vertex $v \in G$ is {\em removed from level set $i$ during the call corresponding to $p$} if $v \in $ $({\mathcal{L}_i}(G, G_p, {\cal F}_{1,p}) \cup {\mathcal{L}_i}(G, G_p, {\cal F}_{2,p}))$ and $v \notin$ $({\mathcal{L}_i}(G, G_c, {\cal F}_{1,c}) \cup {\mathcal{L}_i}(G, G_c, {\cal F}_{2,c}))$. We say that a vertex $v$ is added to level set $i$ in path $P$ if there is at least one node $u$ such that $v$ is removed from level set $i$ during the call corresponding to $u$. Note that it is possible for a vertex $v$ to be removed from level set $i$ in path $P$ multiple times since the level sets ${\cal F}_1$ and ${\cal F}_2$ can get set to the empty set multiple times. Furthermore, note that if the call that corresponds to the first vertex of $P$ has the property that both lists ${\cal F}_1$ and ${\cal F}_2$ are the empty set, then the number of vertices added to level set $i$ in path $P$, counting multiplicity, is at least as much as the number of vertices removed from level set $i$ in path $P$. This holds because if for a node $u \in P$ a vertex $v$ is removed from level set $i$ during the call corresponding to $u$, the vertex $v$ must have been first added to level set $i$ during the call corresponding to $w$ for some $w \in P$ that comes before $u$ in $P$.

%\begin{observation}
%Let $G$ be an $n$-vertex graph $\Sttt$-free graph, $T$ the recursion tree generated by IND$(G, \bot, |G|, \emptyset, \emptyset)$, and $P$ a path of $T$. Then the number of vertices added to level set $i$ in path $P$, counting multiplicity, at least the number of vertices removed from level set $i$ in path $P$, counting multiplicity.
%\end{observation}

\begin{lemma}\label{lem:few vertices added}
Let $t$ be a positive integer, $G$ be an $n$-vertex graph $\Sttt$-free graph, $T$ the recursion tree generated by IND$(G, \bot, |G|, \emptyset, \emptyset)$, and $P$ a path of $T$ such that no edges of $P$ are type 1 extended strip decomposition edge and $P$ contains $c$ balanced separator and boosted balanced separator edges. Then there exists a natural number $N$ such that for any vertex $v \in P$ with parameters $(G_v, X_v, N_v, {\cal F}_{1,v}, {\cal F}_{2,v})$ it holds that $N_v = N$ and at most $c\cdot 56000c_t^4\log^4(N)N/2^{i}$ vertices are added to level set $i$ in $P$.
\end{lemma}

\begin{proof}
Let $t$, $G$, $T$, $c$, and $P$ be as in the statement of this lemma. The fact that there exists a natural number $N$ such that for any vertex $v \in P$ with parameters $(G_v, X_v, N_v, {\cal F}_{1,v}, {\cal F}_{2,v})$ it holds that $N_v = N$ follows from the fact that $P$ has no type 1 extended strip decomposition edges. 

Now, let $p$ be a node of $P$ and let $(G_p, X_p, N, {\cal F}_{1,p}, {\cal F}_{2,p})$ be the parameters of $p$. Note that if there is a vertex $v$ that is added to level set $i$ during call $p$, it must be that the vertex $v$ is in ${\cal L}_{i-1}(G,G_{p}, {\cal F}_{1,p}, N)$ or ${\cal L}_{i-1}(G,G_{p}, {\cal F}_{2,p}, N)$.

First, assume that $p$ is a balanced separator node of $P$ and the set $B_p$ is the balanced separator core added in the call that corresponds to $p$, so $|B_p| \leq 28000c_t^4\log^4(N)$. Since $p$ is a balanced separator node, there is no branchable vertex in the call corresponding to $p$. Therefore each vertex of $B_p$ has less than $N/2^{i-1}$ neighbors in level set ${\cal L}_{i-1}(G,G_{p}, {\cal F}_{2,p}, N)$ and therefore at most $|B_p|N/2^{i-1}$ $\leq$ $56000c_t^4\log^4(N)N/2^{i}$ vertices are added to level set $i$ during call $p$.

Next, assume that $p$ is a boosted balanced separator node of $P$, and let $X$ be the boosted balanced separator core added in the call that corresponds to $p$ (if $X_p \neq \bot$ then $X = X_p$, and if $X_p = \bot$ then the set $X$ was discovered in the call that corresponds to $p$). Let $p'$ be the first ancestor of $p$ in $T$ where $X$ was discovered (possibly with $p' = p$) and let $(G_{p'}, X_{p'}, N_{p'}, {\cal F}_{1,p'}, {\cal F}_{2,p'})$ be the parameters of $p'$, so $X \subseteq V(G_{p'})$ and $|X| \leq c_t\log(N_{p'})$. It follows that all nodes between $p'$ and $p$ are branch nodes, type 2 extended strip decomposition nodes, and balanced separator nodes (since all recursive calls of boosted balanced separator nodes and type 1 extended strip decomposition nodes would reset $X$ to $\bot$). It follows that ${\cal F}_{1,p'} = {\cal F}_{1,p}$ and $N_{p'} = N$ and so $|X| \leq c_t\log(N)$. Additionally, since $X$ was discovered in the call that corresponds to $p'$, $p'$ cannot be a branch node so there is no branchable vertex in the call corresponding to $p'$, in particular, there is no vertex in $G_{p'}$ that has at least $N/2^{i-1}$ neighbors in ${\cal L}_{i-1}(G,G_{p'},{\cal L}_{1,p'}, N)$. Therefore each vertex of $X$ has at most $N/2^{i-1}$ neighbors in level set ${\cal L}_{i-1}(G,G_{p'}, {\cal F}_{1,p'}, N)$ and therefore (since $G_p$ is an induced subgraph of $G_{p'}$ and ${\cal F}_{1,p'} = {\cal F}_{1,p}$) in level set ${\cal L}_{i-1}(G,G_{p}, {\cal F}_{1,p}, N)$. Hence at most $|X|N/2^{i-1}$ $\leq$ $2c_t\log(N)N/2^{i}$ vertices are added to level set $i$ during call $p$.

In either case, we conclude that at most $56000c_t^4\log^4(N)N/2^{i}$ vertices are added to level set $i$ in the call $p$. Since by assumption there are $c$ boosted balanced separator and balanced separator edges in $P$, it follows that at most $c\cdot 56000c_t^4\log^4(N)N/2^{i}$ vertices are added to level set $i$ in $P$. 
\end{proof}

\begin{lemma}\label{lem:few success edges}
Let $t$ be a positive integer, $G$ an $\Sttt$-free $n$-vertex graph, $T$ the recursion tree generated by IND$(G, \bot, |G|, \emptyset, \emptyset)$, and $P$ a path of $T$. Then $P$ contains at most $10^{14}c_t^{12}\log^{16}(n)$ success edges.
\end{lemma}

\begin{proof}
Let $t$, $G$, $T$, and $P$ be as in the statement of this lemma. Let $P'$ be a maximal subpath of $P$ that contains no type 1 extended strip decomposition edges, $c$ the number of balanced separator and boosted balanced separator edges in $P'$, $u$ the first vertex of $P'$ (note by the maximality of $P'$ the in-edge of $u$ is either a type 1 extended strip decomposition or $u$ is the root of $T$), and $(G_u, X_u, N_u, {\cal F}_{1,u}, {\cal F}_{2,u})$ the parameters of $u$. 

By Lemma~\ref{lem:few vertices added} there is a natural number $N$ 
such that for any node $w \in P$ with parameters $(G_w, X_w, N_w, {\cal F}_{1,w}, {\cal F}_{2,w})$ it holds that $N_w = N$, hence $N_u = N$, and there are at most $c\cdot 56000c_t^4\log^4(N)N/2^{i}$ vertices added to level set $i$ in $P'$ (counting multiplicity). Since the in-edge of $u$ is a type 1 extended strip decomposition edge (or $u$ is the root vertex of $T$), it follows that ${\cal F}_{1,u} = {\cal F}_{2 ,u} = \emptyset$. Hence if a vertex $v$ is removed in $P'$, the vertex $v$ must also have been added in $P'$, so at most $c56000c_t^4\log^4(N)N/2^{i}$ vertices can be removed from level set $i$ in $P'$ (counting multiplicity).

Let $pc$ be a success edge of $P'$ and let $(G_p, X_p, N, {\cal F}_{1,p}, {\cal F}_{2,p})$ and $(G_c, X_c, N, {\cal F}_{1,c}, {\cal F}_{2,c})$ be the parameters of $p$ and $c$ respectively. Then there is some vertex $v \in G_p$ such that $G_c = G_p-N_{G_p}[v]$. Furthermore there is some natural number $i$ such that $N_{G_p}[v]$ contains at least $N/2^i$ vertices in either level set ${\cal L}_{i}(G,G_{p}, {\cal F}_{1,p}, N)$ or ${\cal L}_{i}(G,G_{p}, {\cal F}_{2,p}, N)$, call this set of vertices $S$. It follows that all of the vertices of $S$ are removed from level set $i$ in call $p$. Since at most $c\cdot 56000c_t^4\log^4(N)N/2^{i}$ vertices are added to level set $i$ in $P'$ (therefore as noted in the first paragraph of this proof at most $c\cdot 56000c_t^4\log^4(N)N/2^{i}$ vertices can be removed to level set $i$ in $P'$) this can happen to level set $i$ in $P'$ at most $c\cdot 56000c_t^4\log^4(N)$ times before it is empty. Since by Lemma~\ref{lem:fractional packing} for all $j > \log(N)$ ${\cal L}_{j}(G,G_{p}, {\cal F}_{2,p}, N)$ and ${\cal L}_{j}(G,G_{p}, {\cal F}_{2,p}, N)$ are already empty, it follows that $i \leq \log(N)$. Therefore there can be at most $c\cdot 56000c_t^4\log^5(N)$ success edges before all level sets are empty. Hence there are at most $c\cdot 56000c_t^4\log^5(N)$ $\leq$ $c\cdot 56000c_t^4\log^5(n)$ success edges in $P'$.

Lastly, by Lemmas~\ref{lem:few bbs edges} and \ref{lem:few bs edges} there are at most $5200c_t^4\log^5(n)$ and $10^9c_t^8\log^{11}(n)$ boosted balanced separator and balanced separator edges respectively in $P$. It then follows that there are at most $(5200c_t^4\log(n)^5 + 10^9c_t^8\log^{11}(n) + 1) (56000c_t^4\log^5(n))$ $\leq$ $10^{14}c_t^{12}\log^{16}(n)$ success edges in $P$.
\end{proof}

\begin{lemma}\label{lem:few failure edges}
Let $t$ be a positive integer, $G$ an $\Sttt$-free $n$-vertex graph and $T$ the recursion tree generated by IND$(G, \bot, |G|, \emptyset, \emptyset)$. Then there are at most $n$ failure edges on any root to leaf path in~$T$.
\end{lemma}

\begin{proof}
Let $t$, $G$, $P$, and $T$ be as in the statement of this lemma, and let $pc$ be a failure edge of $T$. Then $|G_c| = |G_p|-1$. Hence it is impossible for $P$ to have $n+1$ failure edges.
\end{proof}

Let $G$ be an $n$-vertex $\Sttt$-free graph. In order to bound the number of nodes in the recursion tree $T$ generated by IND$(G, \bot, |G|, \emptyset, \emptyset)$ we will provide a sequence which will uniquely determine every node of $T$, then prove that there is at most $n^{\Oh(c_t^{12}\log^{16}(n))}$ such sequences. Let $u$ be a node in $T$ and let $P$ be the path from the root node to $u$. The sequence we give is just the edges (given by their labels) taken on the path from the root to the vertex $u$. Let us assume that we have been given such a sequence of edge labels, $S$, that corresponds to the sequence of edge labelings of $P$. We show how to use $S$ to reconstruct the path $P$ from the root node to $u$ (proving this sequence uniquely determines the vertex $u$). Assume we are currently at a vertex $w$, if $w$ is a branch node, then the next label in $S$ must either be a success or failure edge, and whichever one it is uniquely determines the next node in our path to $u$. Similarly, if $w$ is a balanced separator or boosted balanced separator node, then the next label in $S$ must be a balanced separator edge or boosted balanced separator edge and $w$ has exactly one child, so again the next node in the path is uniquely determined. If $w$ is a type 1 or type 2 extended strip decomposition node though, there exists some constant $c$ (by Observation~\ref{obs:bounded degree}, independent of the choice of $G$) so that $w$ can have up to $n^c$ children, and all edges going to these children have the same label, hence the next node in the path is not uniquely determined. To fix this issue we define an {\em enriched recursion tree} to be a recursion tree $T$ such that for every type 1 or type 2 extended strip decomposition node, each of its out edge labels are additionally given a unique number $1-n^c$. It follows that with this enriching, the next label of $S$ will uniquely determine a child of $w$ and then this sequence uniquely determines the vertex $u$.

\begin{lemma}\label{lem:bounded tree}
Let $t$ be a positive integer, $G$ an $\Sttt$-free $n$-vertex graph, and $T$ the recursion tree generated by IND$(G, \bot, |G|, \emptyset, \emptyset)$. Then $T$ has at most $n^{\Oh(c_t^{12}\log^{16}(n))}$ nodes.
\end{lemma}

\begin{proof}
Let $t$, $G$ and $T$ be as in the statement of this lemma, and let $c$ be the constant from Observation~\ref{obs:bounded degree}. Furthermore, assume that $T$ is an enriched recursion tree, that is for every $u \in T$ that is a type 1 or type 2 extended strip decomposition node, the labels of the out edges of $u$ are given an additional unique integer between $1$ and $n^c$, where $c$ is the constant given in Observation~\ref{obs:bounded degree}.

We consider the set of all edge label sequences which contain at most $64c_t^2\log^3(n)$ type 1 extended strip decomposition edges (by Lemma~\ref{lem:few t1 esd edges} that is as many as $T$ can have), $5200c_t^4\log^5(n)$ boosted balanced separator edges (by Lemma~\ref{lem:few bbs edges} that is as many as $T$ can have), $10^7c_t^7\log^9(n)$ type 2 extended strip decomposition edges (by Lemma~\ref{lem:few t2 esd edges} that is as many as $T$ can have), $10^9c_t^8\log^{11}(n)$ balanced separator edges (by Lemma~\ref{lem:few bs edges} that is as many as $T$ can have), $10^{14}c_t^{12}\log^{16}(n)$ (by Lemma~\ref{lem:few success edges} that is as many as $T$ can have) success edges, and $n$ failure edges (by Lemma~\ref{lem:few failure edges} that is as many as $T$ can have). To bound the number of such sequenced, first note that since there are six types of edges and none can appear over $n^{\Oh(1)}$ times there are at most $n^{\Oh(1)}$ choices for the number of each type of edge, call these choices  $n_1, n_2, n_3, n_4, n_5,$ and $n_6$ and let $n_6$ be the number which denotes the number of failure edges. The number of possible sequences with these number choices is then ${n_1 + n_2 + n_3 + n_4 + n_5 + n_6 \choose n_1, n_2, n_3, n_4, n_5}$ $\leq$ $(n_1 + n_2 + n_3 + n_4 + n_5 + n_6)^{n_1 + n_2 + n_3 + n_4 + n_5}$. Since $n_6 \leq n$ and for $i < 6$, $n_i \leq 10^{14}c_t^{12}\log^{16}(n)$ this number is at most $n^{\Oh(c_t^{12}\log^{16}(n))}$. Since there are at most $n^{\Oh(1)}$ difference choices for values of the $n_i$'s, it follows there are at most $n^{\Oh(1)}n^{\Oh(c_t^{12}\log^{16}(n))}$ = $n^{\Oh(c_t^{12}\log^{16}(n))}$ sequences of this type.

Next we consider the ``enriched'' version of these sequences, that is, for each type 1 and type 2 extended strip decomposition edge in a sequence of $S$ we give it some number between $1$ and $n^c$. as there are at most $64c_t^2\log^3(n)$ and $10^7c_t^7\log^9(n)$ type 1 and type 2 extended strip decomposition edges respectively there are at most  $n^{c(64c_t^2\log^3(n) + 10^7c_t^7\log^9(n))}n^{\Oh(c_t^{12}\log^{16}(n))}$ = $n^{\Oh(c_t^{12}\log^{16}(n))}$ of these enriched sequences. It follows now that for any $u \in T$, the sequences of edges in the path from the root node to $u$ is contained in $S$, and by the discussion just before the statement of this lemma, for each $u \in T$ is edge sequence is unique. So, since $|S|$ = $n^{\Oh(c_t^{12}\log^{16}(n))}$ it follows that $T$ has at most $n^{\Oh(c_t^{12}\log^{16}(n))}$ nodes.
\end{proof}

We are now ready to prove Theorem~\ref{thm:long claw free qp time}.

\begin{proof}[Proof of Theorem~\ref{thm:long claw free qp time}]
Let $t$ be a positive integer, $G$ an $n$-vertex $\Sttt$-free graph, and $T$ the recursion tree generated by IND($G$, $\bot$, $|G|, \emptyset, \emptyset)$. By Lemma~\ref{lem: correct return} IND($G$, $\bot$, $|G|, \emptyset, \emptyset)$ returns the weight of a maximum weight independent set of $G$ and by Lemma~\ref{lem:bounded tree} $|T|$ = $n^{\Oh(c_t^{12}\log^{16}(n))}$. All that must be verified then is for each $u \in T$ the amount of time spend in the call that correspond to $u$ runs in polynomial time. We justify this runtime by discussing only the runtime of the steps of IND that do not clearly run polynomial time.

That the step of finding $X = \esdthm(G)$ runs in polynomial time was justified in Definition~\ref{def:esd bs generation}. That $\matching(H, \eta)$ runs in polynomial time follows from Lemma~\ref{lem:particle matching} and the facts that in type 2 extended strip decompositions calls, $(H, \eta)$ is rigid, and therefore $|H| = \Oh(n)$, and in type 1 extended strip decomposition calls $(H, \eta)$ is the extended strip decomposition inferred by some $(X', G')$ and so $|H| = n^{\Oh(1)}$ (see the discussion after Definition~\ref{def:esd bs generation}), and by Lemma~\ref{lem:particle matching} $\matching(H, \eta)$ runs in time polynomial in $|G|$ and $|H|$. Lastly, that the second bullet point of IND, where Lemma~\ref{lem: balanced sep or esd} is applied runs in polynomial time is stated in Lemma~\ref{lem: balanced sep or esd}.
\end{proof}

\section{Conclusion}
Let us point out some possible directions for future research.
First, on the structural side, we believe that \cref{thm:esd:intro} could be improved so that in the second outcome the balanced separator is dominated by a constant (depending on $t$) number of vertices.
The only reason why the current statement has the logarithmic bound is that in \cref{thm:ICALP:weight} the number of deleted neighborhoods is logarithmic.
\cite{ICALP-qptas} conjectured that \cref{thm:ICALP:weight} can actually be improved so that the number of deleted neighborhoods is constant.
Proving this conjecture would immediately yield an improved version of our \cref{thm:esd:intro}.
However, such a stronger version, while being more elegant, would not give any essentially new algorithmic result: the running time of our algorithms would still be quasi-polynomial (though a bit faster).

On the algorithmic side, an obvious natural problem is to provide a polynomial-time algorithm for MWIS in $S_{t,t,t}$-free graphs, for all $t$.
While we believe that extended strip decompositions are the right tool to use towards this goal, it seems that decompositions like the ones obtained by \cref{thm:esd:intro} would not lead to such a statement. This is because  recursing into a polynomial number of multiplicatively smaller particles inherently leads to a quasi-polynomial running time. We believe the ultimate goal would be to build an extended strip decomposition where each particle induces a graph from some ``simple'' class. In particular, so that we can solve MWIS for each particle in polynomial time without using recursion. Such decompositions for the simplest case, i.e., claw-free graphs, are provided by a deep structural result of Chudnovsky and Seymour~\cite{DBLP:journals/jct/ChudnovskyS08e}.

An important milestone on the way towards obtaining a polynomial-time algorithm for MWIS in $S_{t,t,t}$-free graphs is to solve the case of $P_t$-free graphs, which is already a very ambitious goal.

\blind{
\paragraph{Acknowledgments.}
We would like to thank Amálka Masaříková and Jana Masaříková for useful discussions while working on this paper. Additionally, we acknowledge the welcoming and productive atmosphere at Dagstuhl Seminar 22481 “Vertex Partitioning in Graphs: From Structure to Algorithms,” where a crucial part of the work leading to the results in this paper was done.
}

\bibliographystyle{alphaurl}
\bibliography{refs}

\end{document}